\documentclass[acmsmall]{acmart}

\usepackage{ifthen}

\newcommand{\version}{full} %

\usepackage{algorithm}
\usepackage[noend]{algpseudocode}
\algrenewcommand{\algorithmiccomment}[1]{\textcolor{gray}{\# #1}}%
\algrenewcommand\algorithmicdo{\textbf{\texttt{:}}}%
\algrenewcommand\algorithmicthen{\textbf{\texttt{:}}}%
\algrenewcommand\algorithmicfunction{\textbf{\texttt{def}}}%
\algrenewcommand\algorithmicindent{1.5ex}%

\usepackage{varwidth}
\usepackage{wrapfig}

\usepackage{amsmath}

\usepackage{amssymb}
\usepackage{array}

\usepackage{stmaryrd}

\usepackage{caption}
\usepackage{subcaption}
\usepackage{listings}
\usepackage{mathpartir} %

\usepackage{xcolor}

\usepackage{xspace} %
\usepackage{paralist} %

\usepackage{tikz}
\usepackage{pgf}
\usepackage{pgfplots, pgfplotstable}
\usetikzlibrary{patterns}

\newcounter{groupcount}
\pgfplotsset{
    draw group line/.style n args={5}{
        after end axis/.append code={
            \setcounter{groupcount}{0}
            \pgfplotstableforeachcolumnelement{#1}\of\datatable\as\cell{%
                \def\temp{#2}
                \ifx\temp\cell
                    \ifnum\thegroupcount=0
                        \stepcounter{groupcount}
                        \pgfplotstablegetelem{\pgfplotstablerow}{X}\of\datatable
                        \coordinate [yshift=#4] (startgroup) at (axis cs:\pgfplotsretval,0);
                    \else
                        \pgfplotstablegetelem{\pgfplotstablerow}{X}\of\datatable
                        \coordinate [yshift=#4] (endgroup) at (axis cs:\pgfplotsretval,0);
                    \fi
                \else
                    \ifnum\thegroupcount=1
                        \setcounter{groupcount}{0}
                        \draw [
                            shorten >=-#5,
                            shorten <=-#5
                        ] (startgroup) -- node [anchor=base, yshift=-2ex] {#3} (endgroup);
                    \fi
                \fi
            }
            \ifnum\thegroupcount=1
                        \setcounter{groupcount}{0}
                        \draw [
                            shorten >=-#5,
                            shorten <=-#5
                        ] (startgroup) -- node [anchor=base, yshift=-2ex] {#3} (endgroup);
            \fi
        }
    }
}

\usepackage{multirow}

\usepackage{circledsteps}
\tikzset{/csteps/outer color=blue}

\usetikzlibrary{positioning}
\usetikzlibrary{trees}
\usetikzlibrary{calc}
\usetikzlibrary{tikzmark}
\usetikzlibrary{matrix}
\usepgfplotslibrary{statistics} %

\ifthenelse{\equal{\version}{conf}}{
  \usepackage[appendix=strip,bibliography=common]{apxproof}
}{
  \usepackage[appendix=append,bibliography=common]{apxproof}
}

\usepackage{xpatch}
\usepackage{textcase}
\makeatletter
\xpatchcmd{\@sect}{\uppercase}{\MakeTextUppercase}{}{}
\xpatchcmd{\@sect}{\uppercase}{\MakeTextUppercase}{}{}
\makeatother

\theoremstyle{plain}
\newtheoremrep{theorem}{Theorem}[section]
\newtheoremrep{lemma}[theorem]{Lemma}
\newtheoremrep{proposition}[theorem]{Proposition}
\newtheoremrep{claim}[theorem]{Claim}
\newtheoremrep{corollary}[theorem]{Corollary}
\newtheoremrep{observation}[theorem]{Observation}
\newtheoremrep{remark}[theorem]{Remark}

\ifthenelse{\equal{\version}{conf}}{
  \newcommand{\inConfVersion}[1]{#1} %
  \newcommand{\inFullVersion}[1]{}   %
}{
  \ifthenelse{\equal{\version}{full}}{
    \newcommand{\inConfVersion}[1]{} %
    \newcommand{\inFullVersion}[1]{#1} %
  }{
    \newcommand{\inConfVersion}[1]{#1} %
    \newcommand{\inFullVersion}[1]{#1} %
  }
}

\newcommand{\epiql}{EpiQL\xspace}

\newcommand{\card}[1]{\ensuremath{|#1|}}
\newcommand{\db}{\textit{db}}
\DeclareMathOperator{\bigo}{\mathcal{O}}

\newcommand{\insize}{\ensuremath{\size{\db}}}
\newcommand{\samplesize}{\ensuremath{k}}

\newcommand{\idxnaive}{\textsc{Bern}\xspace} %
\newcommand{\bbinom}{\textsc{Binom}\xspace} %
\newcommand{\drawgeom}{\textsc{DrawGeo}\xspace} %
\newcommand{\geom}{\textsc{Geo}\xspace} %
\newcommand{\hybrid}{\textsc{Hybrid}\xspace} %
\newcommand{\pertuple}[1]{\textsf{PT}{#1}} %

\newcommand{\binsampling}{\textsf{M-BJ}\xspace} 
\newcommand{\syasampling}{\textsf{M-SYA}\xspace} 
\newcommand{\syasamplingC}{\textsf{M-CSYA}\xspace} 
\newcommand{\syasamplingU}{\textsf{M-USYA}\xspace} 
\newcommand{\geomsampling}{\textsf{I-{\sf \sc Geo}}\xspace}
\newcommand{\geomsamplingC}{\textsf{I$_{\mathsf{C}}$-{\sc\sf Geo}}\xspace}
\newcommand{\geomsamplingU}{\textsf{I$_{\mathsf{U}}$-{\sc\sf Geo}}\xspace}
\newcommand{\nugeomsampling}{\textsf{I-\pertuple{\geom}}\xspace}

\newcommand{\naivesampling}{\textsf{I-{\sf \sc Bern}}\xspace} %
\newcommand{\naivesamplingC}{\textsf{I$_{\mathsf{C}}$-{\sf \sc Bern}}\xspace} %
\newcommand{\naivesamplingU}{\textsf{I$_{\mathsf{U}}$-{\sf \sc Bern}}\xspace} %
\newcommand{\nunaivesampling}{\textsf{I-\pertuple{\idxnaive}}\xspace} %

\newcommand{\pertuplesampling}{\textsf{I-\pertuple{\hybrid}}\xspace}
\newcommand{\pertuplesamplingC}{\textsf{I$_{\mathsf{C}}$-\pertuple{\hybrid}}\xspace}
\newcommand{\pertuplesamplingU}{\textsf{I$_{\mathsf{U}}$-\pertuple{\hybrid}\xspace}}

\newcommand{\chainedNoOpt}{\textsf{chained}}
\newcommand{\chainedOpt}{\textsf{chained-opt}}
\newcommand{\unchainedNoOpt}{\textsf{unchained}}
\newcommand{\unchainedOpt}{\textsf{unchained-opt}}

\newcommand{\ya}{\textsf{YA}\xspace}

\newcommand{\nsa}{\textsf{NSA}\xspace}
\newcommand{\twonsa}{\textsf{2NSA}\xspace}

\newcommand{\sya}{\textsf{SYA}\xspace}
\newcommand{\syafull}{Shredded Yannakakis\xspace}

\newcommand{\seq}[1]{\ensuremath{\overline{#1}}}

\newcommand{\restr}[2]{\ensuremath{{#1}[{#2}]}}

\newcommand{\x}{\ensuremath{\seq x}}
\newcommand{\y}{\ensuremath{\seq y}}
\newcommand{\z}{\ensuremath{\seq z}}

\newcommand{\job}{JOB\xspace}

\newcommand{\stats}{STATS-CEB\xspace}

\newcommand{\low}{\texttt{low}\xspace}
\newcommand{\medium}{\texttt{medium}\xspace}
\newcommand{\high}{\texttt{high}\xspace}

\newcommand{\jobtitle}{\textit{Title}\xspace}
\newcommand{\users}{\textit{Users}\xspace}

\newcommand{\multileft}{\ensuremath{\{\!\!\{}}
\newcommand{\multiright}{\ensuremath{\}\!\!\}}}
\newcommand{\bag}[1]{\multileft #1 \multiright}
\newcommand{\bagof}[2]{\bag{#1 \mid #2}}

\newcommand{\defeq}{\stackrel{\text{def}}{=}}

\newcommand{\size}[1]{\ensuremath{|#1|}}

\DeclareMathOperator{\bern}{\beta}

\newcommand{\full}[1]{\ensuremath{\hat{#1}}}

\newcommand{\mas}{M\&S\xspace}
\newcommand{\iap}{I\&P\xspace}

\newcommand{\idx}{\textit{idx}}
\newcommand{\access}{\textsc{get}}
\newcommand{\pos}{\textit{pos}}

\newcommand{\flatten}{\ensuremath{\mu^*}}
\DeclareMathOperator{\group}{\gamma}

\newcommand{\nsemijoinsymb}{\tikz[baseline=-2pt]{\draw (0,-0.6ex)--(0.5ex,0)--(0,0.6ex)--(0,-0.6ex);\draw[->, >=stealth,draw] (0.5ex,0)--(1.75ex,0);}}

\DeclareMathOperator{\nsemijoin}{\ltimes_{\nu}}

\DeclareMathOperator{\flatsub}{\textit{attr}}
\DeclareMathOperator{\rsub}{\textit{sch}}

\DeclareMathOperator{\weightval}{\textit{weight}}

\newcommand{\phys}[1]{\ensuremath{\mathtt{#1}}}

\newcommand{\store}[1]{\ensuremath{\Sigma_{#1}}}

\newcommand{\len}[1]{\ensuremath{\card{#1}}}

\newcommand{\hd}{\textsf{hd}}
\newcommand{\hol}[1]{\ensuremath{\textsf{hd}_{#1}}}
\newcommand{\holshort}[1]{\textsf{hd}#1}
\newcommand{\weight}[1]{\ensuremath{\textsf{w}_{#1}}}

\newcommand{\nxt}{\textsf{nxt}}

\newcommand{\llen}[1]{\ensuremath{\textsf{len}_{#1}}}

\DeclareMathOperator{\imod}{mod}
\DeclareMathOperator{\idiv}{div}

\newcommand{\csr}{CSR\xspace}
\newcommand{\csrs}{CSRs\xspace}

\newcommand{\usr}{USR\xspace}
\newcommand{\usrs}{USR\xspace}

\newcommand{\perm}{\textsf{perm}}
\newcommand{\start}[1]{\ensuremath{\textsf{start}_{#1}}}
\newcommand{\prefix}{\ensuremath{\textsf{pref}}}
\newcommand{\degree}{deg}

\newcommand{\contactquery}{Q_c}

\setcopyright{cc}
\setcctype{by-nc-nd}
\acmJournal{PACMMOD}
\acmYear{2026} \acmVolume{4} \acmNumber{3 (SIGMOD)} \acmArticle{224}
\acmMonth{6} \acmDOI{10.1145/3802101}

\received{October 2025}
\received[revised]{January 2026}
\received[accepted]{February 2026}



\title{Poisson Sampling over Acyclic Joins}

\author{Liese Bekkers}
\orcid{0000-0002-2725-2131}
\affiliation{
  \institution{UHasselt, Data Science Institute}
  \country{Belgium}
}
\email{liese.bekkers@uhasselt.be}

\author{Frank Neven}
\orcid{0000-0002-7143-1903}
\affiliation{
  \institution{UHasselt, Data Science Institute}
  \country{Belgium}
}
\email{frank.neven@uhasselt.be}

\author{Lorrens Pantelis}
\orcid{0009-0004-4882-2326}
\affiliation{
  \institution{UHasselt, Data Science Institute}
  \country{Belgium}
}
\email{lorrens.pantelis@student.uhasselt.be}

\author{Stijn Vansummeren}
\orcid{0000-0001-7793-9049}
\affiliation{
  \institution{UHasselt, Data Science Institute}
  \country{Belgium}
}
\email{stijn.vansummeren@uhasselt.be}

\begin{document}

\begin{abstract}
  We introduce the problem of Poisson sampling over joins: compute a sample of
  the result of a join query by conceptually performing a Bernoulli trial for
  each join tuple, using a non-uniform and tuple-specific probability. We
  propose an algorithm for Poisson sampling over acyclic joins that is nearly
  instance-optimal, running in time $\bigo(\insize + k \log \insize)$ where
  $\insize$ is the size of the input database, and $k$ is the size of the
  resulting sample. Our algorithm hinges on two building blocks: (1) The
  construction of a random-access index that allows, given a number $i$, to
  randomly access the $i$-th join tuple without fully materializing the
  (possibly large) join result; (2) The probing of this index to construct the
  result sample. We study the engineering trade-offs required to make both
  components practical, focusing on their implementation in column stores, and
  identify best-performing alternatives for both.  Our experiments show that this pair of alternatives significantly outperforms the
  repeated-Bernoulli-trial algorithm for Poisson sampling while also
  demonstrating that the random-access index by itself can be used to
  competively implement Yannakakis' acyclic join processing algorithm when no
  sampling is required. This shows that, as far a query engine design is
  concerned, it is possible to adopt a common basis for both classical acyclic
  join processing and Poisson sampling, both without regret compared to
  classical join and sampling algorithms.
\end{abstract}

\maketitle

\section{Introduction}
\label{sec:intro}

Drawing a fixed-size sample from the result of a query has numerous  interesting
applications including online
aggregation~\cite{DBLP:conf/sigmod/AcharyaGPR99a,DBLP:conf/sigmod/HaasH99,DBLP:journals/tods/LiWYZ19},
large-scale analytics~\cite{DBLP:conf/sigmod/ZhaoC0HY18} and query optimization
in general.  Formally, the sampling problem is usually phrased as follows: given
a query $Q$, input database $\db$ and desired sample size $\samplesize$, draw
$\samplesize$ tuples uniformly at random from $Q(\db)$, without replacement.  A
naive 
solution is to first materialize $Q(\db)$ in a table, and then randomly access
the table to perform the sampling. However, this naive method is inefficient
when $Q$ is a join query since it requires to compute the full join
result---which can be orders of magnitude larger than both the input database
$\db$ and the sample size $\samplesize$---and hence wastes time computing tuples
that do not contribute to the desired output.  Alternative techniques that avoid
computing the full join exist and are typically based on designing an index
structure that can be used to guide the sampling
process~\cite{DBLP:conf/sigmod/AcharyaGPR99a,DBLP:conf/sigmod/ChaudhuriMN99,DBLP:conf/sigmod/ZhaoC0HY18}. 
Specifically,
for \emph{acyclic joins} (formally defined in Section~\ref{sec:problem}) index structures
have been proposed that can be computed in $\bigo(\card{\db})$ time, which can
then be used to sample a single result tuple in $\bigo(\log \card{\db})$
time---yielding overall complexity $\bigo(\card{\db} + k \log \card{\db})$~\cite{DBLP:journals/tods/CarmeliZBCKS22,DBLP:journals/sigmod/DaiHY25}.

In this paper, instead of drawing a fixed-size (uniform) sample, 
which requires each output tuple to be sampled with the same uniform probability,
we study the
following more general problem, which we call \emph{Poisson sampling}. In
this problem, each join tuple output by $Q$ specifies the (not necessarily
uniform) probability with which it is to be included in the sample, and the
sample is to be taken by conceptually doing a Bernoulli trial for each tuple
with its specified probability. 

Poisson sampling has applications in Markov-chain based simulations~\cite{caiSimulationDatabasevaluedMarkov2013,jampaniMonteCarloDatabase2011}. In
particular, we are motivated  by our efforts of
designing \epiql, a declarative language and accompanying high-performance data
engine that focuses on simulation of discrete agent-based
infectious disease models~\cite{Willem2017IBMsReview,Hoang2019SocialContactReview}. Such models are used by epidemiologists to
predict the evolution of transmissible diseases such as measles, influenza, or
covid in various situations.

\newcommand{\personmember}{\texttt{Person}}
\newcommand{\contactprob}{\texttt{ContactProb}}
\newcommand{\contact}{\texttt{Contact}}
\newcommand{\contactprobability}{\texttt{prob}}
\newcommand{\poolid}{\texttt{pool}}
\newcommand{\ageone}{\texttt{age1}}
\newcommand{\agetwo}{\texttt{age2}}
\newcommand{\age}{\texttt{age}}
\newcommand{\personid}{\texttt{pers}}
\newcommand{\personidone}{\texttt{per1}}
\newcommand{\personidtwo}{\texttt{per2}}
\newcommand{\newlyinfected}{\texttt{Infected}}
\newcommand{\infectious}{\texttt{Infectious}}
\newcommand{\susceptible}{\texttt{Susceptible}}
\newcommand{\transmissionprobability}{\ensuremath{p_t}}

\begin{example}
  \label{ex:epiql-contacts}
  \label{ex:epiql-transmission}
  \it
Infectious disease models simulate contact events between individuals in a population in order to calculate, at each simulation timestep, which subset of the population is susceptible to infection, which subset is infectious, which has recovered from the infection, and so on. In EpiQL, the following conjunctive query rules can be used to define the transition from the $\susceptible$ to $\newlyinfected$  disease state. Other rules (not shown)  compute from the $\newlyinfected$ population the subpopulation that becomes $\infectious$, and which can hence cause further spread in later simulation steps.
{\small
  \begin{align*}
    \newlyinfected(\personidone) & \gets \susceptible(\personidone),  \infectious(\personidtwo), \\ &  \contact(\personidone, \personidtwo),  Bernoulli(\transmissionprobability)  \\
    \contact(\personidone, \personidtwo) & \gets \personmember(\personidone, \ageone, \poolid), \\ & \personmember(\personidtwo, \agetwo, \poolid),\\ &  \contactprob(\poolid, \ageone, \agetwo, \contactprobability),\\ & Bernoulli(\contactprobability) 
  \end{align*}}
The first query states that susceptible person $\personidone$ may become infected with probability $p_t$ if there was a contact event with  infectious person $\personidtwo$. Disease models simulate contact events by dividing  the population into \emph{contact pools} (e.g., family households, schools, workplaces) and modeling  contact within a pool according  to statistics derived from real-world diary studies~\cite{Hoang2021ContactFlanders,Willem2012WeatherContacts}. In line with this, the second query above uses a relation $\contactprob(\poolid, \ageone, \agetwo, \contactprobability)$ that indicates that in $\poolid$ there is a $\contactprobability$ chance that two people aged $\ageone$ and $\agetwo$ meet, and a relation $\personmember(\personid, \age, \poolid)$ listing each individual’s age and pool memberships. In particular, the $\contact$ query  joins the \personmember\xspace and \contactprob\xspace relations. From the resulting join, each tuple is independently sampled with probability \contactprobability.  
The sampled tuples are then projected onto the participating persons.
We note that both conjunctive queries are acyclic. 
\end{example}%

Similar to the case for fixed-size sampling, it is important to design
algorithms for Poisson sampling  \emph{without} first
materializing the full join result. To illustrate, in \epiql we found that, even
for a small country with a population of $\pm 10^{7}$ people using \contactprob\xspace data collected from real-world diary studies~\cite{Hoang2021ContactFlanders,Willem2012WeatherContacts}, the full join output
size of the contact query is $\pm 10^{10}$ whereas the expected sample size is only $\pm 10^{8}$, two orders of magnitude less.

In this paper, we design algorithms for Poisson
sampling over acyclic joins that avoid materialization of the full join result and which are nearly instance-optimal from an asymptotic complexity viewpoint.
Furthermore, we study the engineering trade-offs required to make these algorithms practical in column stores. Our contributions are as follows:

\smallskip
\noindent
(1) We introduce the problem of Poisson sampling over join queries, which
generalizes fixed-size sampling. 

\smallskip
\noindent
(2) For acyclic joins, we show that Poisson sampling can be solved in the same time as fixed-size sampling: $\bigo(\insize + \samplesize \log \insize)$, where $\insize$ is the size of the input database, and $\samplesize$ is the size of the resulting sample (Theorem~\ref{thm:poisson-sampling-complexity}). From an asymptotic complexity viewpoint, this is \emph{instance-optimal} up to a $\log \insize$ factor since any correct algorithm will need to read the input and produce the sample.
Conceptually, we follow an Index-and-Probe strategy:
\emph{(i)} we construct a random-access index for a join query $Q$ that avoids materializing the full result $Q(\db)$ while still allowing efficient retrieval of the $j$-th join tuple for each position $j$; \emph{(ii)} we determine the sequence of tuple positions in $Q(\db)$ that define the output sample (referred to as position sampling); 
\emph{(iii)} we generate the sample by repeatedly probing this index.

\smallskip
\noindent
(3) Because asymptotic complexity may hide crucial constant
factors, we next investigate the engineering trade-offs involved in making this
conceptual algorithm efficient in practice, focusing on its implementation in
column stores.

(a) \emph{Random-access index construction. }
For acyclic joins it is known that a random-access index can be constructed in $\bigo(\insize)$ time that allows random access to a single tuple in $\bigo(\log \insize)$ time~\cite{DBLP:journals/tods/CarmeliZBCKS22, DBLP:conf/csl/Brault-Baron12}. Note that, crucially, the query result $Q(\db)$ may be much larger than $\insize$. Constructing a random-access index is closely tied to Yannakakis' seminal
algorithm (\ya) for processing acyclic joins~\cite{yannakakisAlgorithmsAcyclicDatabase1981}, but with additional logic to
ensure $\bigo(\log \insize)$  access time. There has recently been
extensive interest in implementing YA 
 in a manner that is also competitive with traditional binary join algorithms for inputs with  no or little dangling join tuples---situations where YA has
traditionally been slower than binary
joins~\cite{DBLP:journals/pvldb/BekkersNVW25,DBLP:conf/cidr/YangZYK24,Debuking2025Zhao,Aggregate2025VLDB,Yannakakis+,hu2024treetracker}. We observe that the column store implementation approach of
Bekkers et al~\cite{DBLP:journals/pvldb/BekkersNVW25}, which is called Shredded
Yannakakis (\sya), already provides a
linear-time-constructable random-access index. This implementation is based on a \emph{chained shredded representation} (\csr) but has access time 
$\bigo(\log \insize + d)$ instead of $\bigo(\log \insize)$ where $d$ is the
largest join-degree in the input database.
To recover a $\bigo(\log \insize)$ access time, we implement 
an \emph{unchained shredded representation} (\usr) 
within the shredded Yannakakis framework
which lifts the concept of a random-access index structure as proposed by Carmeli et al~\cite{DBLP:journals/tods/CarmeliZBCKS22} to column stores.
However, we empirically find that \csr is faster to build and, suprisingly, also faster to probe which causes \csr-based sampling to outperform \usr-based sampling.

(b) \emph{Position sampling.} 
We distinguish between the uniform and non-uniform cases. For uniform sampling, we compare three strategies for generating the sequence of probe positions.
We demonstrate that these have complementary advantages depending on the specific sampling probability that is being used. Based on this observation, we design a hybrid position-sequence generation algorithm that dynamically adapts to the observed data distribution.
For the non-uniform case, we reduce the problem to a series of uniform sampling steps over groups of tuples sharing the same sampling probability, applying the hybrid method to each group. Further details are provided in Section~\ref{sec:position-sampling}.

\smallskip 
\noindent
(4) We implement all methods inside of Apache Datafusion, a high-performance main-memory-based columnar query engine written in Rust, and experimentally compare these methods based on established join query benchmarks over real-world data as well as the disease transmission use case. We find that:

(a)
The combination of \csr with our hybrid position-sampling method is most efficient across all benchmarks, even though \csr has worse asymptotic complexity than \usr.

(b)
Compared to the naive approach that first materializes the join result and then performs a Bernoulli trial per join tuple, our method is up {to $6.08$x faster}.

(c)
Finally, we observe that \csr can also be used for normal join processing (without sampling), and there we find that \csr is competitive with \usr on a wide variety of benchmarks. This illustrates that, as far as query engine design is
concerned, it suffices to adopt a \emph{single} strategy for implementing
Yannakakis without regret in a column store,\footnote{That is, in a manner that is competitive with traditional binary 
join algorithms for \emph{all} kind of inputs and not just for inputs with many dangling join tuples
for which YA is known to be efficient.} one based on \csr, as this allows for both normal join processing
and sampling.  We believe that this insight is relevant because we know of no current column store that implements the known efficient fixed-size sampling algorithms introduced in \cite{DBLP:journals/tods/CarmeliZBCKS22,DBLP:conf/sigmod/ZhaoC0HY18,DBLP:journals/sigmod/DaiHY25}, presumably because it requires intricate changes to the engine's internals. By adopting \csr we achieve two goals at once: enable efficient sampling and ensure robust acyclic join processing.

\smallskip
\noindent{\bf Organization.}
Section~\ref{sec:problem} defines the Poisson sampling problem, and Section~\ref{sec:solution-overview} outlines our 
approach and background. Sections~\ref{sec:indexing} and~\ref{sec:position-sampling} resp.\ cover indexing and probing. Section~\ref{sec:experiments} presents experiments, Section~\ref{sec:related-work} discusses related work, and Section~\ref{sec:conclusions} concludes.

\section{Problem Statement}
\label{sec:problem}

For a natural number $l > 0$, we denote the set $\{1,\dots,l\}$ by $[l]$.  We are
interested in processing \emph{Poisson sampling queries}, which are  queries of the following form:
\begin{equation}
  \label{eq:tir-query}
  Q = \bern_y \left(R_1(\x_1) \Join \dots \Join R_{l}(\x_{l})\right).
\end{equation}
Here, $l \geq 1$; each $R_i$ is a (not necessarily distinct) relation symbol;  each $\x_i$ is a set of
pairwise distinct attributes  that denotes the schema of
$R_i$, for $i \in [l]$; and $y \in \x_1 \cup \dots \cup \x_{l}$. Expressions of the form $R_i(\x_i)$ are called
\emph{atoms}. The expression $R_1(\x_1) \Join \dots \Join R_{l}(\x_{l})$ is called the \emph{(full) join query} underlying $Q$, which we denote by $\full{Q}$.

Throughout the paper, we adopt a bag-based semantics for relations and
queries. Formally, tuple $t$ over $\x_1 \cup \dots \cup \x_{l}$ occurs
in the join result of $\full{Q}(\db)$ of $\full{Q}$ on database $\db$ if for every
$i \in [l]$ the tuple $t[\x_i]$ (i.e., $t$ projected on $\x_i$), occurs with
multiplicity $m_i > 0$ in input relation $R_i$. The result multiplicity of $t$
in the full join is then $m_1 \times \dots \times m_{l}$.

The operator $\bern$ has the following semantics. If the values of the
$y$-attribute in $\full{Q}(\db)$ are in the range $[0,1]$, then
$\bern_y$ performs a Bernoulli trial with probability
$t[y]$ for each tuple $t \in \full{Q}(\db)$, keeping the tuple if the trial 
succeeds and discarding it otherwise. The output of $\bern_y$ is hence
non-deterministic, as it depends on the stochastic random choices made. In other
words, $\bern$ is a randomized relational operator. In statistics, a process
where each element of a population is subjected to an independent Bernoulli
trial that determines whether the element becomes part of a sample is called
\emph{Poisson sampling}. We hence refer to $\bern$ as the \emph{Poisson sampling
  operator}. For simplicity, we assume throughout the paper that the $y$-values are always in  the range $[0,1]$.

\begin{example}\it
  \label{ex:poisson-query}
  The rule defining $\contact(\personidone, \personidtwo)$
  in Example~\ref{ex:epiql-contacts} can be formalized as the Poisson sampling query 
  \begin{align*}
\contactquery =     
\bern_{\contactprobability} \big( &  
 \personmember(\personidone, \ageone, \poolid) \\ & {} \Join
 \personmember(\personidtwo, \agetwo, \poolid)\\ 
 &{} \Join
 \contactprob(\poolid, \ageone, \agetwo, \contactprobability)
\big).
\end{align*}
\end{example}

\smallskip\noindent\textbf{Acyclicity.}  Throughout the paper, we focus on
Poisson sampling queries for which $\full{Q}$ is acyclic. A join query $\full{Q}$ is
\emph{acyclic} if it admits a join
tree~\cite{DBLP:conf/stoc/BeeriFMMUY81,DBLP:journals/jacm/Fagin83}. A \emph{join
  tree for $\full{Q}$} is a rooted undirected tree $J$ in which each node is an
atom of $\full{Q}$. To be correct under bag semantics, it is required that each
atom in $\full{Q}$ appears exactly as many times in $J$ as it does in $\full{Q}$.  Join trees
are required to satisfy the \emph{connectedness property}: for every attribute
$x$, all the atoms containing $x$ form a connected subtree of $J$.  To
illustrate, Figure~\ref{fig:join-tree} shows a join tree for the join query in
Example~\ref{ex:poisson-query}. The triangle query $R(x,y) \Join S(y,z) \Join T(z,x)$ is the prototypical example of a join query that is \emph{cyclic}, i.e., not acyclic. Checking whether a join query is acyclic and constructing a join tree if it
exists can be done in linear time w.r.t. the size of the query by means of the
GYO
algorithm~\cite{DBLP:conf/compsac/YuO79,graham-gyo,DBLP:journals/siamcomp/TarjanY84}.

\begin{figure}[tbp]
\centering
\begin{tikzpicture}[scale=0.8, every node/.style={transform shape}]
    \node (R) at (0,0) {$\contactprob(\poolid, \ageone, \agetwo, \contactprobability)$};
    \node (S) at (-2,-1) {$\personmember(\personidone, \ageone, \poolid)$};
    \node (T) at (2,-1) {$\personmember(\personidtwo, \agetwo, \poolid)$};

    \draw[-] (R) -- (S);
    \draw[-] (R) -- (T);
\end{tikzpicture}
 \caption{\label{fig:join-tree} A join tree for the join query of Example~\ref{ex:poisson-query}.}
\end{figure}

\smallskip\noindent\textbf{Processing.}  We are interested in the efficient
processing of acyclic Poisson sampling queries.  Just like a deterministic query
may have multiple algorithms (i.e., physical query plans) that implement the
query---differing in their efficiency---a Poisson sampling query may have
multiple randomized algorithms implementing it.

The naive way to process $Q(\db)$ is to first materialize the full join
$\full{Q}(\db)$, and then iterate over the elements $t$ of the resulting
bag---flipping a coin with probability $t[y]$ to see if $t$ needs to be included
in the output. This naive method, which we will refer to as
\emph{Materialize-and-Scan (\mas for short)}, clearly has complexity
$\Omega(\size{\db} + \size{\full{Q}(\db)})$. Note that $\size{\full{Q}(\db)}$
may be much larger than the size of the returned sample.  We are interested in
developing practical alternate algorithms of lower complexity.

In what follows, we analyze our algorithms in the RAM model of computation with
the unit cost model. We assume that  that drawing random values from the uniform
distribution takes constant time; that building hash tables is in linear time; and
that probing a hash tables takes constant time. We consider the query itself
fixed, and hence focus on data complexity.

\smallskip\noindent\textbf{Discussion.}
While we restrict $y$ to be a single attribute in \eqref{eq:tir-query}, our approach naturally extends to the setting where the sampling probability is computed from a number of attributes $\overline{y}$ that all belong to the same relation $R_i$. Also note that we allow the relation symbols $R_i$ to be equal. Hence, our approach also applies to self-joins. 

Our approach also applies to queries with so-called free-connex projection. We denote bag-based projection by $\pi_A$ and set-based projection by $\delta \pi_A$, where $\delta$ is the duplicate elimination operator. 
Specifically, consider queries of the modified form
\begin{equation}
  \label{eq:free-connex-query}
  Q = \bern_y \left( \pi_A \left( \hat{Q} \right)\right) \quad \text{or} \quad Q = \bern_y \left( \delta \pi_A \left( \hat{Q} \right)\right),
  \end{equation}
  where $\hat{Q} =  R_1(\x_1) \Join \dots \Join R_l(\x_l)$, $A \subseteq \bigcup_{i \in [l]} \x_i$, and  $y \in A$. 
In Section~\ref{sec:position-sampling}, we show that our approach and complexity results extend to such queries when $\pi_A (\hat{Q})$ is \emph{free-connex acyclic}, which is defined as follows.
 Let $Q'$ be the full join query obtained by extending $\hat{Q}$ with an additional atom having $A$ as its variables. Then $ \pi_A (\hat{Q})$ is  \emph{free-connex acyclic} if both $\hat{Q}$ and $Q'$ are acyclic.

\section{Solution Overview}
\label{sec:solution-overview}

Let $\full{Q}$ be a join query. A \emph{random-access} index for $\full{Q}$ on input $\db$ is a data structure $\idx$ that is equipped with a method $\access$. The index is free to fix an order on the tuples in $\full{Q}(\db)$. Given positions $\pos = [i_1,\dots,i_k]$ with  $0 \leq i_j < \card{Q(\db)}$ for every $j$, $\idx.\access(\pos)$ returns the bag $\bag{t_{i_1},\dots, t_{i_k}}$, where $t_{i_j}$ is the $(i_j+1)$-th tuple of $\full{Q}(\db)$ in the chosen order. The \emph{access time} refers to the time required to construct this result, given $\pos$.

To efficiently process a Poisson sampling query $Q$ on database $\db$ we will adopt the following  \emph{Index-and-Probe} (\iap) strategy:
\begin{enumerate}[(1)]
\item Compute random-access index $\idx$ of $\full{Q}$ on $\db$.
\item Determine $\pos = [i_1,\dots,i_k]$, the sequence of tuple positions in $Q(\db)$ that define the output sample.
\item Return $\idx.\access(\pos)$.
\end{enumerate}
The sequence $\pos$ computed in step (2) is called the  \emph{probe sequence}. Determining this sequence is called \emph{position sampling}.

Clearly, the running time of \iap is determined by (1) the index construction time; (2) the probe sequence construction time; and (3) the random access time.

We base index construction on the approach recently taken by Bekkers et al.
~\cite{DBLP:journals/pvldb/BekkersNVW25} for processing acyclic joins.
Concretely, Bekkers et al. show that acyclic joins can be expressed in the
\emph{Nested Semijoin Algebra} (\nsa for short) as a sequence of \emph{nested}
semijoin operations, followed by a single flatten operation. Logically, nested
semijoin operations produce \emph{nested relations}, i.e., relations where a
tuples' attribute value may contain an entire relation instead of a scalar value
as is typically the case. The flatten operator turns a nested relation back into
an ordinary flat relation. While nested semijoin and flatten are logical
operators, Bekkers et al. propose a specific physical representation of nested
relations in main-memory column stores, based on so-called query
shredding~\cite{DBLP:journals/pvldb/BekkersNVW25, DBLP:journals/tcs/Bussche01}.
The implementation of nested semijoin and flatten based on this representation
always evaluates a sequence of nested semijoins followed by a flatten in
$\bigo(\insize + \card{\full{Q}(\db)})$ time, recovering Yannakakis' seminal
result that acyclic joins can be evaluated instance-optimally. What is more,
this way of implementing Yannakakis is competitive with classical binary joins
even when the input relations have no or only few dangling tuples---when
traditional Yannakakis implementations are observably
slower~\cite{DBLP:journals/pvldb/BekkersNVW25}.

In Section~\ref{sec:indexing}, we will show that the shredded representation of
Bekkers et al., which we will refer to as the \emph{chained shredded
  representation} (\csr) already provides a random-access index for
$\full{Q}$. In other words, we can see nested semijoins as a logical operator
that physically constructs a random-access index structure. In
Section~\ref{sec:position-sampling}, we will further show how this representation
can be used for efficient position sampling.

Previously, Carmeli et al.~\cite{DBLP:journals/tods/CarmeliZBCKS22} have shown
that acyclic joins admit a random-access index with single-tuple access time
$\bigo(\log \card{\db})$. As we will see, the access time of the \csr index is
asymptotically more expensive, and therefore not optimal. We therefore also
consider in Section~\ref{sec:indexing} a second representation, called the
\emph{unchained shredded representation} (\usr) that engineers the idea of the
random-access index structure of Carmeli et
al~\cite{DBLP:journals/tods/CarmeliZBCKS22} in column stores using the shredded
Yannakakis framework.

Before turning to describe these index structures, however, we introduce the nested semijoin algebra of ~\cite{DBLP:journals/pvldb/BekkersNVW25}, focusing on the nested semijoin ($\nsemijoin$) and flatten ($\flatten$) operators. We will use doubly curly braces $\bag{\dots}$ to denote bags as well as bag comprehension.

\smallskip\noindent\textbf{Schemes and Nested Relations.} Just like flat relations have flat schemes, nested relations have nested schemes. 
We refer to the
attributes that appear in the scheme of classical flat relations as \emph{flat}
attributes.  A \emph{(nested) scheme} is a finite set $X$, like $\{x, \{y\}, \{u,\{v\}\}\}$, that consists of flat attributes ($x$ in this case) and other schemes (i.e., $\{y\}$ and $\{u,\{v\}\}$). No flat attribute is allowed to occur  more than 
once, so 
$\{x, \{y\}, \{u,\{x\}\}\}$ is not a valid scheme. Schemes are also called \emph{nested}
attributes.  We range over flat
attributes by lowercase letters $(x,y,\dots)$;  over schemes by uppercase letters ($X$, $Y$, \dots), both from the end of the
alphabet; and over finite sets of flat attributes by $\x$. 
We write $\flatsub(X)$ for the set of all flat attributes occurring
somewhere in $X$ (either directly or in some inner nested scheme); and $\rsub(X)$
for the set of all schemes occurring in $X$, including $X$ itself (again, either directly or in some
inner nested scheme). So for $X = \{x, \{y\}, \{u,\{v\}\}\}$,
$\flatsub(X) = \{x,y,u,v\}$ and $\rsub(X) = \{ \{y\}, \{u,\{v\}\}, \{v\}, X \}$.

Nested relations and nested tuples are defined mutually recursively as follows. 
A \emph{relation} over a scheme $X$ is a finite bag of tuples over $X$. Here, a \emph{tuple over} $X$ is a mapping $t$ on $X$ such that $t(x)$ is a  scalar data value  for each flat attribute $x \in X$, and $t(Y)$ is a non-empty relation over $Y$ for each nested attribute $Y \in X$. Note that if $X$ is \emph{flat}, i.e., if
$X$ consists of flat attributes only, then this definition of a relation over $X$ coincides with the usual one. We call $R$ a \emph{flat relation} in that case. We restrict inner nested relations to be non-empty as in this paper we always start from flat relations and the operators that we consider will never introduce empty inner nested relations.
We write $R\colon X$ and $t\colon X$ to denote that $R$ is a relation (resp. $t$
is a tuple) over scheme $X$. We write $\card{R}$ to denote the total number of
tuples in $R$. Note that this only refers to the number of tuples in the
outer-most bag of $R$, and does not say anything about the cardinality of the
inner-nested relations appearing in those tuples.
Figure~\ref{fig:R-semijoin-S} shows a nested relation with
cardinality 4 and scheme $\{x, y, p, \{u,a\}\}$.

We adopt the following notation on tuples. If $s\colon X$ and $t\colon Y$ are tuples over disjoint schemes then $s \uplus t$ denotes their concatenation, which is a tuple over $X \cup Y$. 
Furthermore, if $Z \subseteq X$ then $\restr{s}{Z}$ denotes the restriction
(i.e., projection) of $s$ to the attributes in $Z$. The standard relational operators $\pi$ and $\sigma$ are extended to nested relations in the obvious manner. For example,  $\pi_{Z}(R) = \bagof{t[Z]}{t \in R}$ and $\sigma_{\y = s}(R) = \bagof{ t \in R}{t[\y] = s}$.

\smallskip\noindent\textbf{Nested Semijoins and Flatten.} The \emph{nested
  semijoin} operator $\nsemijoin$ takes two arguments, $R\colon X$ and
$S\colon Y$.\footnote{In~\cite{DBLP:journals/pvldb/BekkersNVW25}, the nested
  semijoin operator is defined as the composition of two other operators,
  $\nsemijoinsymb$ and $\group$; here we combined them into the single operator
  $\nsemijoin$ for convenience.}  For a valid application of the nested semijoin
operator, it is required that $X \cap Y$ contains only flat attributes, and that
$X \cup Y$ is again a scheme. Hence, $X = \{x,y, \{v\} \}$ and
$Y = \{ y, \{u, w\}\}$ is allowed. But $X = \{x, y, \{v\} \}$ and
$Y = \{y, \{x,u\}\}$ is not since $X \cup Y$ contains $x$ multiple times.

Let $\z = X \cap Y$ be the flat attributes shared between $X$ and $Y$, and let $Z = Y \setminus X$. Observe that for $\z$-tuple $t$, $\sigma_{\z = t}(S)$ is the subbag of $S$ containing those tuples having $t$ as join key.
Then $R \nsemijoin S$ extends the tuples in $R$ with one nested attribute, $Z$, containing the non-empty subbag of tuples in $S$ that join:
\begin{equation}
  \label{eq:sem-nsemijoin}
  R \nsemijoin S \defeq \bagof{r \uplus \{Z \mapsto \pi_{Z}(\sigma_{\z = r[\z]}(S) \} }{r \in R, \sigma_{\z = r[\z]}(S) \not = \emptyset}.
\end{equation}

To illustrate, Figure~\ref{fig:R-semijoin-S} shows the result of $R \nsemijoin S$ on the relations $R$ and $S$ shown in Figure~\ref{fig:input-relations} where $\z = \{x\}$ and $Z=\{u,a\}$.

\begin{figure}
  {\setlength{\tabcolsep}{3pt}
  \begin{minipage}[t]{0.35\columnwidth} \centering
    \subcaptionbox{Input relations \label{fig:input-relations}}{

\addtolength{\tabcolsep}{-2pt}
\small 

\begin{tikzpicture}[node distance=0.5cm]

\newcommand{\nodedistx}{0.2cm}
\newcommand{\nodedisty}{0.2cm}
\newcommand{\g}[1]{\textcolor{gray}{#1}}
\newcommand{\mybot}{0}

\tikzstyle{detailnode}=[inner sep=0pt, text centered]

\node [detailnode] (R) {%
	\scriptsize
\begin{tabular}{cccc}
    \multicolumn{4}{c}{$R$}\\ 
    \cline{2-4} 
    \cline{2-4}
              &  $x$   & $y$   &  $p$ \\ 
    \cline{2-4}
     \g{$0$}  &  $x_1$ & $y_1$ & $p_1$ \\
     \g{$1$}  &  $x_1$ & $y_2$ & $p_2$ \\
     \g{$2$}  &  $x_4$ & $y_3$ & $p_3$ \\
     \g{$3$}  &  $x_2$ & $y_1$ & $p_4$ \\
     \g{$4$}  &  $x_2$ & $y_2$ & $p_5$ \\
     \g{$5$}  &  $x_4$ & $y_3$ & $p_6$ \\
\end{tabular}
};

\node [detailnode,anchor=north west] (S) 
  at ($ (R.north east) + (\nodedistx, 0em) $){%
	\scriptsize
\begin{tabular}{ccc}
    \multicolumn{3}{c}{$S$}\\ \hline %
      $u$   & $a$   &  $x$  \\ \hline
      $u_1$ & $a_1$ &  $x_1$ \\
      $u_1$ & $a_1$ &  $x_2$ \\
      $u_2$ & $a_1$ &  $x_1$ \\
      $u_3$ & $a_2$ &  $x_1$ \\
      $u_3$ & $a_2$ &  $x_3$ \\
      $u_4$ & $a_3$ &  $x_2$ \\
\end{tabular}
};

\node [detailnode,anchor=north west] (T) 
  at ($ (S.north east) + (\nodedistx, 0em) $){%
	\scriptsize
\begin{tabular}{cc}
    \multicolumn{2}{c}{$T$}\\ \hline %
    $v$   & $y$  \\ \hline
    $v_1$ & $y_4$ \\
    $v_2$ & $y_2$ \\
    $v_3$ & $y_1$ \\
    $v_4$ & $y_2$ \\
    $v_5$ & $y_1$ \\
    $v_6$ & $y_2$ \\
\end{tabular}
};
\end{tikzpicture}

     } 
  \end{minipage}
  \begin{minipage}[t]{0.25\columnwidth} \centering
    \subcaptionbox{$N_1$ \label{fig:R-semijoin-S}}{
\addtolength{\tabcolsep}{-2pt}
\small 

\begin{tikzpicture}[node distance=0.5cm]

\newcommand{\nodedistx}{0.2cm}
\newcommand{\nodedisty}{0.2cm}
\newcommand{\g}[1]{\textcolor{gray}{#1}}
\newcommand{\mybot}{0}

\tikzstyle{detailnode}=[inner sep=0pt, text centered]

\node [detailnode] (R) {%
	\scriptsize
\begin{tabular}{cccc}
    \multicolumn{4}{c}{
     }\\[2pt] \hline %
     $x$   & $y$   &  $p$  & $\{u,a\}$\\[1pt] \hline %
     \\[-5pt]
     $x_1$ & $y_1$ & $p_1$ &  \begin{tabular}{|cc|}
                                    \hline 
                                    $u_1$&$a_1$\\
                                    $u_2$&$a_1$\\
                                    $u_3$&$a_2$\\
                                    \hline 
                               \end{tabular}\\
    \\[-5pt]
     $x_1$ & $y_2$ & $p_2$ &  \begin{tabular}{|cc|}
                                    \hline 
                                    $u_1$&$a_1$\\
                                    $u_2$&$a_1$\\
                                    $u_3$&$a_2$\\
                                    \hline 
                               \end{tabular}\\
    \\[-5pt]
     $x_2$ & $y_1$ & $p_4$ &  \begin{tabular}{|cc|}
                                    \hline 
                                    $u_1$&$a_1$\\
                                    $u_4$&$a_3$\\
                                    \hline 
                               \end{tabular}\\
    \\[-5pt]
     $x_2$ & $y_2$ & $p_5$ &  \begin{tabular}{|cc|}
                                    \hline 
                                    $u_1$&$a_1$\\
                                    $u_4$&$a_3$\\
                                    \hline 
                            \end{tabular}\\
\end{tabular}
};

\end{tikzpicture}     } 
  \end{minipage}
  \begin{minipage}[t]{0.30\columnwidth} \centering
    \subcaptionbox{$N_2$
    \label{fig:R-semijoin-S-semijoin-T}}{%
\addtolength{\tabcolsep}{-2pt}
\small 

\begin{tikzpicture}[node distance=0.5cm]

\newcommand{\nodedistx}{0.2cm}
\newcommand{\nodedisty}{0.2cm}
\newcommand{\g}[1]{\textcolor{gray}{#1}}
\newcommand{\mybot}{0}

\tikzstyle{detailnode}=[inner sep=0pt, text centered]

\node [detailnode] (R) {%
	\scriptsize
\begin{tabular}{ccccc}
    \multicolumn{5}{c}{
     }\\[2pt] \hline 
     $x$   & $y$   &  $p$  & $\{u,a\}$ & $\{v\}$\\[1pt] \hline %
     \\[-5pt]
     $x_1$ & $y_1$ & $p_1$ &  \begin{tabular}{|cc|}         
                                    \hline 
                                    $u_1$&$a_1$\\        
                                    $u_2$&$a_1$\\        
                                    $u_3$&$a_2$\\        
                                    \hline         
                               \end{tabular} 
                          &    \begin{tabular}{|c|}         
                                    \hline 
                                    $v_3$\\        
                                    $v_5$\\        
                                    \hline         
                               \end{tabular}\\       
    \\[-5pt]        
     $x_1$ & $y_2$ & $p_2$ &  \begin{tabular}{|cc|}
                                    \hline 
                                    $u_1$&$a_1$\\
                                    $u_2$&$a_1$\\
                                    $u_3$&$a_2$\\
                                    \hline 
                               \end{tabular}
                               & \begin{tabular}{|c|}
                                    \hline 
                                    $v_2$ \\                                    
                                    $v_4$\\
                                    $v_6$\\
                                    \hline 
                               \end{tabular}\\
    \\[-5pt]
     $x_2$ & $y_1$ & $p_4$ &  \begin{tabular}{|cc|}
                                    \hline 
                                    $u_1$&$a_1$\\
                                    $u_4$&$a_3$\\
                                    \hline 
                               \end{tabular}
                                &    \begin{tabular}{|c|}         
                                    \hline 
                                    $v_3$\\        
                                    $v_5$\\        
                                    \hline         
                               \end{tabular}\\       
    \\[-5pt]
     $x_2$ & $y_2$ & $p_5$ &  \begin{tabular}{|cc|}
                                    \hline 
                                    $u_1$&$a_1$\\
                                    $u_4$&$a_3$\\
                                    \hline 
                            \end{tabular}
                            & \begin{tabular}{|c|}
                                    \hline 
                                    $v_2$\\
                                    $v_4$\\
                                    $v_6$\\
                                    \hline 
                               \end{tabular}\\
\end{tabular}
};

\end{tikzpicture}     } 
  \end{minipage}

  \vspace{.5cm}

\begin{minipage}[t]{0.99\columnwidth} \centering
    \subcaptionbox{ Chained shredded representation of $N_2$
    \label{fig:CSR}}{%
\addtolength{\tabcolsep}{-2pt}
\small 

\begin{tikzpicture}[node distance=0.5cm]

\newcommand{\nodedistx}{0.2cm}
\newcommand{\nodedisty}{0.2cm}
\newcommand{\g}[1]{\textcolor{gray}{#1}}
\newcommand{\bl}[1]{\textcolor{blue}{#1}}
\newcommand{\mybot}{-1}

\tikzstyle{detailnode}=[inner sep=0pt, text centered]

\node [detailnode] (R) { %
    \scriptsize
    \vspace{1cm}
\begin{tabular}{ccc|cc|cc|c}
   \multicolumn{8}{c}{$\store{}(\{x,y,p, \{u,a\}, \{v\}\})$} \\ \hline \hline 
   $x$ &  $y$  & $p$   &\multicolumn{2}{c|}{$\{u,a\}$} & \multicolumn{2}{c|}{$\{v\}$}& ${\prefix{}}$\\
       &       &       & $\holshort{}$ & $w$ & $\holshort{}$ & $w$ & \\   \hline
  $x_1$& $y_1$ & $p_1$ &  3& 3 &4 & 2 & \bl{6}\\
  $x_1$& $y_2$ & $p_2$ &  3& 3 &5 & 3 & \bl{15}\\
  $x_2$& $y_1$ & $p_4$ &  5& 2 &4 & 2 & \bl{19}\\
  $x_2$& $y_2$ & $p_5$ &  5& 2 &5 & 3 & \bl{25}\\
\end{tabular}
\vspace{1cm}
};

\node [detailnode, anchor=north west] (ua)
      at ($ (R.north east) + (\nodedistx, 0em) $)  {%
	\scriptsize
\begin{tabular}{ccc|c}
    \multicolumn{4}{c}{$\store{}(\{u,a\})$}\\ \hline \hline
              &  $u$   &  $a$ & $\nxt$ \\ \hline
     \g{$0$}  &  $u_1$ &  $a_1$ & $\mybot$\\          
     \g{$1$}  &  $u_1$ &  $a_1$ & $\mybot$\\
     \g{$2$}  &  $u_2$ &  $a_1$ & $0$\\
     \g{$3$}  &  $u_3$ &  $a_2$ & $2$ \\
     \g{$4$}  &  $u_3$ &  $a_2$ & $\mybot$ \\
     \g{$5$}  &  $u_4$ &  $a_3$ & $1$\\
\end{tabular}
};

\node [detailnode, anchor=north west] (v)
      at ($ (ua.north east) + (\nodedistx, 0em) $)  {%
	\scriptsize
\begin{tabular}{cc|c}
    \multicolumn{3}{c}{$\store{}(\{v\})$}\\ \hline \hline
       & $v$ & $\nxt$ \\ \hline
     \g{$0$}  &  $v_1$ & $\mybot$\\
     \g{$1$}  &  $v_2$ & $\mybot$\\
     \g{$2$}  &  $v_3$ & $\mybot$ \\
     \g{$3$}  &  $v_4$ & $1$\\
     \g{$4$}  &  $v_5$ & $2$\\
     \g{$5$}  &  $v_6$ & $3$\\
\end{tabular}
};
\end{tikzpicture}

     } 
  \end{minipage}

  \vspace{.5cm}

\begin{minipage}[t]{0.9\columnwidth} \centering
    \subcaptionbox{Unchained shredded representation of $N_2$
    \label{fig:UCSR}}{%
\addtolength{\tabcolsep}{-2pt}
\small 

\begin{tikzpicture}[node distance=0.5cm]

\newcommand{\nodedistx}{0.2cm}
\newcommand{\nodedisty}{0.2cm}
\newcommand{\g}[1]{\textcolor{gray}{#1}}
\newcommand{\bl}[1]{\textcolor{blue}{#1}}
\newcommand{\mybot}{0}

\tikzstyle{detailnode}=[inner sep=0pt, text centered]

\node [detailnode] (R) { %
    \scriptsize
\begin{tabular}{ccc|ccc|ccc|c}
    \multicolumn{10}{c}{$\store{}(\{x,y,p, \{u,a\}, \{v\}\})$} \\ \hline \hline
   $x$ &  $y$  & $p$   &\multicolumn{3}{c|}{$\{u,a\}$} & \multicolumn{3}{c|}{$\{v\}$}& ${\prefix{}}$\\
       &       &       & $\start{}$ & $\llen{}$ & $w$ & $\start{}$ & $\llen{}$ & $w$ & \\   \hline
        $x_1$& $y_1$ & $p_1$ &  0& 3 & 3 & 0 & 2 & 2 & \bl{6}\\
        $x_1$& $y_2$ & $p_2$ &  0& 3 & 3 & 2 & 3 & 3 & \bl{15}\\
        $x_2$& $y_1$ & $p_4$ &  3& 2 & 2 & 0 & 2 & 2 & \bl{19}\\
        $x_2$& $y_2$ & $p_5$ &  3& 2 & 2 & 2 & 3 & 3 & \bl{25}\\
\end{tabular}
};

\node [detailnode, anchor=north west] (ua)
      at ($ (R.north east) + (\nodedistx, 0em) $)  {%
	\scriptsize
\begin{tabular}{ccc|cc}
    \multicolumn{4}{c}{$\store{}(\{u,a\})$}\\ \hline \hline
              &  $u$   &  $a$   & $\perm$ & $\prefix{}$\\ \hline
     \g{$0$}  &  $u_1$ &  $a_1$ &    0     &  \bl{1}       \\
     \g{$1$}  &  $u_1$ &  $a_1$ &    2     &  \bl{2}       \\
     \g{$2$}  &  $u_2$ &  $a_1$ &    3     &  \bl{3}       \\
     \g{$3$}  &  $u_3$ &  $a_2$ &    1     &  \bl{1}       \\
     \g{$4$}  &  $u_3$ &  $a_2$ &    5     &  \bl{2}       \\
     \g{$5$}  &  $u_4$ &  $a_3$ &    4     &  \bl{1}       \\
\end{tabular}
};

\node [detailnode, anchor=north west] (v)
      at ($ (ua.north east) + (\nodedistx, 0em) $)  {%
	\scriptsize
\begin{tabular}{cc|cc}
    \multicolumn{4}{c}{$\store{}(\{v\})$}\\ \hline \hline
              & $v$    & $\perm$ & $\prefix{}$\\ \hline
     \g{$0$}  &  $v_1$ &    2    &   \bl{1}     \\ 
     \g{$1$}  &  $v_2$ &    4    &   \bl{2}     \\ 
     \g{$2$}  &  $v_3$ &    1    &   \bl{1}     \\ 
     \g{$3$}  &  $v_4$ &    3    &   \bl{2}     \\ 
     \g{$4$}  &  $v_5$ &    5    &   \bl{3}     \\ 
     \g{$5$}  &  $v_6$ &    0    &   \bl{1}     \\ 
\end{tabular}
};      
\end{tikzpicture}

     } 
  \end{minipage}
  \vspace{.5cm}

  \caption{Illustration of nested relations and nested semijoins with $N_2=$ $(R(x,y,p) \nsemijoin S(u,a,x))\nsemijoin T(v,y)$ and $N_1 = R(x,y,p)$ $\nsemijoin S(u,a,x)$.\label{fig:inputs-and-R-nsemijoin-S}} }
\end{figure}

The flatten operator $\flatten$ converts a nested relation $N\colon X$ or nested
tuple $t\colon X$ into a flat relation with scheme $\flatsub(X)$. Specifically,
assume that $X$ consists of flat attributes $\x$ and nested attributes
$Y_1,\dots, Y_m$. Then
\begin{align}
  \label{eq:flatten-tuple}
  \flatten(t) & \defeq \bagof{ t[\x] \uplus t_1 \uplus \dots \uplus t_m}{t_j \in
                \flatten\left(t(Y_j)\right), j \in [m]}; \\
  \label{eq:flatten-rel}
  \flatten(N) & \defeq \bigcup \bagof{\flatten(t)}{t \in N}. 
\end{align}
That is, $\flatten(t)$ pairs the flat
attributes of $t$ with all combinations of tuples obtained by flattening the
nested relations $t(Y_1), \dots, t(Y_m)$ and $\flatten(N)$ is the bag union of flattening all of its tuples. Note that when $t$ and $N$ are flat, i.e, $m = 0$ then $\flatten(t) = \bag{t}$ and $\flatten(N) = N$. 

To illustrate, reconsider Figure~\ref{fig:input-relations}. When we flatten the nested relation in Figure \ref{fig:R-semijoin-S} we obtain the same relation as the one that results from $R \Join S$. When we flatten the nested relation in Figure~\ref{fig:R-semijoin-S-semijoin-T} we obtain $R \Join S \Join T$.

\smallskip\noindent\textbf{Weights.}  The \emph{weight} of a nested relation $R$ (resp. tuple $r$)
is the total number of tuples produced when flattening $R$ (resp $r$), i.e.
$\weightval(R) = \card{\flatten(R)}$ (resp. $\weightval(r) = \card{\flatten(r)}$).

\smallskip\noindent\textbf{\nsa.} While Bekkers et al. also consider additional
operators on nested relations, for the purpose of this paper we define the
Nested Semijoin Algebra (\nsa) to be all expressions that are built from flat
relation symbols using valid application of $\nsemijoin$ and $\flatten$. Such an
expression is called \emph{2-phase} (\twonsa) if it is of the form $\flatten(E)$
with $E$ itself not containing any further application of $\flatten$. We call
$E$ \emph{$\flatten$-free} in this case. Bekkers et al.\@ observe that from a join
tree $J$ for an acyclic join query $\full{Q}$ we may compute a \twonsa expression
$\flatten(E)$ as follows. First, construct the $\flatten$-free expression $E$ by
traversing $J$ bottom up. For leaf nodes, $E$ is simply the node's atom. For
interior nodes, $E$ is the result
$(\dots (R(\x) \nsemijoin E_1) \nsemijoin \dots ) \nsemijoin E_n$ of
taking the semijoin of node's atom $R(\x)$ with the expressions recursively
obtained for each of the children. Finally, add $\flatten$ to the expression
obtained for the root. 
For later use we strengthen this observation as follows.

\begin{proposition}
  \label{prop:probability-attribute-at-root}
  For every acyclic join query $\full{Q}$ and any attribute $y$ of $\full{Q}$ we
  can compute an equivalent two-phase \nsa expression $\flatten(E)$ such that
  $y$ is a flat attribute in the output scheme of $E$.
\end{proposition}
\begin{proof}
  First, compute a join tree $J$ of \full{Q}. Pick any node in $J$ that mentions
  $y$, and reroot $J$ into a join tree $J'$ that has this node as the root.
  Then the above-mentioned procedure applied to $J'$ yields the desired \twonsa
  expression.
\end{proof}

\section{Shredded Random-Access Indexing}
\label{sec:indexing}

\smallskip\noindent\textbf{Columnar storage.} We assume that we are
working in main memory, and that a flat relation $R(x_1,\dots,x_n)$ is
physically represented as a tuple
$\phys{R} = (\phys{R}.x_1, \dots, \phys{R}.x_n)$ of vectors $\phys{R}.x_i$, all
of length $\card{R}$. It is understood that values at the same offset in these
vectors encode a complete tuple, i.e.,
$R = \bagof{ (\phys{R}.x_1[i],\dots,\phys{R}.x_n[i])}{0 \leq i < \card{R}}$.
In particular, it is possible to refer to tuples positionally, i.e., the tuple at offset $0$ in $\phys{R}$, the tuple at offset $1$, and so on.  We will
refer to $\phys{R}$ as a \emph{physical relation}, and denote the number of tuples in
$\phys{R}$ by $\len{\phys{R}}$. If
$0 \leq i < \len{\phys{R}}$  and $\y = y_1,\dots,y_k$ is a subset of
$\{x_1,\dots, x_n\}$, then we write $\phys{R}[i](\y)$ for the tuple
$(\phys{R}.y_1[i],\dots,\phys{R}.y_k[i])$. A \emph{position vector for} $\phys{R}$ is a vector of natural numbers, all between $0$ and $\len{\phys{R}} - 1$.

\smallskip\noindent\textbf{Shredded representations.}
\emph{Shredding}~\cite{DBLP:journals/tcs/Bussche01,DBLP:conf/sigmod/CheneyLW14,DBLP:conf/pods/Wong93}
refers to the representation of a nested relation by means of a collection of
flat relations. In our setting, the latter will be physical relations, as the
ability to refer to the individual columns and individual positions of tuples is
important for the efficiency of processing the representation.  Since we will
consider two distinct shredded representations, we first define the concept
generically, and then introduce the concrete representations.  Generically, a
\emph{shredded representation} of a nested relation $R\colon X$ is a collection
$\store{}$ of physical (flat) relations, one physical relation $\store{}(Y)$ for
every $Y \in \rsub(X)$. Each $\store{}(Y)$ contains columns for at least all
flat attributes of $Y$, but also has additional columns to encode the
hierarchical relationships of the nested tuples in $R$. We note that
when $X$ is a flat scheme, $\rsub(X) = \{X\}$, and $\store{}$ hence consists
then only of a single flat table, namely $\store{}(X)$, which is a physical
representation of $R$. 
Define the \emph{size} $\size{\store{}}$ of a shredded representation to be the total number of tuples that it contains, in all of its flat relations.

\subsection{Chained Shredding}
\label{sec:chained}

We next introduce the shredding representation of Bekkers et
al~\cite{DBLP:journals/pvldb/BekkersNVW25}, referred to as the \emph{chained
shredded representation} (\csr). Let $N\colon X$ be a nested relation. A \csr $\store{}$
of $N$ has one physical relation $\store{}(Y)$ for every $Y \in \rsub(X)$. This
relation contains one column for each flat attribute in $Y$. Moreover, for every
nested attribute $Z \in Y$, $\store{}(Y)$ has two extra columns $\hol{Z}$ and
$\weight{Z}$, both of which hold position vectors. Finally, if $Y$ is a strict
subscheme of $X$, i.e. $Y \not =X$, then $\store{}(Y)$ also has a column
$\nxt$. To illustrate, Figure~\ref{fig:CSR} shows a \csr for the nested relation
in Figure~\ref{fig:R-semijoin-S-semijoin-T}.

In a \csr, the $\nxt$ column is used to encode a linked list of tuples: for all offsets $0 \leq i < \len{\store{}(Y)}$, if $\store{}(Y).\nxt[i] < 0$ then the tuple at offset $i$ in $\store{}(Y)$ is the final tuple in the list; otherwise its successor in the list is the tuple at offset $\store{}(Y).\nxt[i]$. Hence, $\nxt$ is used to chain tuples together, see Figure~\ref{fig:CSR}.

Chained shredding works as follows: every tuple $t \in N$ is represented by
exactly one tuple in $\store{}(X)$. Let $i$ be the index of the tuple in
$\store{}(X)$ representing $t$. For every flat $x \in X$ we have
$t(x) = \phys{R}.x[i]$. For every nested attribute $Y \in X$,
$\store{}(X).\weight{Y}[i]$ stores the weight of $t(Y)$. Furthermore,
$\store{}(X).\hol{Y}[i]$ stores the head index of the linked list of tuples in
$\store{}(Y)$ that together represent the tuples occurring in $t(Y)$.  Note that
the tuples in $t(Y)$ may themselves contain further nested relations, and the
shredding hence proceeds recursively.

\begin{example}
  \label{ex:csr}
  To illustrate, Figure~\ref{fig:CSR} shows a \csr for the nested relation $N_2:X$ of
  Figure~\ref{fig:R-semijoin-S-semijoin-T}. Here, 
  $X = \{ x, y, p, \{u,a\}, \{v\}\}$. The $i$-th tuple in
  $N_2$ is represented by the $i$-th tuple in
  $\store{}(X)$. When we look at the first tuple of $N_2$ we see that
  $t(\{v\})= \bag{ v_3,v_5}$. In $\store{}$, this bag is encoded by the linked
  list in $\store{}(\{v\}).\nxt$ starting at offset $\hd = 4$.
\end{example}

It is important to note that multiple tuples in $\store{}$ may repeat the same
$\hol{Y}$ value; this effectively means that these tuples share the same nested
relation in the $Y$-value. See the $\{u,a\}$-value of the first and second tuple
of Figure~\ref{fig:R-semijoin-S-semijoin-T} for an illustration.  This sharing
is important to ensure that, starting from a number of flat relations we can
build a nested relation resulting from their semijoin whose physical
representation size is linear in the input size.

In what follows, we assume a method \texttt{weight\_of} on \csrs that, given a \csr $\store{}$ and one if its schemes $Y$ and an offset $0 \leq i < \card{\store{}(Y)}$, returns the weight  of the nested tuple represented at offset $i$ in $\store{}(Y)$. This is straightforward to compute: when $Y$ is flat, the weight is always one; otherwise it is the product of the weights of the nested attributes.

\newcommand{\csrgroup}{\textsf{csr-group}}
\begin{figure}
          \tt \small
          \begin{algorithmic}[1]
\Function{\csrgroup}{\store{S}}\algorithmicdo
            \State $\nxt = [0, \dots, 0]$ \Comment{$\len{\store{S}(X_S)}$ times}
            \State h = \{\} \Comment{maps keys -> (pos, weight)}
            \For{$0 \leq i < \len{\store{S}(X_S)}$}
            \State key = $\store{S}(X_S)[i](\z)$; w = \texttt{weight\_of}$(\store{S}, X_S, i)$
            \If{h.contains(key)}
            \State (j, prev\_w) = h[key]
            \State nxt[i] = j
            \State h[key] = (i, prev\_w + w)
            \Else
            \State{h[key] = (i, w)}
            \EndIf
            \EndFor
            \State \Return (h, \nxt)
            \EndFunction
          \end{algorithmic}
\caption{\csr grouping algorithm.}
\label{algo:csr-grouping}
\end{figure}

\smallskip\noindent\textbf{Building the representation.}  Bekkers et al
~\cite{DBLP:journals/pvldb/BekkersNVW25} show that, given \csrs $\store{R}$ and $\store{S}$ for nested relations $R\colon X_R$ and $S\colon X_S$, respectively, we can construct a \csr $\store{}$ for $R \nsemijoin S$ in time linear in the size of $\store{R}$ and $\store{S}$ as follows. 
The process is  akin to a hash join. Let $\z = X_R \cap X_S$ be the join attributes and $Z = X_S \setminus X_R$.
\begin{compactenum}[(1)]
\item First, group $S$ on the join attributes $\z$ by computing a pair
  $(h, \nxt)$ consisting of (i) a hash table $h$ that maps the join keys
  occurring in $S$ to pairs of the form $(i, w)$, and (ii) a position vector
  $\nxt$, encoding a collection of linked lists. For every join key $k$, if we
  start traversing $\nxt$ in a linked-list fashion starting at offset $i$, then
  we traverse the offsets of all tuples in $S$ that have join key
  $k$. Moreover, $w$ is the sum of weights of these tuples.
  Figure~\ref{algo:csr-grouping} shows the pseudo-code to compute $(h,\nxt)$
  in a single pass over $\store{S}$.
\item Second, drop all join columns from $\store{S}(X_S)$, and add the $\nxt$  vector computed in the previous step as an additional column.
\item Use  $\store{R}(X_R)$ to obtain $\store{}(X_R \cup \{Z\})$ by 
  looping over the tuples of $\store{R}(X_R)$, probing the hash table $h$ to see if a joining tuple in $\store{S}(Y)$ exists and if so, use the $(i,w)$ value stored in $h$ to populate the $\hol{Z}$ and $\weight{Z}$-values. Tuples without joining tuples in $\store{S}(Y)$ are dropped.
\item Now, $\store{R}$, $\store{S}$ {and $\store{}(X_R \cup \{Z\})$} together form the \csr $\store{}$ for $R \nsemijoin S$.
\end{compactenum}

The interested reader is invited to apply this procedure on the input relations of Figure~\ref{fig:input-relations} to obtain \csr for first $R \nsemijoin S$ and then $(R \nsemijoin S) \nsemijoin T$. Initially, $\store{R}$ consists only of $R$ itself, and similarly for $\store{S}$ and $\store{T}$.

\smallskip\noindent\textbf{Random access.} Given \csr $\store{}$ for nested
relation $N$, we can turn $\store{}$ into a random-access index for
$\flatten(N)$. That is, given an offset $0 \leq i < \len{\flatten(N)}$, we can
extract the $(i+1)$-th tuple of $\flatten(N)$ directly from $\store{}$ (i.e.,
directly from the representation of $N$) without first needing to materialize
$\flatten(N)$.

Conceptually, we order the tuples of $\flatten(N)$ as follows: the tuple $t$ of
$N$ that occurs first in $\store{}$ produces the first $\weightval(t)$ tuples of
$\flatten(N)$, the tuple $t$ that occurs second in $\store{}$ produces the next
$\weightval(t)$ tuples, and so on. Here, each tuple $t$ produces $\weightval(t)$
flat tuples, which cf. Equation \eqref{eq:flatten-tuple}, are obtained by combining the
flat attribute values of $t$ with the result of recursively flattening the
nested attributes $Y_1,\dots, Y_m$ of $X$.  We order the tuples of $\flatten(t)$
itself as follows. Let $0 \leq i < \weightval(t)$.  To be able to
unambiguously identify which combination of the tuples in
$\flatten\left(t(Y_j)\right)$ produces the tuple of $\flatten(t)$ at offset $i$, we view
$i$ as a number in a mixed-radix numeral system where the bases are the
weights of the $t(Y_j)$. That is, let $w_j$ abbreviate $\weightval(t(Y_j))$ and
let $(i_1,\dots, i_m)$ be offsets of tuples in
$(\flatten(t(Y_1),\dots, \flatten(t(Y_m))$, respectively. These tuples
together will produce the tuple of $\flatten(t)$ at offset
\[  i = i_1 + i_2w_1 + i_3 w_1 w_2 + \dots + i_m w_1\dots w_{m-1} \]
We note that, given $i$ it is straightforward to compute the offsets $i_j$ inductively as follows
\begin{align}
  \label{eq:mixed-radix-direct-access}
  i_1 & = i \imod w_1 & q_1 & = i  \idiv w_1 \\
  i_j & = (q_{j-1} \imod w_{j}) & q_j & = (q_{j-1} \idiv w_j) & j > 1
\end{align}
Then the tuple of $\flatten(t)$ at offset $i$ is obtained by combining the $i_j$-th tuple of $\flatten(t(Y_j))$ for every $j \in [m]$.

\begin{example}\it
  Consider Figure~\ref{fig:R-semijoin-S-semijoin-T}, which shows a nested relation
  with scheme $X = \{x,y,p, Y_1, Y_2\}$ where $Y_1 = \{u,a\}$ and $Y_2 =
  \{v\}$. Let $t$ be its first nested tuple,
  which will produce $3 \times 2 = 6$ tuples when flattened.  To identify the
  tuple that will be produced at offset $i = 4$ (i.e., the $5$-th tuple), we
  compute
  \[ i_1 = 4 \imod 3 = 1 \quad q_1 = 4 \idiv 3 = 1 \quad i_2 = (1 \imod 2) =
    1 \] Since $Y_1$ and $Y_2$ are flat schemes,
  $\flatten(t(Y_1)) = t(Y_1)$ and $\flatten(t(Y_2)) = t(Y_2)$. As such, $i_1$
  and $i_2$ give offsets directly in $t(Y_1)$ and $t(Y_2)$, which are
  represented at $\store{}(Y_1)$ and $\store{}(Y_2)$ respectively.  In other
  words, the $5$-th tuple obtained while flattening $t$ has $u = u_2$,
  $a = a_1$, $v = v_5$.
\end{example}

Given this order, we turn $\store{}$ into a random-access index by first adding
one extra column to $\store{}(X)$: the \emph{prefix vector} $\prefix{}$. This
vector contains, for every offset $0 \leq i < \len{N}$, the sum of the
weights of the nested tuples up and including offset $i$. To
illustrate, Figure~\ref{fig:CSR} shows this prefix vector in blue. The prefix
vector can clearly be computed in linear time.

\newcommand{\csraccess}{\textsf{csr-get}}
\newcommand{\csrsubaccess}{{\sc csr-sub}}
\begin{figure}
          \tt \small
          \begin{algorithmic}[1]
\Function{\csraccess}{$\store{},\, X,\, i$}\algorithmicdo
            \State result = $\{\}$ \Comment{result tuple, initially empty}
            \State find smallest $j$ in 0..\card{\store{}(X)} s.t. $i <  \prefix{}[j]$
            \State \csrsubaccess($\store, X, j, i - \prefix{}[j-1]$) \Comment{assume \prefix{}[-1] = 0}
            \State return result
            \EndFunction
            \smallskip
            \Function{\csrsubaccess}{$\store{}, X, j, i$}\algorithmicdo
            \State result = result$\uplus  \{ a \mapsto \store{}(X).a[j] \mid a \in X\}$
            \For{each nested attr $Y \in X$}
            \State $(j', w) = \store{}(X)[j](\hol{Y},\weight{Y})$
            \State $i' = i$ mod $w$; $i = i$ div $w$
            \State curr\_weight = weight\_of$(\store{}, Y, j')$
            \While{$j' \geq 0$ and $i' \geq$ curr\_weight}
            \State $i' = i' - curr\_weight$
            \State $j' = \store{}(Y).\nxt[j']$
            \State curr\_weight = weight\_of$(\store{}, Y, j')$
            \EndWhile
            \State \csrsubaccess($\store{}, Y, j', i'$)
            \EndFor
            \EndFunction
          \end{algorithmic}
\caption{\csr random access.}
\label{algo:csr-access}
\end{figure}

Given an offset $i$ into $\flatten(N)$, we can access the tuple at offset $i$
directly from the \csr $\store{}$ of $N$ as shown in
Figure~\ref{algo:csr-access}:
\begin{compactenum}[(1)]
\item Allocate space to store the result tuple.
\item Use binary search to locate the smallest index
  $0 \leq j < \card{\store{}(X)}$ for which $i < \prefix{}[j]$. This
  takes time $\bigo(\log \len{\store{}}))$ and identifies the nested tuple that produces the tuple at offset $i$ when flattened. 
\item Call \csrsubaccess\xspace to focus on the nested tuple $t$ represented at offset $j$ in $\store{}(X)$, and to construct from $t$ the tuple that is produced at offset  $(i - \prefix{}[j-1])$ when flattening $t$.
\item \csrsubaccess\xspace does this by first copying the values of all flat
  attributes, which are at offset $j$ in $\store{}(X)$. It remains to get the
  correct values for the nested attributes. This is done by iterating through
  all nested attributes $Y$. For each $Y$, we compute the offset $i'$ of the
  tuple that needs to be retrieved from $\flatten(t(Y))$ by means of the
  mixed-radix indexing scheme~(equation \eqref{eq:mixed-radix-direct-access}). Then, starting at $\store{}(Y).\hol{Y}[j]$, we iterate through $\store{Y}.\nxt$'s linked list of tuples, keeping track of how many tuples these would produce when flattened. As soon as this number exceeds $i'$ we have found the correct position, and we call \csrsubaccess\xspace recursively to populate the result values for $\flatsub(Y)$.

\end{compactenum}

\csrsubaccess\xspace will be called $h$ times, where $h =\card{\rsub(X)}$ is
the number of subschemes of $X$. Each call it linearly traverses a linked list.
The overall complexity is $\bigo(\log \size{\store{}(X)} + h \times d)$ where
$d$ is the maximum length of any of the linked lists that we need to traverse.
When $\store{N}$ was constructed using
nested semijoins starting from flat relations, then it is not difficult to see
that $\len{\store{}(X)} \leq \card{\db}$ and that $d$ can be at most the degree
of any join key in $\db$, which is formally defined as follows.

\begin{definition}
  Assume that $R\colon \x$ is a flat relation, let $\y \subseteq \x$ and let $t$ be a $\y$-tuple. The degree of $t$ in $R$ is the number of times that $t$ occurs in $\pi_y(R)$. In other words, it is the cardinality of $\sigma_{\y = t}(R)$. The degree of $\y$ in $R$, denoted $\degree_{\y}(R)$  is the maximum degree of any $y$-tuple in $R$. The \emph{degree} of a join query 
$\full{Q} = R_1(\x_1) \Join \dots \Join R_k(\x_k)$ in a database $\db$ is defined as
\[ \degree_{\full{Q}}(\db) \defeq \max \{ \degree_{\x_i \cap \x_j}(R_i) \mid i \not = j, \x_i \cap \x_j \not = \emptyset \} \]
\end{definition}

We may hence conclude:

\begin{proposition}
  Using nested semijoins it is possible to construct in $\bigo(\size{\db})$ time
  a \csr that can be used as a random-access index for an acyclic join query $\full{Q}$,
  with access time $\bigo(\log \size{\db} + \degree_{\full{Q}}(\db))$.
\end{proposition}
\begin{proof}
  Since $\full{Q}$ is acyclic, we can compute $\full{Q}(\db)$ by means of
  \twonsa expression $\flatten(E)$, where $E$ consists of nested semijoins
  only. By our earlier reasoning we can build a \csr for $N = E(\db)$ in time
  linear in $\db$. We can then compute the prefix vector also in linear time.
  This \csr + prefix vector allows random access into $\flatten(N)$ with access
  time $\bigo(\log \size{N} + h \times \degree_Q(\db))$. The result follows by
  observing that necessarily $\size{N} = \bigo(\size{\db})$ and that $h$ is
  constant in data complexity.
\end{proof}

\smallskip\noindent\textbf{Caching optimization.} 
Above, we focused on random access for a single tuple offset $i$. 
When given a sequence $\pos = [i_1,\dots,i_k]$ of such positions to access randomly in bulk, we can apply the following \emph{caching optimization}. 
During random access to the output tuple at position $i_\ell$, we traverse the linked list starting at $\store{}(Y).\nxt[j']$. We cache the position (i.e., a pointer) at which this traversal terminates. If the subsequent random access to the output tuple at position $i_{\ell+1}$ requires traversing the same linked list and the desired item appears further along that list, we resume the traversal from the cached pointer rather than restarting from the head of the list. This strategy avoids unnecessary repeated traversals from the beginning of the linked list. \inFullVersion{Pseudocode is given in the Supplementary Material.} \inConfVersion{Pseudocode is given in the full paper version~\cite{full-version}.}

\subsection{Unchained Shredding}
\label{sec:unchained}

The \emph{unchained shredded representation} (\usr) of a nested relation
$N\colon X$ differs from \csr in that it 
stores the offsets, and the prefix vector, of tuples in the same nested
attribute value consecutively rather than chaining them together. This allows binary rather than linear search during random access, at the expense of more complicated index construction.

To that end, a \usr uses the following extra columns. Every $\store{}(Y)$ has a column
$\prefix{}$. Additionally, for every nested subscheme $Z \in Y$, $\store{}(Y)$
has three extra columns, $\start{Z}$, $\llen{Z}$, and $\weight{Z}$. Finally, if
$Y \not = X$ then $\store{Y}$ also has a column $\perm$. To illustrate,
Figure~\ref{fig:UCSR} shows a \usr for the nested relation in
Fig~\ref{fig:R-semijoin-S-semijoin-T}.

Unchained shredding works as follows: every tuple $t \in N$ is represented by
exactly one tuple in $\store{}(X)$. Let $i$ be the representations' offset.
Then $t(x) = \store{}(X).x[i]$ for every flat $x \in X$. For every nested
attribute $Y \in X$ we have that $\weightval(t(Y)) =
\store{}(X).\weight{Y}[i]$. Moreover, if $j = \store{}(X).\start{Y}[i]$ and
$l = \store{}(X).\llen{Y}[i]$, then $\store{}(Y).\perm[j:l]$ contains the
offsets of all the tuples in $\store{}(Y)$ that together represent the tuples
occurring in $t(Y)$. Here, $\perm[j:l]$ denotes the slice of $\perm$ starting at
$j$ and ending at $j+l-1$. Finally, $\store{}(Y).\prefix{}[j:l]$ is the prefix
sum of the weights of those tuples, see Figure~\ref{fig:UCSR}.  For $X$ itself,
$\store{R}(X).\prefix{}$ contains the prefix sum of all tuples.

\smallskip\noindent\textbf{Discussion.} A \usr is an
implementation, in the shredded framework, of the random-access index proposed
in \cite{DBLP:journals/tods/CarmeliZBCKS22}, with the following difference. A
\usr uses the slice columns $\start{Z}, \llen{Z}$ to index into a permutation
vector $\perm$ to encode the offsets of elements of inner nested attributes
$Z$. By contrast, Carmeli et al. leave open how these elements are to be
represented exactly, suggesting implementation wise to use an additional index
data structure on $Z$ and probing this index during random access. This is
wasteful since when building the \usr we already have the correct positional
information, avoiding extra index space and associated costs.

\smallskip\noindent\textbf{Building the representation.} Similar to the chained
case, a \usr representing $R \nsemijoin S$ can be built in linear time given
\usrs $\store{R}$ and $\store{S}$ for $R\colon X_R$ and $S\colon X_S$,
respectively. The procedure is entirely similar to the building of a $\csr$ for
$R\colon X_R$ and $S\colon X_S$ described earlier, but differs in the algorithm
used to group $S$ by the join key. We hence focus next only on the grouping
aspect.

Specifically, to group $S$ on the join attributes $\z$, we now construct a triple $(h, \perm, \prefix{})$ consisting of (i) a hash table $h$ that maps the join keys occurring in $S$ to triples of the form $(i, l, w)$, (ii) a position vector $\perm$, and (iii) a prefix vector $\prefix{}$. For every join key $k$, $\perm[i:l]$ contains the positions of all tuples having join key $k$, and $\prefix{}[i:l]$ is the prefix sum of their weights, in the order that the tuples are specified in $\perm[i:l]$. The triple $(h, \perm, \prefix{})$ is then used in the probing phase to construct the necessary columns of the \usr representation. 

Compared to Algorithm~\ref{algo:csr-grouping}, which describes the \csr
grouping, \usr grouping needs to make two hashing passes over $S$ to construct
$(h, \perm, \prefix{})$. In the first pass we count, per join key $k$, the number
of tuples $l$ with that join key. Using this information, we can determine for
each join key the pair $(i,l)$ so that the positions of all tuples with join key
$k$ can be stored in the slice $\perm[i:l]$. A second hashing pass through the
data then actually stores the positions of the tuples with join key $k$ in this
location, and computes the prefix vector.

Given that \usr requires an extra hashing pass, \usr building is expected to be slower than
\csr building even though they share the same asymptotic complexity. We return to this in 
Section~\ref{subsec:unchained-experiments}.

\newcommand{\usraccess}{{\sc usr-get}}
\newcommand{\usrsubaccess}{{\sc usr-sub}}
\begin{figure}
          \tt \small
          \begin{algorithmic}[1]
\Function{\usraccess}{\store, i}\algorithmicdo
\State result = $\{\}$ \Comment{result tuple, initially empty}
            \State \perm\  = [1, 2, ..., \card{\store{}(X)}]
            \State \usrsubaccess(\store, X, i, \store{}(X).\prefix{}, \perm)
            \State return result
            \EndFunction
            \smallskip

            \Function{\usrsubaccess}{\store, $X$, i,\prefix{}, \perm}\algorithmicdo
            \State find smallest $j$ in $0..\len{\prefix{}}$ s.t. $i < 
\prefix{}[j]$
            \State $i = i - \prefix{}[j-1]$ \Comment{assume \prefix{}[-1] = 0}
            \State j = \perm[j]
            \State result = result $\uplus  \{ a \mapsto \store{}(X).a[j] \mid a \in X\}$
            \For{each nested attr $Y \in X$}
            \State $w = $ weight\_of($\store, Y, j$); $i' = i \imod w$; $i =  i \idiv w$
            \State $(s, l) = \store{}(X)[j](\start\_{Y},\llen\_{Y})$
            \State \usrsubaccess$(\store, Y, i', \store{}(Y).\prefix{}[s:l], \store{}(Y).\perm[s:l])$
            \EndFor
            \EndFunction
          \end{algorithmic}
\caption{\usr Access}
\label{algo:usr-access}
\end{figure}

\smallskip\noindent\textbf{Random access.} The fact that the positions and prefix vector are stored consecutively per join key aids in random access complexity. The \usr random access procedure, shown in Algorithm~\ref{algo:usr-access} is similar to the \csr access procedure. However, the fact that each nested attribute (and not only the top-most scheme $X$) has a prefix vector allows us to perform binary search at every level, obviating the need to iterate linearly through the nested sets. The overall complexity is therefore $\bigo(\log \size{\store{}(X)} + h \times \log(d))$ where $d$ is the maximum cardinality of any nested set and $h = \card{\rsub(X)}$ is the number of times that \usrsubaccess \xspace is called.

\begin{proposition}
  Using nested semijoins it is possible to construct in $\bigo(\size{\db})$ time
  a \usr that can be used as a random-access index for acyclic join query
  $\full{Q}$, with access time $\bigo(\log \size{\db}))$.
\end{proposition}

\smallskip\noindent\textbf{Optimization.} 
The same optimization used for \csr \emph{bulk} random access can also be applied for \usr. Specifically, for every nested attribute $Y$, we cache the result of the binary search so that, if the subsequent binary search for $Y$ is performed over the same search space, the search can resume from the point where the previous one has ended. \inFullVersion{Pseudocode is given in the Supplementary Material.} \inConfVersion{Pseudocode is given in the full paper version~\cite{full-version}.}

\section{Position sampling}
\label{sec:position-sampling}

We next describe how to construct $\pos = [i_1, \dots, i_k]$, the sequence of
tuple offsets 
that defines the output sample, and which we use to probe to the random-access
index.  To that end, we assume given a nested relation $N$ and a shredded
random-access index $\store{}$ for $N$. Throughout the section, let
$n = \card{\flatten(N)}$. Note that we can obtain $n$ from $\store{}$ in
constant time as it is the last element of the $\prefix{}$ vector of
$\store{}(X)$ in both \csr and \usr. To get insight into the position sampling
process, we first discuss how to draw $\pos$ in the simplified setting when all
tuples have the same probability. Using this insight, we turn to the general
setting where tuples have non-uniform probabilities.

\smallskip\noindent\textbf{Uniform.}  Assume that all tuples in $\flatten(N)$
have the same fixed probability $p \in [0,1]$. Then the sample space $I$ to draw
$\pos$ from is $I = \{0, \dots, n-1\}$. The \emph{uniform position sampling
  problem for $p$ and $n$} asks to construct $\pos \subseteq I$ such that each
element $i \in I$ is included in $\pos$ with probability $p$.
Observe that the expected size of $\pos$ follows a Binomial distribution with
parameters $n$ and $p$, and hence equals $np$.

The naive approach is to perform, for each $i \in I$, an independent Bernoulli trial with probability $p$, and to include $i$ in $\pos$ whenever the trial succeeds. This method, which we call $\idxnaive$, requires $\Theta(n)$ time, irrespective of $p$ and irrespective of the produced position vector length $k$. Consequently, $\idxnaive$ is suboptimal when $p$ is small, since it still performs a Bernoulli trial for every flat output tuple, even though the expected sample size is only $np$. 

Alternatively we can also construct $\pos$ by repeatedly sampling from the geometric distribution. Figure~\ref{algo:geometric} shows the pseudocode of this approach, denoted $\geom$.
Let $X$ be a stochastic variable that follows a geometric distribution with probability $p$, denoted $X \sim \mathrm{Geometric}(p)$. The value of $X$ represents the number of independent Bernoulli($p$) trials that result in failure before the first success occurs, i.e, $Pr[X = i] = p (1-p)^i$ for $i = 0, 1, \dots$.
Function \drawgeom shows how to draw from the  geometric distribution in constant time. Intuitively, we sample from the inverse of the exponential distribution and apply truncation to map the exponential distribution to the geometric distribution, see also \cite{DBLP:books/sp/Devroye86}. In $\geom$, successive calls to \drawgeom \xspace determine the gaps between sampled positions until the running index exceeds $n$. 
The complexity of \geom is $\bigo(k)$ with $k$ the size of the sample; thus the expected complexity of $\geom$ is $\bigo(np)$. We hence expect $\geom$ to outperform $\idxnaive$ for smaller sampling probabilities. However, we will 
experimentally show in Section~\ref{subsec:uniform-experiments} that for larger probabilities, $\geom$ is actually slower than $\idxnaive$. For this reason, we also define the hybrid method $\hybrid$ which uses $\geom$ when $p$ is smaller than a fixed threshold and $\idxnaive$ otherwise. We experimentally determine the threshold to be $p = 0.5$ in Section~\ref{subsec:uniform-experiments}.

  Finally, it is also possible to construct $\pos$ by first drawing $k$ from $\mathrm{Binomial}(n, p)$, the binomial distribution with parameters $n$ and $p$, and subsequently drawing a subset of size $k$ from $\{ 0 , \dots , n-1 \}$.
  The first step is $\bigo(n)$ in the worst case, but has expected runtime  $\bigo(n \times \min(p, 1-p))$~\cite{DBLP:journals/cacm/KachitvichyanukulS88}. The second step can be done in $\bigo(k)$  time~\cite{DBLP:journals/tods/CarmeliZBCKS22}. We denote this approach $\bbinom$ in what follows. 

\begin{figure}
\begin{algorithmic}[1]
  \small
  \Function{\drawgeom}{$p$}\algorithmicdo
  \State $u \sim \text{Uniform}(0,1)$
  \State \Return $\lfloor \ln(u) / \ln(1 - p) \rfloor$ 
\EndFunction
\smallskip
\Function{$\geom$}{$p, n$}\algorithmicdo
\If{$p = 0$ \Return $[]$}
\ElsIf{$p = 1$ \Return $[0, 1, 2, \dots, n-1]$}
\Else\algorithmicdo
  \State $\textsf{result} = []$; $i = \Call{DrawGeom}{p}$
  \While{$i < n$}
    \State $\textsf{result} =  \textsf{result} \textsf{++} [i]$ \Comment{append $i$ to the vector $\textsf{result}$}
    \State $j = \Call{DrawGeom}{p}$; $i = i + 1 + j$
  \EndWhile
\State \Return $\textsf{result}$
\EndIf
\EndFunction
\end{algorithmic}
\caption{$\geom$ algorithm}
\label{algo:geometric}
\end{figure}

\smallskip\noindent\textbf{Non-uniform.}
We next discuss the non-uniform case, assuming that the attribute $y$ that contains the probability that each tuple should be sampled with is a flat attribute of $N$. We may do so without loss of generality by Proposition~\ref{prop:probability-attribute-at-root}.

Every nested tuple $t$ of $N$ specifies in attribute $y$ the sampling
probability of all tuples in $\flatten(t)$, of which there are $\weightval(t)$
in total.  Hence, conceptually we can reduce the construction of $\pos$ in the
non-uniform case to $\card{N}$ uniform position sampling problems---one for each
$t \in N$.  It suffices to iterate over the nested tuples $t \in N$,
constructing a position vector $\pos_t$ for the uniform position sampling
problem with probability $t(y)$ and length $\weightval(t)$, and concatenating
these vectors---making sure to update offsets to take into account the positions
sampled for previous nested tuples. Implementation-wise, we simply iterate over the tuples of $\store{}(X)$.

\begin{sloppypar}
In what follows, when $M$ is one of the uniform position sampling methods, i.e.,
$M \in \{\idxnaive, \geom, \hybrid, \bbinom\}$, then we denote by $\pertuple{M}$ this
non-uniform position sampling method.
\end{sloppypar}

\smallskip\noindent\textbf{Asymptotic complexity of Poisson Sampling.} We close this section by establishing the asymptotic complexity of Poisson sampling over acyclic joins.

\begin{theorem}
  \label{thm:poisson-sampling-complexity}
  Poisson sampling over acyclic joins, as well as over free-connex projections of such joins, can be solved in time
  $\bigo(\insize + k \log \insize)$ in data complexity, where $\insize$ is the size
  of the input database, and $k$ the size of the resulting sample.
\end{theorem}
\begin{proof}
  Let $Q = \beta_y\left(R_1(\x_1) \Join \dots \Join R_m(\x_m)\right)$ be an
  acyclic Poisson sampling query. By
  Proposition~\ref{prop:probability-attribute-at-root} we can compute a \twonsa
  expression $\flatten(E)$ equivalent to $\full{Q}$, such that $y$ is a flat
  attribute of $E$.

  Given $E$, we can compute in $\bigo(\insize)$ time a \usr $\store{}$ of the
  nested relation $N \defeq E(\db)$, which supports single-tuple random access
  into $\flatten(N) = \full{Q}(\db)$ in $\bigo(\log \insize)$ time. Given
  $\store{}$ we can use $\pertuple{\geom}$ to compute a probe sequence $\pos$ of
  length $k$ in $\bigo(\card{N} + k) = \bigo(\insize + k)$ time, and use this to
  produce the resulting sample in $\bigo(k \log \insize)$ time by probing the
  index. Summing up the complexity of each individual step yields the claimed
  complexity.
  
  For Poisson sampling queries of the form $\beta_y(\delta \pi_A(\hat{Q}))$ with
  $\pi_A(\hat{Q})$ free-connex acyclic we reason as follows. Carmeli et
  al.~\cite{DBLP:journals/tods/CarmeliZBCKS22} observe that for every
  free-connex acyclic conjunctive query $\delta \pi_A(\hat{Q})$ and database $D$, one can
  compute in linear time a full acyclic join query $Q'$ and a database $D'$ such
  that $\delta \pi_A(\hat{Q})(D) = Q'(D')$. In other words, processing
  $\beta(\delta \pi_A(\hat{Q}))$ on $D$ is equivalent to processing $\beta(Q')$ on $D'$
  and our techniques apply to the latter. Since there is only a linear time
  overhead to obtain the latter from the former, our complexity results transfer
  to sampling over free-connex queries where the projection is set-based. For bag-based projection, our complexity results trivially hold because $\beta_y(\pi_A(\hat{Q})) = \pi_A(\beta_y(\hat{Q}))$.
\end{proof}

\section{Experimental Evaluation}
\label{sec:experiments}

\noindent\textbf{Implementation.}
We start from the Shredded Yannakakis (\sya) implementation of~\cite{DBLP:journals/pvldb/BekkersNVW25} which is implemented in Apache DataFusion, a Rust-based in-memory columnar query engine. We extend this implementation by (1) adding the \csr access method to \sya's existing \csr implementation; (2) adding the unchained \usr variant; and (3) adding the position sampling methods. All code and experiments are available in our reproducibility package~\cite{code}.

\medskip\noindent\textbf{Benchmarks.}  We evaluate our methods on two sets of
sampling queries: (1) uniform sampling queries of the form $\bern_p(\full{Q})$
where $p$ is a \emph{fixed uniform} probability $p \in [0,1]$; and (2) true
non-uniform Poisson sampling queries $\bern_y(\full{Q})$. The former allows to
gain insight into the trade-offs  across different fixed
probabilities $p$; while the latter allows to confirm these insights for Poisson sampling.

The full join queries $\full{Q}$ are drawn from the join order benchmark
\job~\cite{DBLP:journals/pvldb/LeisGMBK015} and
\stats~\cite{DBLP:journals/pvldb/HanWWZYTZCQPQZL21} which contain synthetic
queries evaluated on real-world data. For Poisson sampling, we additionally add
the real-world benchmark query $\contactquery$ from
Example~\ref{ex:epiql-contacts}, which simulates contact patterns in infectious
disease transmission scenarios and which is executed on real-world contact
probability data~\cite{Hoang2021ContactFlanders,Willem2012WeatherContacts}
modeling the population of Belgium, consisting of $1.1 \times 10^{7}$
individuals.

Both \job and \stats consist of SQL queries that perform base table filters and
equijoins, followed by a single aggregation. We omit the final aggregation step
to obtain the full join queries, all of which are acyclic.  Since Datafusion's
standard binary join method will serve as one of our baselines, we remove
queries for which the obtained full join fails to execute in Datafusion, leaving 112 \job queries and 142 \stats queries.

All of these queries are used for the uniform sampling experiments. For  Poisson sampling, however, we also need an attribute $y$ that specifies the sampling probability. Because neither \job nor \stats contain a probability attribute, we generate such an attribute as follows. 
Relation \jobtitle is central in the query graph of the \job workload~\cite{DBLP:journals/pvldb/LeisGMBK015}, and this relation also appears in all \job queries. Therefore, we add probability attribute $y$ to \jobtitle. For \stats, no single relation appears in all queries. We therefore attach probability attribute $y$ to the relation that occurs most frequently: \users,  present in 113 out of 142 queries. We populate attribute $y$ in both \jobtitle and \users according to three different distributions, denoted \low, \medium, and \high respectively:
\begin{itemize}%
    \item $\text{Beta}(a=2,b=10)$ with $\mathbb{E}[y] \approx 0.167$ and $\mathrm{Var}[y] \approx 0.011$
    \item $\text{Normal}(\mu=0.5,\sigma=0.2)$ with $\mathbb{E}[y] = 0.5$ and $\mathrm{Var}[y] = 0.04$
    \item $\text{Beta}(a=10,b=2)$ with $\mathbb{E}[y] \approx 0.833$ and $\mathrm{Var}[y] \approx 0.011$
\end{itemize}
Table~\ref{tab:benchmarks} summarizes 
the number of queries in each benchmark as well as the distribution of full join sizes in each setup. (For $\contactquery$ the benchmark has only one query, hence the distribution has only one measurement.) 

\begin{table}
\centering
\setlength{\tabcolsep}{4pt}
\resizebox{.7\linewidth}{!}{%
\begin{tabular}{@{}llrrrrr@{}}
\toprule
\textbf{Benchmark} & \textbf{Setup} & \textbf{\#Queries} & \multicolumn{4}{c}{\textbf{Full join output size}} \\
\cmidrule(lr){4-7}
 &  &  & \textbf{min} & \textbf{avg} & \textbf{median} & \textbf{max} \\
\midrule
\multirow{1}{*}{\job}   
 & \begin{tabular}[c]{@{}l@{}}uniform \\  poisson\end{tabular} & 112 & 0 & $1.39\times10^{5}$ & $4.44\times10^{2}$ & $3.71\times10^{6}$ \\
\midrule
\multirow{2}{*}{\stats} 
 & uniform & 142 & $2.00\times10^{2}$ & $5.75\times10^{7}$ & $1.48\times10^{6}$ & $2.26\times10^{9}$ \\
                   & poisson & 113 & $2.00\times10^{2}$ & $3.98\times10^{7}$ & $1.46\times10^{6}$ & $5.93\times10^{8}$ \\
\midrule
$\contactquery$ & poisson & 1 & \multicolumn{4}{c}{$1.32 \times 10^{10}$} \\  
\bottomrule
\end{tabular}
}
\caption{Overview of benchmark queries.}
\label{tab:benchmarks}
\vspace{-5ex}
\end{table}

\medskip\noindent\textbf{Setup.}
All experiments are conducted on a Ubuntu 22.04.4 LTS machine using a single thread with an Intel Core i7-11800 CPU and 32GB of RAM. All runtimes are averages of 5 different runs.
We apply caching optimization, as described in Section~\ref{sec:indexing}, when probing a chained index because it consistently improves performance on our benchmarks. In contrast, for unchained index probing, we found that the overhead of caching may slightly outweigh its benefits, as the search spaces are extremely small (see Section~\ref{subsec:unchained-experiments}). We therefore disabled unchained caching, unless explicitly specified otherwise. \inFullVersion{A detailed experimental analysis of caching is provided in the Supplementary Material.}\inConfVersion{A detailed experimental analysis of caching is provided in the full paper version~\cite{full-version}.}

\medskip\noindent\textbf{Baseline.} We consider three possible \mas implementations as potential baselines, differing in how they materialize the full join:
\begin{itemize}
\item \binsampling :  using a sequence of binary
  hashjoins.
    \item \syasamplingC :  using \csr-based \sya, i.e., by means of \csr-based nested semijoins followed by a flatten.
    \item \syasamplingU : the \usr variant of the previous method.
\end{itemize}
A per-tuple Bernoulli trial is always performed to create the sample.

Prefix M indicates that these methods fully materialize the join. 
Since Apache Datafusion lacks a join order optimizer, we use DuckDB to generate binary join plans for \binsampling and then apply the cost-based rewriting technique of \cite{DBLP:journals/pvldb/BekkersNVW25} to obtain the corresponding nested semijoin + flatten plans that are used in the other methods. Bekkers et al~\cite{DBLP:journals/pvldb/BekkersNVW25} have shown that computing full joins using chained \sya is instance-optimal, and therefore faster and more robust than using binary joins. On  our benchmarks, we observe that \syasamplingC is on average $91.4$ms faster than \binsampling for \job and $690$ms for \stats. 
This hence confirms the findings of~\cite{DBLP:journals/pvldb/BekkersNVW25}. Furthermore, we find the \syasamplingC  is faster than \syasamplingU on \job ($21.4$ms on avg.), while being competitive for \stats (median runtimes differ only $0.09$ms). A more detailed comparison of chained and unchained variants is deferred to Section~\ref{subsec:unchained-experiments}. Given that \syasamplingC is significantly faster than \binsampling and faster or competitive with \syasamplingU, we fix it as the baseline to compare $\iap$ against in what follows.

\subsection{Uniform Sampling}
\label{subsec:uniform-experiments}

\medskip\noindent\textbf{Position sampling.} 
We first analyze the efficiency of the uniform position sampling methods $\idxnaive$, $\geom$, and $\bbinom$ described in Section~\ref{sec:position-sampling}. 
Their expected runtimes are $\bigo(n)$, $\bigo(np)$, and $\bigo(n \times \min (p, 1-p) + np)$, respectively. 
$\idxnaive$ should be outperformed by $\geom$ and $\bbinom$ for smaller sampling probabilities.
Figure~\ref{fig:uniform-sampling-positions} confirms this hypothesis: the smaller $p$, the greater the improvement of $\geom$ and $\bbinom$ over $\idxnaive$. However, for larger values of $p$, $\idxnaive$ becomes consistently faster. 
$\bbinom$ follows the same trend as $\geom$, but is slightly slower due to higher constant-factor overhead. We hence discard $\bbinom$ from further discussion and experiments.

A closer look reveals that, although the expected runtime of $\idxnaive$ is independent of $p$, the observed runtime actually increases with $p$ up to approximately $p=0.5$, and decreases for larger values of $p$. We attribute this non-monotonic behavior to the effects of \emph{branch prediction}. $\idxnaive$ makes repeated random choices that depend on $p$. When $p \approx 0.5$, each branch is essentially unpredictable, causing frequent mispredictions, which slows execution. In contrast, when $p$ deviates from $0.5$, the branch outcome becomes more predictable, reducing misprediction rates and improving runtime.

For large sampling probabilities, $\geom$ requires nearly as many iterations as $\idxnaive$. 
As a $\geom$ iteration is more costly than performing a single lightweight Bernoulli trial, the constant overhead of $\geom$ dominates, making $\idxnaive$ more efficient for high values of $p$.

\begin{figure}[t]
    \centering
    \begin{tikzpicture}
        \begin{axis}[
            width=.7\columnwidth,
            height=6cm,
            ymode=log,
            xlabel={$p$},
            ylabel={Avg.\ position sampling time (ms)},
            y label style={yshift=-1em},
            label style={font=\small},
            tick label style={font=\scriptsize},
            legend style={
                font=\scriptsize,
                at={(0.5,1.05)},
                anchor=south,
                legend columns=3, 
                /tikz/every even column/.append style={column sep=0.5em},
                draw=none,
            },
            ymajorgrids=true,
            grid style=dashed,
            enlarge x limits=0.05,
            symbolic x coords={0.0001,0.001,0.01,0.1,0.2,0.3,0.4,0.5,0.6,0.7,0.8,0.9,0.99},
            xtick=data,
            xticklabels={$10^{-4}$,$10^{-3}$,$10^{-2}$,$10^{-1}$,$0.2$,$0.3$,$0.4$,$0.5$,$0.6$,$0.7$,$0.8$,$0.9$,$0.99$}, 
        ]
        \addplot+[blue, solid, thick, mark=oplus, mark options={solid, color=blue}] table [x=p, y=time_ms, col sep=comma] {./data/pos-sampling/job-Bern-pos-sampling.csv};
        \addlegendentry{$\idxnaive$ (\job)};  
        \addplot+[blue, solid, thick, mark=star, mark options={solid, color=blue}] table [x=p, y=time_ms, col sep=comma] {./data/pos-sampling/job-Geo-pos-sampling.csv};
        \addlegendentry{$\geom$ (\job)};
        \addplot+[blue, solid, thick, mark=triangle, mark options={solid, color=blue}] table [x=p, y=time_ms, col sep=comma] {./data/pos-sampling/job-Binomial-pos-sampling.csv};
        \addlegendentry{$\bbinom$ (\job)}; 
        \addplot+[red, solid, thick, mark=oplus, mark options={solid, color=red}] table [x=p, y=time_ms, col sep=comma] {./data/pos-sampling/stats-Bern-pos-sampling.csv};
        \addlegendentry{$\idxnaive$ (\stats)}
        \addplot+[red, solid, thick, mark=star, mark options={solid, color=red}] table [x=p, y=time_ms, col sep=comma] {./data/pos-sampling/stats-Geo-pos-sampling.csv};
        \addlegendentry{$\geom$ (\stats)};
        \addplot+[red, solid, thick, mark=triangle, mark options={solid, color=red}] table [x=p, y=time_ms, col sep=comma] {./data/pos-sampling/stats-Binomial-pos-sampling.csv};
        \addlegendentry{$\bbinom$ (\stats)}
        \end{axis}
    \end{tikzpicture}
    \caption{
      Position sampling efficiency in function of $p$.}
    \label{fig:uniform-sampling-positions}
\end{figure}

\medskip\noindent\textbf{End-to-end runtimes.}
We next compare the end-to-end runtimes of \iap methods for uniform sampling. 
We adopt the following naming scheme: for an index type $T \in \{\text{U(SR)}, \text{C(SR)}\}$ and uniform position sampling method $M \in \{\geom, \idxnaive\}$, let $\textsf{I}_T\textsf{-}{M}$ denote the \iap algorithm obtained by using $T$ and $M$ as index and position sampling method, respectively. 

Since $\geom$ outperforms $\idxnaive$ for $p \le 0.5$, we focus on
$\textsf{I}\textsf{-}{\geom}$ when $p \le 0.5$ and on
$\textsf{I}\textsf{-}{\idxnaive}$ otherwise.

Figure~\ref{fig:job-uniform-runtime-breakdown} shows that for \job, chained \iap is faster than the \syasampling baseline across all sampling probabilities $p$, whereas unchained \iap is consistently slower. The differences remain small: averaged over all queries, the relative speedup of chained \iap over \syasamplingC remains close to 1 for all $p$, reaching a maximum of $1.0056\times$ at $p=0.0001$. This is because the total runtime is dominated (82\%) by reading input relations from disk and base table filtering, which is the same for all methods. Moreover, as shown in Figure~\ref{fig:job-uniform-runtime-breakdown}, the remaining 18\% of the runtime is dominated by index construction, which is independent of the sampling probability. Hence the small difference between the three approaches across different values of $p$. Figure~\ref{fig:job-uniform-runtime-breakdown} also shows that chained sampling is faster than unchained  due to the dominating but faster index building phase. 

For \stats, the results in Figure~\ref{fig:statsceb-uniform-runtime-breakdown} show more pronounced differences. Chained \iap remains the fastest method up to $p \leq 0.8$ and achieves substantially larger speedups over the baseline than for \job, with a maximum speedup (averaged over all queries) of $38.79$x at $p=0.0001$.
Because \stats has larger full join sizes (see Table~\ref{tab:benchmarks}),
the overhead of materializing join tuples that are not part of the sampled output becomes more significant. 
This is confirmed in Figure~\ref{fig:statsceb-uniform-runtime-breakdown}, showing that position sampling and index probing (for $\iap$) as well as flattening all full join tuples (given an index) and performing Bernoulli trails (for $\mas$) contribute more to the total runtime compared to \job. Note that the index construction time is plotted as well, but so small that it is not visible. Moreover, chained $\iap$ is faster than the unchained variant due to the faster index probing phase. This is in contrast to \job, where chained $\iap$ was faster due to faster index construction.

For \stats, chained \iap becomes slightly slower than \syasamplingC for $p \ge 0.9$, which was not the case for \job.
To understand why, we note that the 
\syasamplingC baseline uses a flatten operation that is highly optimized for sequential memory access~\cite{DBLP:journals/pvldb/BekkersNVW25}. In contrast, the index-based methods cannot exploit such optimizations, as they must assume random access and incur additional overhead during $\access$
. 
These costs grow significant when the output is large---as in \stats, but not in \job---explaining why index-based methods become slower for \stats at high sampling probabilities while remaining faster for \job.

\begin{figure*}
    \centering
\begin{subfigure}{\textwidth}
    \begin{tikzpicture}
      \pgfplotstableread[col sep=comma]{./data/job-uniform-runtime-breakdown.csv}\datatable
      
      \begin{axis}[
          axis lines*=left, ymajorgrids,
          width=\linewidth, height=4cm,
          ymin=0,
          ybar stacked,
          bar width=5pt,
          xtick=data,
          xticklabels from table={\datatable}{method},
          xticklabel style={rotate=90,anchor=mid east,font=\footnotesize},
          ylabel={Avg.\ runtime (ms)},    
          legend style={
              at={(0.5,1)},
              anchor=south,
              draw=none,
          },
          legend columns=-1,
          /tikz/every even column/.append style={column sep=1.2em},
          draw group line={p}{0.0001}{$0.0001$}{-9ex}{4pt},
          draw group line={p}{0.001}{$0.001$}{-9ex}{4pt},
          draw group line={p}{0.01}{$0.01$}{-9ex}{4pt},
          draw group line={p}{0.1}{$0.1$}{-9ex}{4pt},
          draw group line={p}{0.2}{$0.2$}{-9ex}{4pt},
          draw group line={p}{0.3}{$0.3$}{-9ex}{4pt},
          draw group line={p}{0.4}{$0.4$}{-9ex}{4pt},
          draw group line={p}{0.5}{$0.5$}{-9ex}{4pt},
          draw group line={p}{0.6}{$0.6$}{-9ex}{4pt},
          draw group line={p}{0.7}{$0.7$}{-9ex}{4pt},
          draw group line={p}{0.8}{$0.8$}{-9ex}{4pt},
          draw group line={p}{0.9}{$0.9$}{-9ex}{4pt},
          draw group line={p}{0.99}{$0.99$}{-9ex}{4pt},
          draw group line={p}{1}{$1$}{-9ex}{4pt},
          after end axis/.append code={
              \path [anchor=base east, yshift=-10.5ex]
                  (rel axis cs:0,0) node  {$p$};
          }
      ]
      
      \addplot [
              fill=yellow!30,
              postaction={pattern=dots, pattern color=black}
          ] table [x=X, y=Index building] {\datatable};
          \addlegendentry{1. Index building}
          
          \addplot [
              fill=gray!25,
              postaction={pattern=grid, pattern color=black}
          ] table [x=X, y=Position sampling] {\datatable};
          \addlegendentry{2. Position sampling}
          
          \addplot [
              fill=green!30!white,
              postaction={pattern=horizontal lines, pattern color=black}
          ] table [x=X, y=Index probing] {\datatable};
          \addlegendentry{3. Index probing}
          
          \addplot [
              fill=red!25!white,
              postaction={pattern=north east lines, pattern color=black}
          ] table [x=X, y=Flatten+Bernoulli] {\datatable};
          \addlegendentry{4. Flatten+Bernoulli}
      
      \end{axis}
\end{tikzpicture}

    \caption{\job}
    \label{fig:job-uniform-runtime-breakdown}
\end{subfigure}

\medskip

\begin{subfigure}{\textwidth}
      \begin{tikzpicture}
        \pgfplotstableread[col sep=comma]{./data/stats-uniform-runtime-breakdown.csv}\datatable
        
        \begin{axis}[
            axis lines*=left, ymajorgrids,
            width=\linewidth, height=4cm,
            ymin=0,
            ybar stacked,
            bar width=5pt,
            xtick=data,
            xticklabels from table={\datatable}{method},
            xticklabel style={rotate=90,anchor=mid east,font=\footnotesize},
            ylabel={Avg.\ runtime (ms)},    
            legend style={
                at={(0.5,1)},
                anchor=south,
                draw=none,
            },
            legend columns=-1,
            /tikz/every even column/.append style={column sep=1.2em},
            draw group line={p}{0.0001}{$0.0001$}{-9ex}{4pt},
            draw group line={p}{0.001}{$0.001$}{-9ex}{4pt},
            draw group line={p}{0.01}{$0.01$}{-9ex}{4pt},
            draw group line={p}{0.1}{$0.1$}{-9ex}{4pt},
            draw group line={p}{0.2}{$0.2$}{-9ex}{4pt},
            draw group line={p}{0.3}{$0.3$}{-9ex}{4pt},
            draw group line={p}{0.4}{$0.4$}{-9ex}{4pt},
            draw group line={p}{0.5}{$0.5$}{-9ex}{4pt},
            draw group line={p}{0.6}{$0.6$}{-9ex}{4pt},
            draw group line={p}{0.7}{$0.7$}{-9ex}{4pt},
            draw group line={p}{0.8}{$0.8$}{-9ex}{4pt},
            draw group line={p}{0.9}{$0.9$}{-9ex}{4pt},
            draw group line={p}{0.99}{$0.99$}{-9ex}{4pt},
            draw group line={p}{1}{$1$}{-9ex}{4pt},
            after end axis/.append code={
                \path [anchor=base east, yshift=-10.5ex]
                    (rel axis cs:0,0) node  {$p$};
            }
        ]
        
        \addplot [
                fill=yellow!30,
                postaction={pattern=dots, pattern color=black}
            ] table [x=X, y=Index building] {\datatable};
            
            \addplot [
                fill=gray!25,
                postaction={pattern=grid, pattern color=black}
            ] table [x=X, y=Position sampling] {\datatable};
            
            \addplot [
                fill=green!30!white,
                postaction={pattern=horizontal lines, pattern color=black}
            ] table [x=X, y=Index probing] {\datatable};
            
            \addplot [
                fill=red!25!white,
                postaction={pattern=north east lines, pattern color=black}
            ] table [x=X, y=Flatten+Bernoulli] {\datatable};
        
        \end{axis}
\end{tikzpicture}

      \caption{\stats}
      \label{fig:statsceb-uniform-runtime-breakdown}
      \end{subfigure}
      \caption{A breakdown of the total runtime (excl.\ input reading and base table filtering) for each sampling probability $p$. Runtimes are averaged over all queries. For $\iap$ methods, the remaining execution time consists of (1) index building, (2) position sampling and (3) index probing. For $\syasamplingC$, the runtime is split into two phases: (1) index building, and (4) using the index to output the full join result (flatten) and performing a Bernoulli trial for each full join output tuple.}
    \end{figure*}

\smallskip\noindent\textbf{Conclusion.}
In summary, the choice of position sampling method should depend on the sampling probability. For small $p$, $\geom$ has clear advantages due to its lower expected number of iterations, whereas for large $p$, the simpler control flow of $\idxnaive$ results in better performance. When considering end-to-end runtimes,  chained \iap is preferable over unchained \iap, significantly outperforming full materialization for moderate and low sampling probabilities, confirming that avoiding full joins is particularly effective when only a fraction of the result is needed.

\subsection{Poisson Sampling}
\label{subsec:nonuniform-experiments}

We next turn to \iap methods for non-uniform Poisson sampling, adopting the following naming scheme: by $\textsf{I}_T\textsf{-}{M}$ we denote the \iap algorithm for Poisson sampling obtained by using $T \in  \{\text{U(SR)}, \text{C(SR)}\}$ and $M \in \{\pertuple{\geom}, \pertuple{\idxnaive}, \pertuple{\hybrid}\}$ as index and position sampling method, respectively. We omit the subscript when referring to both chained and unchained variants at the same time.
Note that the non-uniform position sampling methods require the probability attribute $y$ to appear as a flat attribute at the root of the shredded representation. We therefore rewrite the plans from Section~\ref{subsec:uniform-experiments} accordingly (see Proposition~\ref{thm:poisson-sampling-complexity}).

We observed in Section~\ref{subsec:uniform-experiments} that $\geom$ outperforms $\idxnaive$ when $p \le 0.5$. Based on this observation, we configure $\hybrid$ to use $\geom$ when $p \leq 0.5$, and $\idxnaive$ otherwise.

\begin{figure}
\includegraphics[width=.7\linewidth]{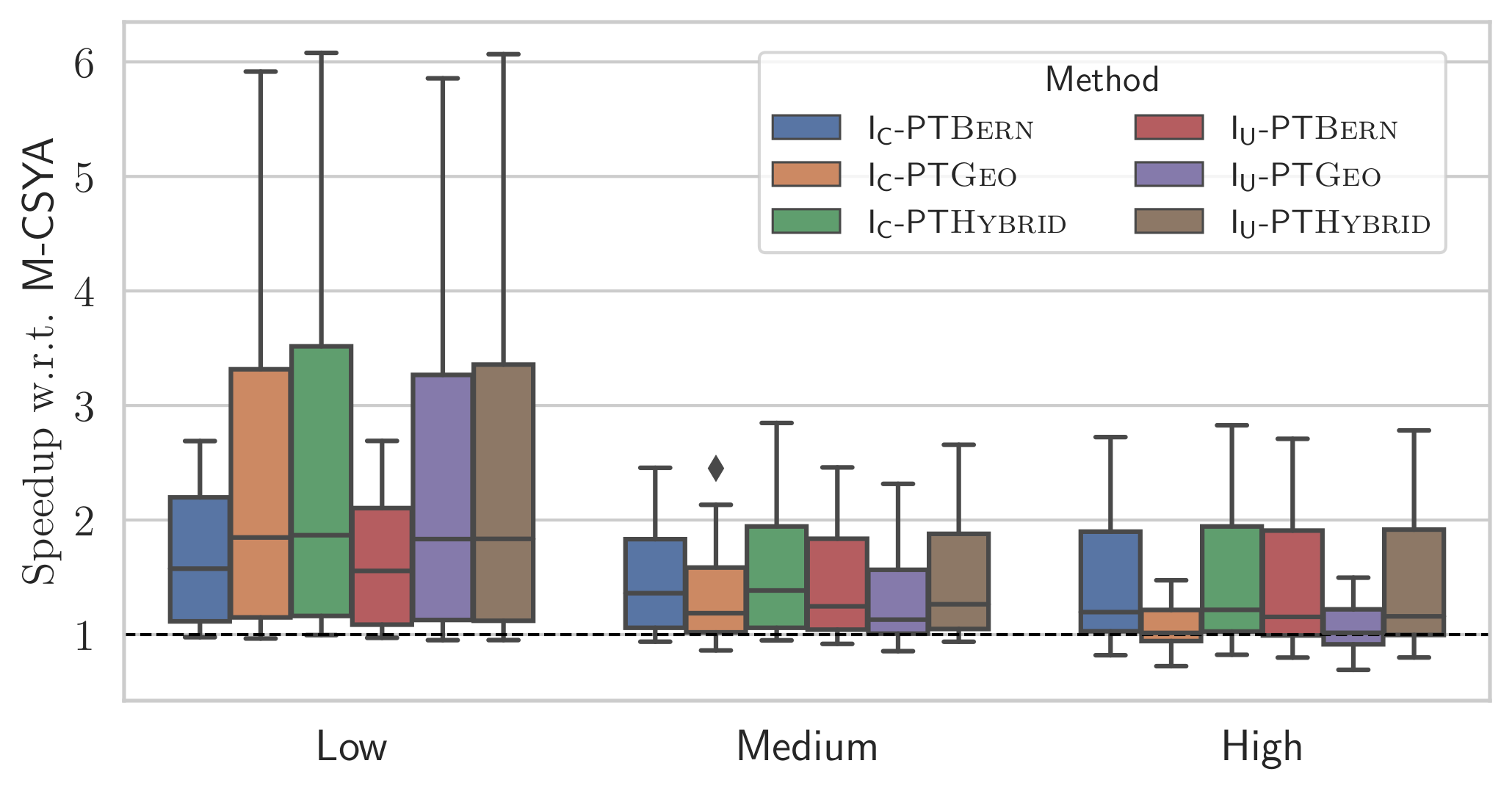}
\caption{Poisson sampling speedups w.r.t. \syasamplingC on 
  \stats for \low, \medium, and \high sampling probabilities.}
\label{fig:stats-nonuniform-speedups}
\vspace{-1ex}
\end{figure}

\smallskip\noindent\textbf{\stats.}
Figure~\ref{fig:stats-nonuniform-speedups} shows the relative end-to-end speedups obtained on \stats by the \iap methods  
compared to the \mas baseline \syasamplingC.
Overall, we observe that the chained 
variants perform slightly better than their unchained counterparts. 
For \medium and \high sampling probabilities, \nugeomsampling performs the worst, and only for \low probabilities it is competitive with \pertuplesampling. This aligns with the findings from Section~\ref{subsec:uniform-experiments}, where we observed that $\geom$ is slower than $\idxnaive$ (and hence $\hybrid$) when $p > 0.5$. 
\nunaivesampling becomes increasingly competitive with \pertuplesampling as the sampling probability grows, since both methods perform the same operations when $p>0.5$.

Overall, \pertuplesamplingC performs consistently best, when looking at the distribution of the observed speedups in both relative and absolute terms. These are (min/avg/max):
\begin{itemize}
    \item ($0.995$/$2.39$/$6.08$) x relative, and ($-4.7 \times 10^{-5}$/$0.4$/$5.54$) s absolute for \low, 
    \item ($0.95$/$1.54$/$2.85$) x  and  ($-0.0017$/$0.32$/$5.25$) s for \medium, and
    \item ($0.82$/$1.49$/$2.83$) x and  ($-0.007$/$0.28$/$5.72$) s for \high probabilities.
    \end{itemize}  
For comparison, \pertuplesamplingU \xspace achieves
(min/avg/max) speedups of
\begin{itemize}
    \item ($0.95$/$2.34$/$6.07$) x and ($-0.0006$/$0.40$/$5.53$) s for \low, 
    \item ($0.94$/$1.49$/$2.66$) x and ($-0.0021$/$0.29$/$5.12$) s for \medium, and
    \item ($0.80$/$1.45$/$2.78$) x and ($-0.8021$/$0.24$/$5.65$) s for \high probabilities.
    \end{itemize}
We conclude that on \stats, the chained and unchained variants are competitive, with chained slightly more robust w.r.t. absolute regressions for \high probabilities.

\smallskip\noindent\textbf{\job.}
\begin{table*}

\begin{subtable}{\linewidth}
  \centering
  \resizebox{.8\linewidth}{!}{%
  \begin{tabular}{@{}lrrrrrrrrrrrr@{}}
    \toprule
           & \multicolumn{3}{c}{min}                                       && \multicolumn{3}{c}{avg}                                  && \multicolumn{3}{c}{max}               \\  
           \cmidrule{2-4} \cmidrule{6-8}  \cmidrule{10-12} 
          $p$ & \geom           & \idxnaive          & \hybrid && \geom         & \idxnaive & \hybrid       && \geom  & \idxnaive         & \hybrid      \\ 
    \midrule
    \low    & \textbf{-13.0} & -13.3          & \multicolumn{1}{l}{-13.3}  && \textbf{7.3} & 6.4   & \multicolumn{1}{l}{7.2}          && 102.8 & 102.0         & \textbf{103.8}   \\
    \medium & -13.1          & \textbf{-8.2}  & \multicolumn{1}{l}{-12.6}  && 4.6          & 4.3   & \multicolumn{1}{l}{\textbf{4.8}} && 93.5  & \textbf{99.8} & 97.4             \\
    \high   & -28.4          & \textbf{-13.3} & \multicolumn{1}{l}{-14.6}  && 1.8          & 2.5   & \multicolumn{1}{l}{\textbf{2.6}} && 97.8  & 98.2          & \textbf{98.9}     \\ 
    \bottomrule
    \end{tabular}%
  }
  \caption{Chained}
  \label{tab:job-nonuniform-speedups-chained}
\end{subtable}

\begin{subtable}{\linewidth}
  \centering
  \resizebox{.8\linewidth}{!}{%
  \begin{tabular}{@{}lrrrrrrrrrrrr@{}}
    \toprule
          & \multicolumn{3}{c}{min}                         && \multicolumn{3}{c}{avg}                      && \multicolumn{3}{c}{max} \\  
           \cmidrule{2-4} \cmidrule{6-8}  \cmidrule{10-12} 
          $p$ & \geom    & \idxnaive   & \hybrid  && \geom   & \idxnaive  & \hybrid && \geom  & \idxnaive  & \hybrid  \\ 
    \midrule
    \low    & -2207.7 & -2218.3 & \multicolumn{1}{l}{-2222.2} && -108.7 & -109.5 & \multicolumn{1}{l}{-109.4} && 33.5  & 37.4   & 33.6    \\
    \medium & -2254.5 & -2253.1 & \multicolumn{1}{l}{-2258.1} && -112.1 & -112.0 & \multicolumn{1}{l}{-112.0} && 38.1  & 37.2   & 38.5    \\
    \high   & -2283.4 & -2274.2 & \multicolumn{1}{l}{-2276.9} && -115.3 & -114.3 & \multicolumn{1}{l}{-114.2} && 45.0  & 44.9   & 44.3    \\ 
    \bottomrule
    \end{tabular}%
  }
  \caption{Unchained}
  \label{tab:job-nonuniform-speedups-unchained}
\end{subtable}

\caption{Minimum, average, and maximum absolute Poisson sampling speedups (ms) on \job.}
\label{tab:job-nonuniform-speedups}

\end{table*}
Table~\ref{tab:job-nonuniform-speedups} reports on the distribution of the \emph{absolute} end-to-end time differences obtained on \job compared to the \syasamplingC baseline. Positive numbers are improvements; negative numbers are regressions. Best-performing method is highlighted in bold. 

We observe that chained methods behave more robust than their unchained counterparts: the largest regressions (in the min column) are significantly smaller for chained compared to unchained (100 times or more). Furthermore, on average (in the avg column) queries exhibit a slight chained speedup but still an unchained regression. Finally, the largest chained speedups are approximately three times that of the unchained speedups.

Overall, the speedups and regressions exhibited by chained methods are only
marginal in terms of end-to-end runtime---on the order of only a few ms on
average to a hundred ms at most, while total runtimes are in the range of
seconds. By contrast the maximum regression exhibited by the unchained methods
are on the same order (seconds), hence significant, while speedups are
outperformed by the chained methods.  

As far as the choice of  position sampling method is concerned we see that there is little absolute difference between the three methods for the same index type. While $\pertuplesampling$ is not always the best method in absolute terms, it is always competitive with the best method---differing only a few ms. This is consistent with our earlier observations.

\smallskip\noindent\textbf{\epiql.}
Figure~\ref{fig:epiql-total-runtimes} reports total runtimes on $\contactquery$ for varying population sizes, where each size corresponds to a subset of the Belgian population. Method \binsampling runs out of memory for the full population size. In contrast, the instance-optimality of \syafull allowed us to materialize the full join output for 11M people within available memory. For a population size of 1M individuals, \syasamplingC is $4.2$x (or $130.1$s) faster than \binsampling. 

Among \iap methods, the unchained variants are consistently outperformed by
their chained counterparts.  The best-performing \iap method across all
population sizes is \pertuplesamplingC, which achieves an end-to-end runtime
improvement of $415.6$s ($5.3$x) compared to the \syasamplingC baseline for 11M
individuals.
For comparison, \pertuplesamplingU \xspace is only $382.4$s ($3.9$x) faster than the baseline.
These results demonstrate that avoiding full join materialization
can be highly effective---especially in real-world Monte Carlo simulations where
such queries must be executed repeatedly and efficiency gains accumulate.

We next compare the runtimes of the position sampling methods in
isolation. $\hybrid$ is the most efficient, closely followed by $\geom$, while
$\idxnaive$ performs worst. 
For a population size of 11M individuals, position sampling using $\geom$ is 34
s (or $3.4$x) faster than $\idxnaive$, and $\hybrid$ is even 36 s (or $3.7$x)
faster than $\idxnaive$.  The explanation is twofold. First, the probability of
two people having contact is in practice very low. In particular, the average
sampling probability is only $2.4\%$.
Second, 
on 11M people the output of $\contactquery$ is $1.3 \times 10^{10}$. This large
full join size in combination with the small sampling probabilities accounts for
the observed differences between position sampling methods.

\begin{figure}
  \centering
\begin{tikzpicture}
  \begin{loglogaxis}
    [
      xlabel={Population size},
      ylabel={Total runtime (s)},
      legend style={
        at={(1.02,0.5)},
        anchor=west,
        draw=none,
      },
      clip=false,
    ]
    \addplot+[solid] table [x=popsize_int, y=duration(s), col sep=comma] {./data/epiql-nonuniform/BinaryJoin+Bernoulli_runtimes.csv};
    \addlegendentry{$\binsampling$};  
    \addplot+[solid] table [x=popsize_int, y=duration(s), col sep=comma] {./data/epiql-nonuniform/SYA+Flatten+Bernoulli_runtimes.csv};
    \addlegendentry{$\syasamplingC$};  
    \addplot+[solid] table [x=popsize_int, y=duration(s), col sep=comma] {./data/epiql-nonuniform/USYA+Flatten+Bernoulli_runtimes.csv};
    \addlegendentry{$\syasamplingU$};  
    \addplot+[solid] table [x=popsize_int, y=duration(s), col sep=comma] {./data/epiql-nonuniform/IdxGeom-Chained_runtimes.csv};
    \addlegendentry{$\geomsamplingC$};  
    \addplot+[solid] table [x=popsize_int, y=duration(s), col sep=comma] {./data/epiql-nonuniform/IdxGeom-Unchained_runtimes.csv};
    \addlegendentry{$\geomsamplingU$};  
    \addplot+[solid] table [x=popsize_int, y=duration(s), col sep=comma] {./data/epiql-nonuniform/IdxNaive-Chained_runtimes.csv};
    \addlegendentry{$\naivesamplingC$};  
    \addplot+[solid] table [x=popsize_int, y=duration(s), col sep=comma] {./data/epiql-nonuniform/IdxNaive-Unchained_runtimes.csv};
    \addlegendentry{$\naivesamplingU$};  
    \addplot+[solid] table [x=popsize_int, y=duration(s), col sep=comma] {./data/epiql-nonuniform/PerNestedTuple-Chained_runtimes.csv};
    \addlegendentry{$\pertuplesamplingC$};  
    \addplot+[solid] table [x=popsize_int, y=duration(s), col sep=comma] {./data/epiql-nonuniform/PerNestedTuple-Unchained_runtimes.csv};
    \addlegendentry{$\pertuplesamplingU$};  
  \end{loglogaxis}
\end{tikzpicture}
\caption{Total runtime on $\contactquery$.\vspace{-1ex}}
\label{fig:epiql-total-runtimes}
\end{figure}

\smallskip\noindent\textbf{Conclusion.}
Across all benchmarks, \pertuplesamplingC consistently achieves the best performance among the Poisson sampling methods even though its asymptotic complexity is worse than that of \pertuplesamplingU. 

\begin{table*}[t]
        \centering

            \setlength{\tabcolsep}{3pt}
            \resizebox{\linewidth}{!}{%
            \begin{tabular}{@{}lllllllllrrrrrrrr@{}} \toprule
                        &  \multicolumn{6}{c}{$d$}   & & \multicolumn{2}{c}{average}                       &           & \multicolumn{2}{c}{median}                    &          & \multicolumn{2}{c}{maximum}                   \\
                        \cmidrule{2-7}  \cmidrule{9-10}  \cmidrule{12-13}  \cmidrule{15-16}
                Benchmark             & min & 25\% & 50\% & 75\% & 95\% & max  & & chained                  & unchained              &           & chained               & unchained             &          & chained                   & unchained         \\
                \midrule
                \job (uniform)        & 1 & 1 & 2 & 4 & 12 & 342381          & & \textbf{4.30}            & 4.76                   &           & 0.54                  & \textbf{0.51}         &          & \textbf{168.28}           & 193.40            \\
                \job (poisson)        & 1 & 2 & 2 & 6 & 9 & 15              & & 6.53                     & \textbf{6.07}          &           & \textbf{0.46}         & \textbf{0.46}         &          & \textbf{207.13}           & 214.53            \\ 
                \stats (uniform)      & 1 & 3 & 5 & 7 & 21 & 55              & & \textbf{202.57}          & 226.31                 &           & \textbf{1.90}         & \textbf{1.90}              &          & \textbf{18866.44}         & 20584.97          \\
                \stats (poisson)      & 2 & 4 & 11 & 15 & 59 & 137           & & \textbf{129.91}          & 149.76                 &           & \textbf{1.90}          & 2.14                  &          & \textbf{4902.71}          & 6222.82           \\ 
                \hline
                $\contactquery$ (poisson)    &             \multicolumn{6}{c}{$d=50$}  & & \multicolumn{8}{c}{$\mathbf{4.4 \times 10^3}$ \textbf{(chained)} \quad $21.1 \times 10^3$ (unchained)} \\
                \bottomrule
                \end{tabular}%
            }
            \caption{Time in ms to probe a chained vs.\ unchained index.}
            \label{tab:probe-time}

        \end{table*}

        \begin{table*}[t]
            \centering

                \centering
                \resizebox{.56\linewidth}{!}{%
                \begin{tabular}{@{}llrrrr@{}}\toprule
                    Benchmark & Method        & min   & avg    & med     & max     \\ \midrule
                    \job      & chained SYA   & 18.83                & 1211.04              & 1270.75              & 3028.94              \\
                    \job      & unchained SYA & 19.28                & 1232.21              & 1279.84              & 3113.61              \\
                    \job      & binary join & 125.41          & 1304.88        & 1295.94        & 3164.28              \\
                    \stats    & chained SYA   & 1.286                & 212.59               & 14.87                & 5249.68              \\
                    \stats    & unchained SYA & 1.227                & 186.35               & 14.78                & 4728.56              \\
                    \stats    & binary join  & 1.551          & 1178.65        & 26.47          & 46024.57   \\
                    \bottomrule
                    \end{tabular}
                   }
                \caption{End-to-end full join runtimes (ms).}
                \label{tab:fulljoin-chained-unchained}

            \end{table*}

\subsection{Chained vs.\@ unchained}
\label{subsec:unchained-experiments}

\noindent\textbf{Index building.}
We next compare \usr and \csr in more depth.
We begin by comparing the index-building time. Recall that \usr construction requires two hashing passes, and is therefore expected to be slower than \csr construction. 
{\sloppy
This is confirmed in our experiments: for \job, the (min/avg/max) \csr build-time speedups compared to \usr are ($0.91$/$1.3$/$2.17$)x or ($-9.9$/$69$/$131.3$)ms; for \stats they are ($0.82$/$1.02$/$1.36$)x or ($-8.1$/$0.14$/$3.7$)ms. For $\contactquery$ with population size 11M, \csr building is $1.17$x ($12.8$s) faster.
\par}

\medskip\noindent\textbf{Index probing.}  Table~\ref{tab:probe-time} reports the
avg/med/max 
probe times for both index types,
along with a distribution of the largest join degree $d$ over the entire set of queries.
Uniform runtimes are aggregated
across all sampling probabilities. 
We observe that \csr probing  is often faster than \usr probing.  
This is because $d$ remains low for most of the queries, and for such $d$ the binary search is empirically slower than linear search. One \job query exhibits an extremely high max degree $d$ (=342381). However, its full join output size is extremely low (=10), hence there is not much to be probed.  

\medskip\noindent\textbf{Full join processing.}  We next analyze the end-to-end
runtimes of computing full joins using \usr-based \sya versus \csr-based \sya. Once the index is built, these
methods use the flatten operator (\flatten) to produce the join output instead
of probing the index. 
In Table~\ref{tab:fulljoin-chained-unchained} we observe that on \job, computing
full joins using \csr is faster due to the dominating but faster
index construction time. However, for \stats, computing full joins using \csr 
is actually slower. We observe that the flatten operation for \usr 
is on average slightly faster than for \csr. 
We observe the same for $\contactquery$; the unchained variant computes the full join $1.16$x faster ($60.7$s) than the chained variant.
\inFullVersion{A detailed analysis of the flatten operator is provided in the Supplementary Material. } \inConfVersion{A detailed analysis of the flatten operator is provided in the full paper version~\cite{full-version}.} 
When considering \emph{median}
runtimes of both the flatten operator and the end-to-end time, however, the
chained and unchained variants are competitive.

\smallskip\noindent\textbf{Chained or unchained?} When are unchained methods
preferable over the chained variants? To answer this question, we conducted an
additional experiment on synthetic data, focusing on the binary join
$\beta_p \left(S(x,y) \Join T(y,z)\right)$ with $T$ appearing in the right-hand side
  of the corresponding nested semijoin and $p \in [10^{-4},10^{-1},0.5,0.9]$. We consider a setup that allows us to modulate the degree $d$, as this controls  the theoretical complexity difference between \csr and \usr probing.  We
  consider three scenarios, fixing the join output size $O$ to be $10^5, 10^6$, and
  $10^7$, respectively.  Per fixed value of $O$ and $p$, the number of probed tuples, which is a rough measure for the amount of work that needs to be done, hence remains constant. For each output size $O = 10^o$, we consider all
  configurations of inputs $(S,T)$ such that $|S| = 10^s$ and
  $\deg_y(T) = 10^{o-s}$ for $1 \leq s < o$, ensuring that each $S$-tuple
  joins with $\deg_y(T)$ tuples in $T$ to get the desired output size. $S$ is populated such that each key value appears only once in $S$, and hence $O = |T|$. Because each tuple in $S$ joins with the same number of tuples in $T$, the maximum join degree $d$ equals $\deg_y(T)$. 
  Tuples in $T$ are randomly permuted to ensure that tuples with the same
  $y$-value are not consecutive in the chained representation. We use $\hybrid$ for position sampling. 
  We enable caching for both \csr and \usr as it consistently improves performance for both on this experiment.

   \inFullVersion{Detailed measurements for all configurations and probabilities are given in the Supplemental Material}\inConfVersion{Detailed measurements for all configurations and probabilities are given in the full paper version~\cite{full-version}}; here we summarize the main trends.
   We observe that for $O = |T| = 10^5$ the \csr \nxt \xspace
  array is small enough to fit in L2 CPU cache and for this reason, \csr
  probing (as well as total execution time) is faster than \usr probing for all $p$, even
  for large degree $d = 10^4$.
  When $O = |T|$ increases to $10^6$,  \csr
  probing becomes slower than  \usr probing for all $p$, even for low degree $d = 10^1$. 
  However, the total \csr runtime remains faster than \usr for
  all join degrees (up to $10^5$) due to the faster build time.
  When $O = |T|$ further increases to $10^7$, \csr probing remains slower than \usr probing but also the total runtime becomes slower for the majority of configurations. However, there remain  configurations and values of $p$ where  \csr's faster index construction time results in lower total runtime for \csr.

In conclusion, this synthetic experiment shows that there are cases where the theoretically faster \usr probe time is also experimentally faster, even when the degree is relatively small. However, this does not necessarily translate to a faster overall runtime, due to the interaction with (i) CPU caching effects; (ii) the index construction time; and (iii) the sampling probability. On our real-world benchmarks, chained \iap methods are more robust compared to their unchained counterparts. Since both variants are competitive for full join materialization we recommend the chained index as the preferred choice for both sampling and full join computation.

\section{Related Work}
\label{sec:related-work}

There is extensive
work on \emph{uniform sampling} over acyclic
joins~\cite{DBLP:conf/sigmod/ChaudhuriMN99,DBLP:conf/sigmod/AcharyaGPR99a,DBLP:conf/sigmod/ZhaoC0HY18,DBLP:journals/tods/CarmeliZBCKS22,DBLP:journals/sigmod/DaiHY25}, differing in the subclasses of acyclic joins that are supported, from binary
joins~\cite{DBLP:conf/sigmod/ChaudhuriMN99}, over foreign-key
joins~\cite{DBLP:conf/sigmod/AcharyaGPR99a} to arbitrary acyclic
joins~\cite{DBLP:conf/sigmod/ZhaoC0HY18,DBLP:journals/tods/CarmeliZBCKS22,DBLP:journals/sigmod/DaiHY25}
as well as the specific kind of uniform sampling considered---from arbitrary
uniform
sampling~\cite{DBLP:conf/sigmod/ChaudhuriMN99,DBLP:conf/sigmod/AcharyaGPR99a,DBLP:conf/sigmod/ZhaoC0HY18},
over fixed-size sampling~\cite{DBLP:journals/tods/CarmeliZBCKS22} to fixed-size
sampling in a streaming context~\cite{DBLP:journals/sigmod/DaiHY25}. We consider
arbitrary Poisson sampling without fixing a sample size.

Existing works also differ in whether sampling is performed with replacement~\cite{DBLP:conf/sigmod/ChaudhuriMN99,DBLP:conf/sigmod/AcharyaGPR99a,DBLP:conf/sigmod/ZhaoC0HY18} or, like our work, without replacement~\cite{DBLP:journals/tods/CarmeliZBCKS22,DBLP:journals/sigmod/DaiHY25}.
Any sampling-with-replacement algorithm can be adapted to sampling without replacement by rejecting duplicates; however, as the sample size grows, repeated resampling leads to increasing rejection overhead, making this approach inefficient compared to methods designed for sampling without replacement~\cite{DBLP:journals/tods/CarmeliZBCKS22}.
The fixed-size sampling approach taken by Carmeli et al~\cite{DBLP:journals/tods/CarmeliZBCKS22} is the closest to our work. However, we consider non-uniform sampling with a focus on implementation in column stores and show that from a non-asymptotic viewpoint the theoretically best method is not the practically most efficient.
Prior work on \emph{subset sampling}~\cite{10.1145/355744.355749,DBLP:journals/pvldb/ZhangJ023,huang2023subsetsamplingextensions}—where each element of a set has its own sampling probability—addresses a simpler setting, as it does not involve the challenge of avoiding materializing full joins.

In parallel to our own work, Esmailpour et al.~\cite{esmailpour2025subsetsamplingjoins} have recently also studied Poisson sampling over acyclic joins. Compared to our work, they focus on a theoretical perspective, also considering the setting where the sampling probability is not in a single relation, but computed from attributes from multiple relations. By contrast, we focus on the practical implementation in column stores, limiting the probability to come from attributes of a single relation.

\section{Conclusions}
\label{sec:conclusions}

We introduced the problem of Poisson sampling over acyclic joins, where each join tuple specifies the probability with which it is to be included in the sample. This setting generalizes classical uniform sampling. We presented a nearly instance-optimal algorithm based on the Index-and-Probe paradigm.
We have shown how to efficiently implement this approach in column stores by building upon the \sya algorithm from~\cite{DBLP:journals/pvldb/BekkersNVW25}. Additionally, we compared two different index structures, \csr and \usr. While \usr achieves the theoretically optimal access time, we found empirically that \csr offers the best end-to-end performance. Since \csr and \usr are competitive for computing full joins, it is possible to adopt \sya with \csr as a uniform basis for both sampling and join processing.

\inFullVersion{An interesting direction for future work is to extend Poisson sampling to \emph{cyclic} joins.
While uniform sampling over cyclic queries has received attention~\cite{DBLP:conf/sigmod/ZhaoC0HY18, DBLP:journals/sigmod/DaiHY25}, efficient Poisson sampling over cyclic queries remains unexplored. }

\begin{acks}
Liese Bekkers and Stijn Vansummeren were supported by the Bijzonder
Onderzoeksfonds (BOF) of Hasselt University (Belgium) under Grants
No. BOF22DOC07 and BOF20ZAP02. This research was further supported by  Research Foundation Flanders (FWO) under Grant No. G0B9623N.
\end{acks}

\bibliographystyle{ACM-Reference-Format}
\bibliography{references}

\inFullVersion{
\clearpage
\appendix
\section*{Supplementary Material for Section~\ref{sec:indexing}}

\subsection*{Optimized \csr random access}

\newcommand{\offsets}{\textit{offsets}}
\newcommand{\subpos}{\textit{subpos}}
\newcommand{\previ}{\textit{i$_{\mathrm{prev}}$}}
\newcommand{\prevj}{\textit{j$_{\mathrm{prev}}$}}

For a given nested relation $N$ and an offset $i$ into $\flatten(N)$, Figure~\ref{algo:csr-access} describes the procedure for accessing the tuple at offset $i$ directly from the \csr $\store{}$ of $N$. When accessing multiple offsets $\pos = [i_1, \ldots, i_k]$ in bulk, the procedure can be optimized by reusing computations from previous offsets, an optimization we refer to as \emph{caching optimization}. The procedure for accessing multiple offsets in bulk with caching optimization is shown in Figure~\ref{algo:csr-access-batched}. We next highlight the key differences compared to the original \csr access algorithm from Figure~\ref{algo:csr-access}.
\begin{compactenum}
    \item First, \csraccess \xspace now takes a list of offsets $\pos$ instead of a single offset $i$ [line 1].  
    \item Identify, \emph{for each offset $i$ in $\pos$}, the nested tuple that produces the tuple at offset $i$ in $\flatten(N)$, and the corresponding sub-offset within that nested tuple [lines 4-7]. Store these in the lists $\offsets$ and $\subpos$ respectively. Originally, this was done for a single offset $i$.
    \item Next, call the helper function  \csrsubaccess \xspace  with the lists $\offsets$ and $\subpos$ instead of a single offset $i$ within $\flatten(N)$ and offset $j$ within the nested tuple. [line 8] 
    \item Copy, for every $j \in \offsets$, the values of all flat attributes to the output result [line 11]. This was originally done for a single offset $j$. 
    \item For each nested attribute $Y$, instead of accessing a single tuple, we need to access \emph{all} output tuples by iterating over all pairs $(j, i)$ in $(\offsets, \subpos)$ [line 15]. Offsets $i'$ within the nested attribute $Y$ are computed as before, i.e.\ by means of the mixed-radix indexing schema. However, since we compute $i'$ for all output tuples in bulk, we can apply caching optimization: if the current offset $j$ is the same as the previous offset \prevj \xspace and the current sub-offset $i$ is not smaller than the previous sub-offset \previ, meaning that the tuple will be found later in the same linked list, we can resume the linked list traversal from where we left off in the previous iteration instead of starting from the head of the linked list [lines 16-20]. Otherwise, we start from the head of the linked list as before [lines 21-25].
\end{compactenum}

\begin{figure}[h]
    \tt \small
    \begin{algorithmic}[1]
        \Function{\csraccess}{$\store{}, X, \pos$}
            \State result $\gets \emptyset$ \Comment{output tuples}
            \State $\offsets \gets [\,]$; $\subpos \gets [\,]$ 
        
            \For{each $i \in \pos$}
                \State find smallest $j$ in 0..\card{\store{}(X)} s.t. $i <  \prefix{}[j]$
                \State append $j$ to $\offsets$
                \State append $i - \prefix[j-1]$ to $\subpos$ \Comment{assume $\prefix[-1]=0$}
            \EndFor
        
            \State \Call{\csrsubaccess}{$\store{}, X, \offsets, \subpos$}
            \State \Return result
        \EndFunction
        \smallskip
        
        \Function{\csrsubaccess}{$\store{}, X, \offsets, \subpos$}
            \State result $\gets$ result $\uplus 
                \{ a \mapsto \store{}(X).a[j] \mid a \in X,\, j \in \offsets \}$
            
            \For{each nested attribute $Y \in X$}
                \State $\offsets_Y \gets [\,]$; $\subpos_Y \gets [\,]$
                \State $\previ \gets -1$; $\prevj \gets -1$
                \For{each pair $(j, i) \in (\offsets, \subpos)$}
                    \If{$j == \prevj$ \textbf{and} $i \ge \previ$}
                        \State  \Comment{resume previous linked list iteration}
                        \State reuse $j', w$ from previous iteration
                        \State $i' \gets i \bmod w$ 
                        \State $i' \gets i' - $ sum\_weights 
                    \Else
                        \State  \Comment{start at beginning of linked list}
                        \State $(j', w) \gets \store{}(X)[j](\hol{Y}, \weight{Y})$ 
                        \State $i' \gets i \bmod w$
                        \State sum\_weights $\gets 0$
                    \EndIf

                    \State $i \gets i$ div $w$ \Comment{overwrite $i$ in $\subpos$}
                    \State curr\_weight $\gets$ weight\_of$(\store{}, Y, j')$
                    \While{$j' \ge 0$ \textbf{and} $i' \ge$ curr\_weight}
                        \State sum\_weights $\gets$ sum\_weights $+$ curr\_weight
                        \State $i' \gets i' -$ curr\_weight
                        \State $j' \gets \store{}(Y).\nxt[j']$
                        \State curr\_weight $\gets$ weight\_of$(\store{}, Y, j')$
                    \EndWhile
        
                    \State append $j'$ to $\offsets_Y$
                    \State append $i'$ to $\subpos_Y$
                    \State $\previ \gets i$; $\prevj \gets j$
                \EndFor
        
                \State \Call{\csrsubaccess}{$\store{}, Y, \offsets_Y, \subpos_Y$}
            \EndFor
        \EndFunction
        \end{algorithmic}
\caption{Optimized \csr random access.}
\label{algo:csr-access-batched}
\end{figure}

\subsection*{Optimized \usr random access}

\newcommand{\slices}{\textit{slices}}
\newcommand{\prevs}{\textit{s$_{\mathrm{prev}}$}}
\newcommand{\prevl}{\textit{l$_{\mathrm{prev}}$}}

The same caching optimization can be applied to \usr random access as well. The optimized variant of the procedure in Figure~\ref{algo:usr-access} is shown in Figure~\ref{algo:usr-access-batched}. The differences compared to the original \usr access algorithm from Figure~\ref{algo:usr-access} are as follows.

\begin{compactenum}
    \item \usraccess \xspace now takes a list of offsets $\pos$ instead of a single offset $i$ [line 1]. Similarly, the helper function \usrsubaccess \xspace now takes lists of offsets $\pos$ and, since each offsets may correspond to a different slice of the \prefix{} and \perm \xspace arrays, a list of pairs $\slices$ [line 7] where each pair $(s, l)$ indicates the start and length of the relevant portion of the \prefix{} and \perm \xspace arrays for that offset. 
    \item Next, for each offset $i$ in $\pos$ and its corresponding slice $(s, l)$ in $\slices$, we perform a binary search to find the position $j$ within the slice where the tuple at offset $i$ can be found [lines 10-22]. Since we do this for all $i \in \pos$ in bulk, it is possible to apply caching optimization as follows:
    \begin{compactenum}[a)]
        \item If two subsequent offsets require a binary search in the same slice, the search can be resumed from where the previous search left off [lines 10-14]. To facilitate this, we keep track of the previous offset $\previ$, its corresponding slice $(\prevs, \prevl)$ and binary search result $\prevj$ [lines 9,18,22]. Binary search is now performed in $\prefix[s+\prevj:l]$ instead of the full slice $\prefix[s:l]$.
        \item Otherwise, a new binary search is performed in the full slice $\prefix[s:l]$ [lines 15-17].
    \end{compactenum}
\end{compactenum}

\begin{figure}
    \tt \small
    \begin{algorithmic}[1]
\Function{\usraccess}{\store, \pos}\algorithmicdo
      \State result $\gets \emptyset$ \Comment{output tuples, initially empty}
      \State $\perm \gets [1, 2, ..., \card{\store{}(X)}]$
      \State $\slices \gets$ array of length $|\pos|$, each entry $(1,\, |\store{}(X)|)$
      \State \usrsubaccess(\store, X, pos, \slices, \store{}(X).\prefix{}, \perm)
      \State return result
      \EndFunction
      \smallskip

      \Function{\usrsubaccess}{\store, $X$, \pos, \slices, \prefix{}, \perm}\algorithmicdo

      \State $I \gets [\,]$; $J \gets [\,]$ 
      \State $\previ \gets -1$; $\prevs \gets -1$; $\prevl \gets -1$; $\prevj \gets -1$
      \For{each pair $(i,(s,l)) \in (\pos,\slices)$}
        \If{$s == \prevs$ and $l == \prevl$ and $i \ge \previ$}
            \State \Comment{Optimization: resume previous binary search}
            \State \Comment{search only in \prefix{}[(s+\prevj):l]}
            \State find smallest $j$ in $0..\len{\prefix{}[s:l]}$ s.t. $i < \prefix{}[s+j]$
        \Else
            \State \Comment{No optimization: binary search in original slice}
            \State find smallest $j$ in $0..\len{\prefix{}[s:l]}$ s.t. $i < \prefix{}[s+j]$
        \EndIf
        \State $\prevj \gets j$
        \State $i \gets i - \prefix{}[s+j-1]$ \Comment{assume \prefix{}[-1] = 0}
        \State $j \gets \perm[s+j]$
        \State append $i$ to $I$; append $j$ to $J$
        \State $\previ \gets i$; $\prevs \gets s$; $\prevl \gets l$
    \EndFor
      \State result = result $\uplus  \{ a \mapsto \store{}(X).a[j] \mid a \in X,\, j \in J \}$
      \For{each nested attr $Y \in X$}
       \State $\pos_Y \gets [\,]$; $\slices_Y \gets [\,]$
      \For{each pair $(i, j) \in (I, J)$}
        \State $w \gets $ weight\_of($\store, Y, j$)
        \State $i' \gets i \imod w$; append $i'$ to $\pos_Y$
        \State $i \gets  i \idiv w$ \Comment{overwrite $i$ in $I$}
        \State $(s, l) \gets \store{}(X)[j](\start\_{Y},\llen\_{Y})$
        \State append $(s, l)$ to $\slices_Y$
    \EndFor
      \State \usrsubaccess$(\store, Y, \pos_Y, \slices_Y, \store{}(Y).\prefix{}, \store{}(Y).\perm)$
      \EndFor
      \EndFunction
    \end{algorithmic}
\caption{Optimized \usr random access.}
\label{algo:usr-access-batched}
\end{figure}
 \section*{Supplementary Material for Section~\ref{sec:experiments}}
\subsection*{Speedups of \iap over \syasamplingC for uniform sampling}

Table~\ref{tab:uniform-total-runtimes} shows the end-to-end runtime speedups of \iap over \syasamplingC for the \job and \stats workloads when varying the sampling probability $p$ for uniform sampling. Relative and absolute (in ms) speedups are reported for both chained and unchained indexes.
\begin{table*}[]
\centering
\resizebox{\textwidth}{!}{%
\begin{tabular}{@{}llrrrrrrrrrrr@{}} \toprule
    &                               & \multicolumn{5}{c}{\textbf{\job}}                              && \multicolumn{5}{c}{\textbf{\stats}}                            \\
    \cmidrule{3-7} \cmidrule{9-13} 
    &                               & \multicolumn{2}{c}{chained speedups} && \multicolumn{2}{c}{unchained speedups} && \multicolumn{2}{c}{chained speedups} && \multicolumn{2}{c}{unchained speedups} \\
    \cmidrule{3-4} \cmidrule{6-7} \cmidrule{9-10} \cmidrule{12-13}
p      & method                        & Rel.               & Abs.            && Rel.               & Abs.               && Rel.              & Abs.             && Rel.                & Abs.             \\ \midrule
0,0001 & \geomsampling  & 1.0056                         & 6.76                        && 0.9884                         & -14.26                         && 38.7854                       & 457.79                       && 38.7063                         & 457.77                       \\
0,001  & \geomsampling  & 1.0051                         & 6.20                        && 0.9882                         & -14.49                         && 35.5493                       & 460.74                       && 35.3258                         & 460.66                       \\
0,01   & \geomsampling  & 1.0057                         & 6.86                        && 0.9884                         & -14.21                         && 19.9837                       & 471.33                       && 19.5483                         & 470.77                       \\
0,1    & \geomsampling  & 1.0045                         & 5.44                        && 0.9873                         & -15.57                         && 4.3527                        & 465.13                       && 4.1934                          & 459.86                       \\
0,2    & \geomsampling  & 1.0037                         & 4.48                        && 0.9866                         & -16.46                         && 2.4768                        & 412.45                       && 2.3871                          & 401.95                       \\
0,3    & \geomsampling  & 1.0035                         & 4.23                        && 0.9862                         & -16.95                         && 1.7905                        & 347.01                       && 1.7342                          & 332.75                       \\
0,4    & \geomsampling  & 1.0025                         & 3.08                        && 0.9855                         & -17.87                         && 1.4292                        & 266.14                       && 1.3894                          & 248.37                       \\
0,5    & \geomsampling  & 1.0021                         & 2.57                        && 0.9844                         & -19.20                         && 1.1962                        & 157.64                       && 1.1554                          & 129.26                       \\
0,6    & \naivesampling & 1.0013                         & 1.63                        && 0.9836                         & -20.18                         && 1.0773                        & 68.44                        && 1.0307                          & 28.40                        \\
0,7    & \naivesampling & 1.0010                         & 1.22                        && 0.9832                         & -20.69                         && 1.0506                        & 44.21                        && 0.9982                          & -1.65                        \\
0,8    & \naivesampling & 1.0011                         & 1.28                        && 0.9831                         & -20.86                         && 1.0221                        & 19.16                        && 0.9643                          & -32.78                       \\
0,9    & \naivesampling & 1.0003                         & 0.40                        && 0.9826                         & -21.47                         && 0.9943                        & -4.96                        && 0.9325                          & -62.41                       \\
0,99   & \naivesampling & 1.0005                         & 0.56                        && 0.9827                         & -21.38                         && 0.9742                        & -22.47                       && 0.9112                          & -82.69                       \\ 
\bottomrule
\end{tabular}
}
\caption{Average end-to-end runtime speedups, both relative and absolute (ms), compared to \syasamplingC for different values of $p$.}
\label{tab:uniform-total-runtimes}
\end{table*} %

\subsection*{Caching}

In this section we motivate why in Sections~\ref{subsec:uniform-experiments} and Section~\ref{subsec:nonuniform-experiments} we enabled caching for the chained index but disabled it for the unchained index. In particular, we compare the effect of caching on both indexes by running the benchmarks from Section~\ref{sec:experiments} with and without caching enabled.
For uniform sampling, $p$ varies from $0.01\%$ to $99\%$ where we use $\geomsampling$ when $p \le 0.5$ and $\naivesampling$ otherwise. For Poisson sampling, we use $\pertuplesampling$ with $\low$, $\medium$, and $\high$ sampling probabilities. Each experiment is repeated 5 times, and reported runtimes are averages over these 5 runs.
The results for \job and \stats are summarized in Table~\ref{tab:caching}. 
For the unchained index, caching increases runtime. While caching reduces the search space, it also introduces additional bookkeeping overhead. In our benchmarks, this overhead outweighs any potential benefit because the search spaces are already small. As shown in Table~\ref{tab:probe-time}, the maximum join degree $d$, and hence the size of the search space, remains low throughout all experiments, leaving little opportunity for caching to prune substantial portions of the search space.

This begs the question of why caching consistently improves performance for the chained index, even when the search spaces are small. The key difference lies in memory access behavior. Probing a chained index involves traversing linked lists whose elements may reside at non-contiguous locations in main memory, resulting in more cache misses and actual random memory access. In this setting, caching reduces the number of such costly accesses. In contrast, binary search in the unchained index operates over contiguous memory, which is already cache-friendly. Consequently, the marginal benefit of caching is minimal, while its overhead remains.

Caching on a chained index yields more significant performance improvements for the \stats benchmark compared to \job, and the differences become more pronounced as the sampling probability $p$ increases. For \job, the full join output sizes are relatively small (see Table~\ref{tab:benchmarks}), and the total runtime is therefore dominated by index construction rather than probing as shown in Figure~\ref{fig:job-uniform-runtime-breakdown}. As a result, caching optimization during the probing phase has only a limited impact on overall performance. In contrast, \stats exhibits much larger full join outputs, causing a greater fraction of the total runtime to be spent on index probing (see Figure~\ref{fig:statsceb-uniform-runtime-breakdown}). Consequently, improving the probe time by means of caching has a more substantial impact on the total runtime.

Running query $\contactquery$ on a population of $1.1 \times 10^7$ individuals is $492$ms faster (chained) and $572$ms slower (unchained) with caching optimization.

Based on these observations, we conclude that enabling caching for the chained index and disabling it for the unchained index results in the best possible runtimes in our benchmarks.
\begin{table*}[]
    \centering
\resizebox{\textwidth}{!}{%
    \begin{tabular}{@{}lrrrrrrrrrrrr@{}} \toprule
        &  & \multicolumn{5}{c}{\job}                                       &  & \multicolumn{5}{c}{\stats}                                         \\
        \cmidrule{3-7} \cmidrule{9-13}
        $p$&  & \chainedNoOpt & \chainedOpt    &  & \unchainedNoOpt       & \unchainedOpt &  & \chainedNoOpt & \chainedOpt    &  & \unchainedNoOpt       & \unchainedOpt     \\
        \midrule
0.0001  &  & 1.7644  & \textbf{1.7627} &  & \textbf{1.7914} & 1.7930        &  & 0.0114  & \textbf{0.0112} &  & 0.01142          & \textbf{0.01137} \\
0.001   &  & 1.7645  & \textbf{1.7623} &  & \textbf{1.7916} & 1.7939        &  & 0.0131  & \textbf{0.0125} &  & \textbf{0.01268} & 0.01270          \\
0.01    &  & 1.7665  & \textbf{1.7639} &  & \textbf{1.7926} & 1.7946        &  & 0.0294  & \textbf{0.0239} &  & \textbf{0.0245}  & 0.0247           \\
0.1     &  & 1.7677  & \textbf{1.7660} &  & \textbf{1.7931} & 1.7955        &  & 0.1920  & \textbf{0.1370} &  & \textbf{0.1423}  & 0.1445           \\
0.2     &  & 1.7699  & \textbf{1.7662} &  & \textbf{1.7945} & 1.7963        &  & 0.3880  & \textbf{0.2774} &  & \textbf{0.2872}  & 0.2920           \\
0.3     &  & 1.7721  & \textbf{1.7681} &  & \textbf{1.7981} & 1.7986        &  & 0.6035  & \textbf{0.4362} &  & \textbf{0.4496}  & 0.4571           \\
0.4     &  & 1.7751  & \textbf{1.7688} &  & \textbf{1.7977} & 1.7998        &  & 0.8381  & \textbf{0.6174} &  & \textbf{0.6330}  & 0.6432           \\
0.5     &  & 1.7758  & \textbf{1.7700} &  & \textbf{1.8005} & 1.8018        &  & 1.0770  & \textbf{0.7997} &  & \textbf{0.8219}  & 0.8342           \\
0.6     &  & 1.7782  & \textbf{1.7703} &  & \textbf{1.8005} & 1.8029        &  & 1.2068  & \textbf{0.8746} &  & \textbf{0.9113}  & 0.9272           \\
0.7     &  & 1.7804  & \textbf{1.7720} &  & \textbf{1.8027} & 1.8053        &  & 1.2523  & \textbf{0.8646} &  & \textbf{0.9047}  & 0.9224           \\
0.8     &  & 1.7825  & \textbf{1.7733} &  & \textbf{1.8025} & 1.8047        &  & 1.3044  & \textbf{0.8612} &  & \textbf{0.9040}  & 0.9242           \\
0.9     &  & 1.7855  & \textbf{1.7736} &  & \textbf{1.8040} & 1.8054        &  & 1.3623  & \textbf{0.8639} &  & \textbf{0.9095}  & 0.9329           \\
0.99    &  & 1.7909  & \textbf{1.7746} &  & \textbf{1.8053} & 1.8077        &  & 1.4185  & \textbf{0.8679} &  & \textbf{0.9141}  & 0.9396           \\
        &  &         & \textbf{}       &  & \textbf{}       &               &  &         & \textbf{}       &  & \textbf{}        &                  \\
\low    &  & 1.3380  & \textbf{1.3350} &  & \textbf{1.4540} & 1.4550        &  & 0.4110  & \textbf{0.1290} &  & \textbf{0.1330}  & 0.1340           \\
\medium &  & 1.3470  & \textbf{1.3420} &  & \textbf{1.4610} & 1.4620        &  & 2.1350  & \textbf{0.3700} &  & \textbf{0.3950}  & 0.3960           \\
\high   &  & 1.3480  & \textbf{1.3390} &  & \textbf{1.4590} & 1.4600        &  & 3.2050  & \textbf{0.4310} &  & \textbf{0.4650}  & 0.4660           \\
\bottomrule
\end{tabular}
}
\caption{Comparison of total runtimes (s) with and without caching optimization for both the chained and unchained index.}
    \label{tab:caching}
\end{table*}

\subsection*{Flatten}

\begin{figure*}
    \centering
    \begin{tikzpicture}
        \begin{loglogaxis}[
            xlabel={\syasamplingC (s)},
            ylabel={\syasamplingU (s)},
            legend pos=north west,
            xmin=1e-6, xmax=1e2,
            ymin=1e-6, ymax=1e2,
        ]       
        \addplot[
            only marks,
            mark=o,
            mark size=1pt,
            color=blue
        ]
        table [x={Yannakakis (chained)}, y={Yannakakis (unchained)}, col sep=comma]
        {./data/job-flatten.csv};
        \addlegendentry{\job};
        \addplot[
            only marks,
            mark=o,
            mark size=1pt,
            color=red
        ]
        table [x={Yannakakis (chained)}, y={Yannakakis (unchained)}, col sep=comma]
        {./data/statsceb-flatten.csv};
        \addplot [domain=1e-6:1e2, solid, color=gray, on layer=axis background] {x};
        \addlegendentry{\stats};
        \end{loglogaxis}
    \end{tikzpicture}
    \caption{Comparison of flatten times for chained and unchained \sya.}
    \label{fig:flattentimes}
\end{figure*}

Figure~\ref{fig:flattentimes} shows a log–log scatterplot comparing the flatten times for each query in \job and \stats. Each dot corresponds to a single query; dots below (above) the diagonal indicate that the unchained flatten is faster (slower) than the chained flatten for that query. On average, the unchained flatten is $26$ ms faster for \stats, while for \job the average difference is below $1$ ms. Since all points lie on or very close to the diagonal, we conclude that the chained and unchained flatten are competitive in practice.

\subsection*{Synthetic benchmark}

This appendix provides all figures for the synthetic benchmark detailed in Section~\ref{subsec:unchained-experiments}. The runtime breakdowns for $O = |T| = 10^5$, $10^6$, and $10^7$ are given in Figures~\ref{fig:synth-100k}, \ref{fig:synth-1000k}, and \ref{fig:synth-10000k} respectively. The labels on the x-axis refer to the three methods that we compare. From left to right: chained index with caching optimization, unchained index without caching optimization, and unchained index with caching optimization. For this synthetic benchmark, caching optimization has a significant effect on the performance of both indexes, in contrast to the other benchmarks (\job, \stats) where \usr caching has only a marginal (negative) effect (see Table~\ref{tab:caching}).

\begin{figure*}
    \centering
    \ref{namedlegend}
    \vspace{1em} 
\begin{subfigure}{.48\textwidth}

    \begin{tikzpicture}
      \pgfplotstableread[col sep=comma]{./data/synthetic-paper-data/s_10_d_10000.csv}\datatable
      
      \begin{axis}[
          axis lines*=left, ymajorgrids,
          width=\linewidth, height=4cm,
          ymin=0,
          ybar stacked,
          bar width=7pt,
          xtick=data,
          xticklabels from table={\datatable}{method},
          xticklabel style={rotate=90,anchor=mid east,font=\scriptsize},
          ylabel={Avg.\ runtime (ms)},    
          legend to name=namedlegend,
          legend columns=-1,
          legend style={draw=none, /tikz/every even column/.append style={column sep=1.2em}},
          draw group line={p}{0.0001}{$0.0001$}{-12ex}{4pt},
          draw group line={p}{0.001}{$0.001$}{-12ex}{4pt},
          draw group line={p}{0.01}{$0.01$}{-12ex}{4pt},
          draw group line={p}{0.1}{$0.1$}{-12ex}{4pt},
          draw group line={p}{0.2}{$0.2$}{-12ex}{4pt},
          draw group line={p}{0.3}{$0.3$}{-12ex}{4pt},
          draw group line={p}{0.4}{$0.4$}{-12ex}{4pt},
          draw group line={p}{0.5}{$0.5$}{-12ex}{4pt},
          draw group line={p}{0.6}{$0.6$}{-12ex}{4pt},
          draw group line={p}{0.7}{$0.7$}{-12ex}{4pt},
          draw group line={p}{0.8}{$0.8$}{-12ex}{4pt},
          draw group line={p}{0.9}{$0.9$}{-12ex}{4pt},
          draw group line={p}{0.99}{$0.99$}{-12ex}{4pt},
          draw group line={p}{1}{$1$}{-12ex}{4pt},
          after end axis/.append code={
              \path [anchor=base east, yshift=-13.5ex]
                  (rel axis cs:0,0) node  {$p$};
          }
      ]
        \addplot [
            fill=green!30!white,
            postaction={pattern=horizontal lines, pattern color=black}
        ] table [x=X, y=Index probing] {\datatable};
      
        \addplot [
            fill=gray!25,
            postaction={pattern=grid, pattern color=black}
        ] table [x=X, y=Position sampling] {\datatable};
    
        \addplot [
            fill=yellow!30,
            postaction={pattern=dots, pattern color=black}
        ] table [x=X, y=Index building] {\datatable};

          \addlegendentry{Index probing}
          \addlegendentry{Position sampling}
          \addlegendentry{Index building}
      
      \end{axis}
\end{tikzpicture}

    \caption{$|S| = 10^1$, $\deg_y(T) = 10^4$}
\end{subfigure}
\hfill
\begin{subfigure}{.48\textwidth}

    \begin{tikzpicture}
      \pgfplotstableread[col sep=comma]{./data/synthetic-paper-data/s_100_d_1000.csv}\datatable
      
      \begin{axis}[
          axis lines*=left, ymajorgrids,
          width=\linewidth, height=4cm,
          ymin=0,
          ybar stacked,
          bar width=7pt,
          xtick=data,
          xticklabels from table={\datatable}{method},
          xticklabel style={rotate=90,anchor=mid east,font=\scriptsize},
          ylabel={Avg.\ runtime (ms)},    
          legend style={
              at={(0.5,1)},
              anchor=south,
              draw=none,
          },
          legend columns=-1,
          /tikz/every even column/.append style={column sep=1.2em},
          draw group line={p}{0.0001}{$0.0001$}{-12ex}{4pt},
          draw group line={p}{0.001}{$0.001$}{-12ex}{4pt},
          draw group line={p}{0.01}{$0.01$}{-12ex}{4pt},
          draw group line={p}{0.1}{$0.1$}{-12ex}{4pt},
          draw group line={p}{0.2}{$0.2$}{-12ex}{4pt},
          draw group line={p}{0.3}{$0.3$}{-12ex}{4pt},
          draw group line={p}{0.4}{$0.4$}{-12ex}{4pt},
          draw group line={p}{0.5}{$0.5$}{-12ex}{4pt},
          draw group line={p}{0.6}{$0.6$}{-12ex}{4pt},
          draw group line={p}{0.7}{$0.7$}{-12ex}{4pt},
          draw group line={p}{0.8}{$0.8$}{-12ex}{4pt},
          draw group line={p}{0.9}{$0.9$}{-12ex}{4pt},
          draw group line={p}{0.99}{$0.99$}{-12ex}{4pt},
          draw group line={p}{1}{$1$}{-12ex}{4pt},
          after end axis/.append code={
              \path [anchor=base east, yshift=-13.5ex]
                  (rel axis cs:0,0) node  {$p$};
          }
      ]
      
      \addplot [
        fill=green!30!white,
        postaction={pattern=horizontal lines, pattern color=black}
    ] table [x=X, y=Index probing] {\datatable};
  
    \addplot [
        fill=gray!25,
        postaction={pattern=grid, pattern color=black}
    ] table [x=X, y=Position sampling] {\datatable};

    \addplot [
        fill=yellow!30,
        postaction={pattern=dots, pattern color=black}
    ] table [x=X, y=Index building] {\datatable};

      \end{axis}
\end{tikzpicture}

\caption{$|S| = 10^2$, $\deg_y(T) = 10^3$}
\end{subfigure}

\medskip

\begin{subfigure}{.48\textwidth}

    \begin{tikzpicture}
      \pgfplotstableread[col sep=comma]{./data/synthetic-paper-data/s_1000_d_100.csv}\datatable
      
      \begin{axis}[
          axis lines*=left, ymajorgrids,
          width=\linewidth, height=4cm,
          ymin=0,
          ybar stacked,
          bar width=7pt,
          xtick=data,
          xticklabels from table={\datatable}{method},
          xticklabel style={rotate=90,anchor=mid east,font=\scriptsize},
          ylabel={Avg.\ runtime (ms)},    
          legend style={
              at={(0.5,1)},
              anchor=south,
              draw=none,
          },
          legend columns=-1,
          /tikz/every even column/.append style={column sep=1.2em},
          draw group line={p}{0.0001}{$0.0001$}{-12ex}{4pt},
          draw group line={p}{0.001}{$0.001$}{-12ex}{4pt},
          draw group line={p}{0.01}{$0.01$}{-12ex}{4pt},
          draw group line={p}{0.1}{$0.1$}{-12ex}{4pt},
          draw group line={p}{0.2}{$0.2$}{-12ex}{4pt},
          draw group line={p}{0.3}{$0.3$}{-12ex}{4pt},
          draw group line={p}{0.4}{$0.4$}{-12ex}{4pt},
          draw group line={p}{0.5}{$0.5$}{-12ex}{4pt},
          draw group line={p}{0.6}{$0.6$}{-12ex}{4pt},
          draw group line={p}{0.7}{$0.7$}{-12ex}{4pt},
          draw group line={p}{0.8}{$0.8$}{-12ex}{4pt},
          draw group line={p}{0.9}{$0.9$}{-12ex}{4pt},
          draw group line={p}{0.99}{$0.99$}{-12ex}{4pt},
          draw group line={p}{1}{$1$}{-12ex}{4pt},
          after end axis/.append code={
              \path [anchor=base east, yshift=-13.5ex]
                  (rel axis cs:0,0) node  {$p$};
          }
      ]
      
      \addplot [
        fill=green!30!white,
        postaction={pattern=horizontal lines, pattern color=black}
    ] table [x=X, y=Index probing] {\datatable};
  
    \addplot [
        fill=gray!25,
        postaction={pattern=grid, pattern color=black}
    ] table [x=X, y=Position sampling] {\datatable};

    \addplot [
        fill=yellow!30,
        postaction={pattern=dots, pattern color=black}
    ] table [x=X, y=Index building] {\datatable};

      \end{axis}
\end{tikzpicture}

\caption{$|S| = 10^3$, $\deg_y(T) = 10^2$}
\end{subfigure}
\hfill
\begin{subfigure}{.48\textwidth}

    \begin{tikzpicture}
      \pgfplotstableread[col sep=comma]{./data/synthetic-paper-data/s_10000_d_10.csv}\datatable
      
      \begin{axis}[
          axis lines*=left, ymajorgrids,
          width=\linewidth, height=4cm,
          ymin=0,
          ybar stacked,
          bar width=7pt,
          xtick=data,
          xticklabels from table={\datatable}{method},
          xticklabel style={rotate=90,anchor=mid east,font=\scriptsize},
          ylabel={Avg.\ runtime (ms)},    
          legend style={
              at={(0.5,1)},
              anchor=south,
              draw=none,
          },
          legend columns=-1,
          /tikz/every even column/.append style={column sep=1.2em},
          draw group line={p}{0.0001}{$0.0001$}{-12ex}{4pt},
          draw group line={p}{0.001}{$0.001$}{-12ex}{4pt},
          draw group line={p}{0.01}{$0.01$}{-12ex}{4pt},
          draw group line={p}{0.1}{$0.1$}{-12ex}{4pt},
          draw group line={p}{0.2}{$0.2$}{-12ex}{4pt},
          draw group line={p}{0.3}{$0.3$}{-12ex}{4pt},
          draw group line={p}{0.4}{$0.4$}{-12ex}{4pt},
          draw group line={p}{0.5}{$0.5$}{-12ex}{4pt},
          draw group line={p}{0.6}{$0.6$}{-12ex}{4pt},
          draw group line={p}{0.7}{$0.7$}{-12ex}{4pt},
          draw group line={p}{0.8}{$0.8$}{-12ex}{4pt},
          draw group line={p}{0.9}{$0.9$}{-12ex}{4pt},
          draw group line={p}{0.99}{$0.99$}{-12ex}{4pt},
          draw group line={p}{1}{$1$}{-12ex}{4pt},
          after end axis/.append code={
              \path [anchor=base east, yshift=-13.5ex]
                  (rel axis cs:0,0) node  {$p$};
          }
      ]
      \addplot [
        fill=green!30!white,
        postaction={pattern=horizontal lines, pattern color=black}
    ] table [x=X, y=Index probing] {\datatable};
  
    \addplot [
        fill=gray!25,
        postaction={pattern=grid, pattern color=black}
    ] table [x=X, y=Position sampling] {\datatable};

    \addplot [
        fill=yellow!30,
        postaction={pattern=dots, pattern color=black}
    ] table [x=X, y=Index building] {\datatable};

      \end{axis}
\end{tikzpicture}

\caption{$|S| = 10^4$, $\deg_y(T) = 10^1$}
\end{subfigure}

      \caption{$O = |T| = 10^5$}
      \label{fig:synth-100k}
    \end{figure*}

    \begin{figure*}
        \centering
        \ref{namedlegend}
    \begin{subfigure}{.48\textwidth}
    
        \begin{tikzpicture}
          \pgfplotstableread[col sep=comma]{./data/synthetic-paper-data/s_10_d_100000.csv}\datatable
          
          \begin{axis}[
              axis lines*=left, ymajorgrids,
              width=\linewidth, height=4cm,
              ymin=0,
              ybar stacked,
              bar width=7pt,
              xtick=data,
              xticklabels from table={\datatable}{method},
              xticklabel style={rotate=90,anchor=mid east,font=\scriptsize},
              ylabel={Avg.\ runtime (ms)},    
              legend to name=namedlegend,
              legend columns=-1,
              legend style={draw=none, /tikz/every even column/.append style={column sep=1.2em}},
              /tikz/every even column/.append style={column sep=1.2em},
              draw group line={p}{0.0001}{$0.0001$}{-12ex}{4pt},
              draw group line={p}{0.001}{$0.001$}{-12ex}{4pt},
              draw group line={p}{0.01}{$0.01$}{-12ex}{4pt},
              draw group line={p}{0.1}{$0.1$}{-12ex}{4pt},
              draw group line={p}{0.2}{$0.2$}{-12ex}{4pt},
              draw group line={p}{0.3}{$0.3$}{-12ex}{4pt},
              draw group line={p}{0.4}{$0.4$}{-12ex}{4pt},
              draw group line={p}{0.5}{$0.5$}{-12ex}{4pt},
              draw group line={p}{0.6}{$0.6$}{-12ex}{4pt},
              draw group line={p}{0.7}{$0.7$}{-12ex}{4pt},
              draw group line={p}{0.8}{$0.8$}{-12ex}{4pt},
              draw group line={p}{0.9}{$0.9$}{-12ex}{4pt},
              draw group line={p}{0.99}{$0.99$}{-12ex}{4pt},
              draw group line={p}{1}{$1$}{-12ex}{4pt},
              after end axis/.append code={
                  \path [anchor=base east, yshift=-13.5ex]
                      (rel axis cs:0,0) node  {$p$};
              }
          ]
            \addplot [
                fill=green!30!white,
                postaction={pattern=horizontal lines, pattern color=black}
            ] table [x=X, y=Index probing] {\datatable};
          
            \addplot [
                fill=gray!25,
                postaction={pattern=grid, pattern color=black}
            ] table [x=X, y=Position sampling] {\datatable};
        
            \addplot [
                fill=yellow!30,
                postaction={pattern=dots, pattern color=black}
            ] table [x=X, y=Index building] {\datatable};
    
              \addlegendentry{Index probing}
              \addlegendentry{Position sampling}
              \addlegendentry{Index building}
          
          \end{axis}
    \end{tikzpicture}
    
        \caption{$|S| = 10^1$, $\deg_y(T) = 10^5$}
    \end{subfigure} 

    \bigskip

    \begin{subfigure}{.48\textwidth}
    
        \begin{tikzpicture}
          \pgfplotstableread[col sep=comma]{./data/synthetic-paper-data/s_100_d_10000.csv}\datatable
          
          \begin{axis}[
              axis lines*=left, ymajorgrids,
              width=\linewidth, height=4cm,
              ymin=0,
              ybar stacked,
              bar width=7pt,
              xtick=data,
              xticklabels from table={\datatable}{method},
              xticklabel style={rotate=90,anchor=mid east,font=\scriptsize},
              ylabel={Avg.\ runtime (ms)},    
              legend style={
                  at={(0.5,1)},
                  anchor=south,
                  draw=none,
              },
              legend columns=-1,
              /tikz/every even column/.append style={column sep=1.2em},
              draw group line={p}{0.0001}{$0.0001$}{-12ex}{4pt},
              draw group line={p}{0.001}{$0.001$}{-12ex}{4pt},
              draw group line={p}{0.01}{$0.01$}{-12ex}{4pt},
              draw group line={p}{0.1}{$0.1$}{-12ex}{4pt},
              draw group line={p}{0.2}{$0.2$}{-12ex}{4pt},
              draw group line={p}{0.3}{$0.3$}{-12ex}{4pt},
              draw group line={p}{0.4}{$0.4$}{-12ex}{4pt},
              draw group line={p}{0.5}{$0.5$}{-12ex}{4pt},
              draw group line={p}{0.6}{$0.6$}{-12ex}{4pt},
              draw group line={p}{0.7}{$0.7$}{-12ex}{4pt},
              draw group line={p}{0.8}{$0.8$}{-12ex}{4pt},
              draw group line={p}{0.9}{$0.9$}{-12ex}{4pt},
              draw group line={p}{0.99}{$0.99$}{-12ex}{4pt},
              draw group line={p}{1}{$1$}{-12ex}{4pt},
              after end axis/.append code={
                  \path [anchor=base east, yshift=-13.5ex]
                      (rel axis cs:0,0) node  {$p$};
              }
          ]
          
          \addplot [
            fill=green!30!white,
            postaction={pattern=horizontal lines, pattern color=black}
        ] table [x=X, y=Index probing] {\datatable};
      
        \addplot [
            fill=gray!25,
            postaction={pattern=grid, pattern color=black}
        ] table [x=X, y=Position sampling] {\datatable};
    
        \addplot [
            fill=yellow!30,
            postaction={pattern=dots, pattern color=black}
        ] table [x=X, y=Index building] {\datatable};

          \end{axis}
    \end{tikzpicture}
    
    \caption{$|S| = 10^2$, $\deg_y(T) = 10^4$}
    \end{subfigure}
    \hfill 
    \begin{subfigure}{.48\textwidth}
    
        \begin{tikzpicture}
          \pgfplotstableread[col sep=comma]{./data/synthetic-paper-data/s_1000_d_1000.csv}\datatable
          
          \begin{axis}[
              axis lines*=left, ymajorgrids,
              width=\linewidth, height=4cm,
              ymin=0,
              ybar stacked,
              bar width=7pt,
              xtick=data,
              xticklabels from table={\datatable}{method},
              xticklabel style={rotate=90,anchor=mid east,font=\scriptsize},
              ylabel={Avg.\ runtime (ms)},    
              legend style={
                  at={(0.5,1)},
                  anchor=south,
                  draw=none,
              },
              legend columns=-1,
              /tikz/every even column/.append style={column sep=1.2em},
              draw group line={p}{0.0001}{$0.0001$}{-12ex}{4pt},
              draw group line={p}{0.001}{$0.001$}{-12ex}{4pt},
              draw group line={p}{0.01}{$0.01$}{-12ex}{4pt},
              draw group line={p}{0.1}{$0.1$}{-12ex}{4pt},
              draw group line={p}{0.2}{$0.2$}{-12ex}{4pt},
              draw group line={p}{0.3}{$0.3$}{-12ex}{4pt},
              draw group line={p}{0.4}{$0.4$}{-12ex}{4pt},
              draw group line={p}{0.5}{$0.5$}{-12ex}{4pt},
              draw group line={p}{0.6}{$0.6$}{-12ex}{4pt},
              draw group line={p}{0.7}{$0.7$}{-12ex}{4pt},
              draw group line={p}{0.8}{$0.8$}{-12ex}{4pt},
              draw group line={p}{0.9}{$0.9$}{-12ex}{4pt},
              draw group line={p}{0.99}{$0.99$}{-12ex}{4pt},
              draw group line={p}{1}{$1$}{-12ex}{4pt},
              after end axis/.append code={
                  \path [anchor=base east, yshift=-13.5ex]
                      (rel axis cs:0,0) node  {$p$};
              }
          ]
          
          \addplot [
            fill=green!30!white,
            postaction={pattern=horizontal lines, pattern color=black}
        ] table [x=X, y=Index probing] {\datatable};
      
        \addplot [
            fill=gray!25,
            postaction={pattern=grid, pattern color=black}
        ] table [x=X, y=Position sampling] {\datatable};
    
        \addplot [
            fill=yellow!30,
            postaction={pattern=dots, pattern color=black}
        ] table [x=X, y=Index building] {\datatable};

          \end{axis}
    \end{tikzpicture}
    
    \caption{$|S| = 10^3$, $\deg_y(T) = 10^3$}
    \end{subfigure}

    \bigskip

    \begin{subfigure}{.48\textwidth}
    
        \begin{tikzpicture}
          \pgfplotstableread[col sep=comma]{./data/synthetic-paper-data/s_10000_d_100.csv}\datatable
          
          \begin{axis}[
              axis lines*=left, ymajorgrids,
              width=\linewidth, height=4cm,
              ymin=0,
              ybar stacked,
              bar width=7pt,
              xtick=data,
              xticklabels from table={\datatable}{method},
              xticklabel style={rotate=90,anchor=mid east,font=\scriptsize},
              ylabel={Avg.\ runtime (ms)},    
              legend style={
                  at={(0.5,1)},
                  anchor=south,
                  draw=none,
              },
              legend columns=-1,
              /tikz/every even column/.append style={column sep=1.2em},
              draw group line={p}{0.0001}{$0.0001$}{-12ex}{4pt},
              draw group line={p}{0.001}{$0.001$}{-12ex}{4pt},
              draw group line={p}{0.01}{$0.01$}{-12ex}{4pt},
              draw group line={p}{0.1}{$0.1$}{-12ex}{4pt},
              draw group line={p}{0.2}{$0.2$}{-12ex}{4pt},
              draw group line={p}{0.3}{$0.3$}{-12ex}{4pt},
              draw group line={p}{0.4}{$0.4$}{-12ex}{4pt},
              draw group line={p}{0.5}{$0.5$}{-12ex}{4pt},
              draw group line={p}{0.6}{$0.6$}{-12ex}{4pt},
              draw group line={p}{0.7}{$0.7$}{-12ex}{4pt},
              draw group line={p}{0.8}{$0.8$}{-12ex}{4pt},
              draw group line={p}{0.9}{$0.9$}{-12ex}{4pt},
              draw group line={p}{0.99}{$0.99$}{-12ex}{4pt},
              draw group line={p}{1}{$1$}{-12ex}{4pt},
              after end axis/.append code={
                  \path [anchor=base east, yshift=-13.5ex]
                      (rel axis cs:0,0) node  {$p$};
              }
          ]
          \addplot [
            fill=green!30!white,
            postaction={pattern=horizontal lines, pattern color=black}
        ] table [x=X, y=Index probing] {\datatable};
      
        \addplot [
            fill=gray!25,
            postaction={pattern=grid, pattern color=black}
        ] table [x=X, y=Position sampling] {\datatable};
    
        \addplot [
            fill=yellow!30,
            postaction={pattern=dots, pattern color=black}
        ] table [x=X, y=Index building] {\datatable};

          \end{axis}
    \end{tikzpicture}
    
    \caption{$|S| = 10^4$, $\deg_y(T) = 10^2$}
    \end{subfigure}
    \hfill
    \begin{subfigure}{.48\textwidth}
    
        \begin{tikzpicture}
          \pgfplotstableread[col sep=comma]{./data/synthetic-paper-data/s_100000_d_10.csv}\datatable
          
          \begin{axis}[
              axis lines*=left, ymajorgrids,
              width=\linewidth, height=4cm,
              ymin=0,
              ybar stacked,
              bar width=7pt,
              xtick=data,
              xticklabels from table={\datatable}{method},
              xticklabel style={rotate=90,anchor=mid east,font=\scriptsize},
              ylabel={Avg.\ runtime (ms)},    
              legend style={
                  at={(0.5,1)},
                  anchor=south,
                  draw=none,
              },
              legend columns=-1,
              /tikz/every even column/.append style={column sep=1.2em},
              draw group line={p}{0.0001}{$0.0001$}{-12ex}{4pt},
              draw group line={p}{0.001}{$0.001$}{-12ex}{4pt},
              draw group line={p}{0.01}{$0.01$}{-12ex}{4pt},
              draw group line={p}{0.1}{$0.1$}{-12ex}{4pt},
              draw group line={p}{0.2}{$0.2$}{-12ex}{4pt},
              draw group line={p}{0.3}{$0.3$}{-12ex}{4pt},
              draw group line={p}{0.4}{$0.4$}{-12ex}{4pt},
              draw group line={p}{0.5}{$0.5$}{-12ex}{4pt},
              draw group line={p}{0.6}{$0.6$}{-12ex}{4pt},
              draw group line={p}{0.7}{$0.7$}{-12ex}{4pt},
              draw group line={p}{0.8}{$0.8$}{-12ex}{4pt},
              draw group line={p}{0.9}{$0.9$}{-12ex}{4pt},
              draw group line={p}{0.99}{$0.99$}{-12ex}{4pt},
              draw group line={p}{1}{$1$}{-12ex}{4pt},
              after end axis/.append code={
                  \path [anchor=base east, yshift=-13.5ex]
                      (rel axis cs:0,0) node  {$p$};
              }
          ]
          \addplot [
            fill=green!30!white,
            postaction={pattern=horizontal lines, pattern color=black}
        ] table [x=X, y=Index probing] {\datatable};
      
        \addplot [
            fill=gray!25,
            postaction={pattern=grid, pattern color=black}
        ] table [x=X, y=Position sampling] {\datatable};
    
        \addplot [
            fill=yellow!30,
            postaction={pattern=dots, pattern color=black}
        ] table [x=X, y=Index building] {\datatable};
          
          \end{axis}
    \end{tikzpicture}
    
        \caption{$|S| = 10^5$, $\deg_y(T) = 10^1$}
    \end{subfigure}
    
          \caption{$O = |T| = 10^6$}
          \label{fig:synth-1000k}
        \end{figure*}

\begin{figure*}
    \centering
    \ref{namedlegend}
\begin{subfigure}{.48\textwidth}

    \begin{tikzpicture}
        \pgfplotstableread[col sep=comma]{./data/synthetic-paper-data/s_10_d_1000000.csv}\datatable
        
        \begin{axis}[
            axis lines*=left, ymajorgrids,
            width=\linewidth, height=4cm,
            ymin=0,
            ybar stacked,
            bar width=7pt,
            xtick=data,
            xticklabels from table={\datatable}{method},
            xticklabel style={rotate=90,anchor=mid east,font=\scriptsize},
            ylabel={Avg.\ runtime (ms)},    
            legend to name=namedlegend,
            legend columns=-1,
            legend style={draw=none, /tikz/every even column/.append style={column sep=1.2em}},
            /tikz/every even column/.append style={column sep=1.2em},
            draw group line={p}{0.0001}{$0.0001$}{-12ex}{4pt},
            draw group line={p}{0.001}{$0.001$}{-12ex}{4pt},
            draw group line={p}{0.01}{$0.01$}{-12ex}{4pt},
            draw group line={p}{0.1}{$0.1$}{-12ex}{4pt},
            draw group line={p}{0.2}{$0.2$}{-12ex}{4pt},
            draw group line={p}{0.3}{$0.3$}{-12ex}{4pt},
            draw group line={p}{0.4}{$0.4$}{-12ex}{4pt},
            draw group line={p}{0.5}{$0.5$}{-12ex}{4pt},
            draw group line={p}{0.6}{$0.6$}{-12ex}{4pt},
            draw group line={p}{0.7}{$0.7$}{-12ex}{4pt},
            draw group line={p}{0.8}{$0.8$}{-12ex}{4pt},
            draw group line={p}{0.9}{$0.9$}{-12ex}{4pt},
            draw group line={p}{0.99}{$0.99$}{-12ex}{4pt},
            draw group line={p}{1}{$1$}{-12ex}{4pt},
            after end axis/.append code={
                \path [anchor=base east, yshift=-13.5ex]
                    (rel axis cs:0,0) node  {$p$};
            }
        ]
        \addplot [
            fill=green!30!white,
            postaction={pattern=horizontal lines, pattern color=black}
        ] table [x=X, y=Index probing] {\datatable};
        
        \addplot [
            fill=gray!25,
            postaction={pattern=grid, pattern color=black}
        ] table [x=X, y=Position sampling] {\datatable};
    
        \addplot [
            fill=yellow!30,
            postaction={pattern=dots, pattern color=black}
        ] table [x=X, y=Index building] {\datatable};

            \addlegendentry{Index probing}
            \addlegendentry{Position sampling}
            \addlegendentry{Index building}
        
        \end{axis}
\end{tikzpicture}

    \caption{$|S| = 10^1$, $\deg_y(T) = 10^6$}
\end{subfigure}
\hfill    
\begin{subfigure}{.48\textwidth}

    \begin{tikzpicture}
        \pgfplotstableread[col sep=comma]{./data/synthetic-paper-data/s_100_d_100000.csv}\datatable
        
        \begin{axis}[
            axis lines*=left, ymajorgrids,
            width=\linewidth, height=4cm,
            ymin=0,
            ybar stacked,
            bar width=7pt,
            xtick=data,
            xticklabels from table={\datatable}{method},
            xticklabel style={rotate=90,anchor=mid east,font=\scriptsize},
            ylabel={Avg.\ runtime (ms)},    
            legend style={
                at={(0.5,1)},
                anchor=south,
                draw=none,
            },
            legend columns=-1,
            /tikz/every even column/.append style={column sep=1.2em},
            draw group line={p}{0.0001}{$0.0001$}{-12ex}{4pt},
            draw group line={p}{0.001}{$0.001$}{-12ex}{4pt},
            draw group line={p}{0.01}{$0.01$}{-12ex}{4pt},
            draw group line={p}{0.1}{$0.1$}{-12ex}{4pt},
            draw group line={p}{0.2}{$0.2$}{-12ex}{4pt},
            draw group line={p}{0.3}{$0.3$}{-12ex}{4pt},
            draw group line={p}{0.4}{$0.4$}{-12ex}{4pt},
            draw group line={p}{0.5}{$0.5$}{-12ex}{4pt},
            draw group line={p}{0.6}{$0.6$}{-12ex}{4pt},
            draw group line={p}{0.7}{$0.7$}{-12ex}{4pt},
            draw group line={p}{0.8}{$0.8$}{-12ex}{4pt},
            draw group line={p}{0.9}{$0.9$}{-12ex}{4pt},
            draw group line={p}{0.99}{$0.99$}{-12ex}{4pt},
            draw group line={p}{1}{$1$}{-12ex}{4pt},
            after end axis/.append code={
                \path [anchor=base east, yshift=-13.5ex]
                    (rel axis cs:0,0) node  {$p$};
            }
        ]
        
        \addplot [
        fill=green!30!white,
        postaction={pattern=horizontal lines, pattern color=black}
    ] table [x=X, y=Index probing] {\datatable};
    
    \addplot [
        fill=gray!25,
        postaction={pattern=grid, pattern color=black}
    ] table [x=X, y=Position sampling] {\datatable};

    \addplot [
        fill=yellow!30,
        postaction={pattern=dots, pattern color=black}
    ] table [x=X, y=Index building] {\datatable};

        \end{axis}
\end{tikzpicture}

\caption{$|S| = 10^2$, $\deg_y(T) = 10^5$}
\end{subfigure}

\bigskip  

\begin{subfigure}{.48\textwidth}

    \begin{tikzpicture}
        \pgfplotstableread[col sep=comma]{./data/synthetic-paper-data/s_1000_d_10000.csv}\datatable
        
        \begin{axis}[
            axis lines*=left, ymajorgrids,
            width=\linewidth, height=4cm,
            ymin=0,
            ybar stacked,
            bar width=7pt,
            xtick=data,
            xticklabels from table={\datatable}{method},
            xticklabel style={rotate=90,anchor=mid east,font=\scriptsize},
            ylabel={Avg.\ runtime (ms)},    
            legend style={
                at={(0.5,1)},
                anchor=south,
                draw=none,
            },
            legend columns=-1,
            /tikz/every even column/.append style={column sep=1.2em},
            draw group line={p}{0.0001}{$0.0001$}{-12ex}{4pt},
            draw group line={p}{0.001}{$0.001$}{-12ex}{4pt},
            draw group line={p}{0.01}{$0.01$}{-12ex}{4pt},
            draw group line={p}{0.1}{$0.1$}{-12ex}{4pt},
            draw group line={p}{0.2}{$0.2$}{-12ex}{4pt},
            draw group line={p}{0.3}{$0.3$}{-12ex}{4pt},
            draw group line={p}{0.4}{$0.4$}{-12ex}{4pt},
            draw group line={p}{0.5}{$0.5$}{-12ex}{4pt},
            draw group line={p}{0.6}{$0.6$}{-12ex}{4pt},
            draw group line={p}{0.7}{$0.7$}{-12ex}{4pt},
            draw group line={p}{0.8}{$0.8$}{-12ex}{4pt},
            draw group line={p}{0.9}{$0.9$}{-12ex}{4pt},
            draw group line={p}{0.99}{$0.99$}{-12ex}{4pt},
            draw group line={p}{1}{$1$}{-12ex}{4pt},
            after end axis/.append code={
                \path [anchor=base east, yshift=-13.5ex]
                    (rel axis cs:0,0) node  {$p$};
            }
        ]
        
        \addplot [
        fill=green!30!white,
        postaction={pattern=horizontal lines, pattern color=black}
    ] table [x=X, y=Index probing] {\datatable};
    
    \addplot [
        fill=gray!25,
        postaction={pattern=grid, pattern color=black}
    ] table [x=X, y=Position sampling] {\datatable};

    \addplot [
        fill=yellow!30,
        postaction={pattern=dots, pattern color=black}
    ] table [x=X, y=Index building] {\datatable};

        \end{axis}
\end{tikzpicture}

\caption{$|S| = 10^3$, $\deg_y(T) = 10^4$}
\end{subfigure}
\hfill
\begin{subfigure}{.48\textwidth}

    \begin{tikzpicture}
        \pgfplotstableread[col sep=comma]{./data/synthetic-paper-data/s_10000_d_1000.csv}\datatable
        
        \begin{axis}[
            axis lines*=left, ymajorgrids,
            width=\linewidth, height=4cm,
            ymin=0,
            ybar stacked,
            bar width=7pt,
            xtick=data,
            xticklabels from table={\datatable}{method},
            xticklabel style={rotate=90,anchor=mid east,font=\scriptsize},
            ylabel={Avg.\ runtime (ms)},    
            legend style={
                at={(0.5,1)},
                anchor=south,
                draw=none,
            },
            legend columns=-1,
            /tikz/every even column/.append style={column sep=1.2em},
            draw group line={p}{0.0001}{$0.0001$}{-12ex}{4pt},
            draw group line={p}{0.001}{$0.001$}{-12ex}{4pt},
            draw group line={p}{0.01}{$0.01$}{-12ex}{4pt},
            draw group line={p}{0.1}{$0.1$}{-12ex}{4pt},
            draw group line={p}{0.2}{$0.2$}{-12ex}{4pt},
            draw group line={p}{0.3}{$0.3$}{-12ex}{4pt},
            draw group line={p}{0.4}{$0.4$}{-12ex}{4pt},
            draw group line={p}{0.5}{$0.5$}{-12ex}{4pt},
            draw group line={p}{0.6}{$0.6$}{-12ex}{4pt},
            draw group line={p}{0.7}{$0.7$}{-12ex}{4pt},
            draw group line={p}{0.8}{$0.8$}{-12ex}{4pt},
            draw group line={p}{0.9}{$0.9$}{-12ex}{4pt},
            draw group line={p}{0.99}{$0.99$}{-12ex}{4pt},
            draw group line={p}{1}{$1$}{-12ex}{4pt},
            after end axis/.append code={
                \path [anchor=base east, yshift=-13.5ex]
                    (rel axis cs:0,0) node  {$p$};
            }
        ]
        \addplot [
        fill=green!30!white,
        postaction={pattern=horizontal lines, pattern color=black}
    ] table [x=X, y=Index probing] {\datatable};
    
    \addplot [
        fill=gray!25,
        postaction={pattern=grid, pattern color=black}
    ] table [x=X, y=Position sampling] {\datatable};

    \addplot [
        fill=yellow!30,
        postaction={pattern=dots, pattern color=black}
    ] table [x=X, y=Index building] {\datatable};

        \end{axis}
\end{tikzpicture}

\caption{$|S| = 10^4$, $\deg_y(T) = 10^3$}
\end{subfigure}

\bigskip

\begin{subfigure}{.48\textwidth}

    \begin{tikzpicture}
        \pgfplotstableread[col sep=comma]{./data/synthetic-paper-data/s_100000_d_100.csv}\datatable
        
        \begin{axis}[
            axis lines*=left, ymajorgrids,
            width=\linewidth, height=4cm,
            ymin=0,
            ybar stacked,
            bar width=7pt,
            xtick=data,
            xticklabels from table={\datatable}{method},
            xticklabel style={rotate=90,anchor=mid east,font=\scriptsize},
            ylabel={Avg.\ runtime (ms)},    
            legend style={
                at={(0.5,1)},
                anchor=south,
                draw=none,
            },
            legend columns=-1,
            /tikz/every even column/.append style={column sep=1.2em},
            draw group line={p}{0.0001}{$0.0001$}{-12ex}{4pt},
            draw group line={p}{0.001}{$0.001$}{-12ex}{4pt},
            draw group line={p}{0.01}{$0.01$}{-12ex}{4pt},
            draw group line={p}{0.1}{$0.1$}{-12ex}{4pt},
            draw group line={p}{0.2}{$0.2$}{-12ex}{4pt},
            draw group line={p}{0.3}{$0.3$}{-12ex}{4pt},
            draw group line={p}{0.4}{$0.4$}{-12ex}{4pt},
            draw group line={p}{0.5}{$0.5$}{-12ex}{4pt},
            draw group line={p}{0.6}{$0.6$}{-12ex}{4pt},
            draw group line={p}{0.7}{$0.7$}{-12ex}{4pt},
            draw group line={p}{0.8}{$0.8$}{-12ex}{4pt},
            draw group line={p}{0.9}{$0.9$}{-12ex}{4pt},
            draw group line={p}{0.99}{$0.99$}{-12ex}{4pt},
            draw group line={p}{1}{$1$}{-12ex}{4pt},
            after end axis/.append code={
                \path [anchor=base east, yshift=-13.5ex]
                    (rel axis cs:0,0) node  {$p$};
            }
        ]
        \addplot [
        fill=green!30!white,
        postaction={pattern=horizontal lines, pattern color=black}
    ] table [x=X, y=Index probing] {\datatable};
    
    \addplot [
        fill=gray!25,
        postaction={pattern=grid, pattern color=black}
    ] table [x=X, y=Position sampling] {\datatable};

    \addplot [
        fill=yellow!30,
        postaction={pattern=dots, pattern color=black}
    ] table [x=X, y=Index building] {\datatable};

        \end{axis}
\end{tikzpicture}

\caption{$|S| = 10^5$, $\deg_y(T) = 10^2$}
\end{subfigure}
\hfill
\begin{subfigure}{.48\textwidth}

    \begin{tikzpicture}
        \pgfplotstableread[col sep=comma]{./data/synthetic-paper-data/s_1000000_d_10.csv}\datatable
        
        \begin{axis}[
            axis lines*=left, ymajorgrids,
            width=\linewidth, height=4cm,
            ymin=0,
            ybar stacked,
            bar width=7pt,
            xtick=data,
            xticklabels from table={\datatable}{method},
            xticklabel style={rotate=90,anchor=mid east,font=\scriptsize},
            ylabel={Avg.\ runtime (ms)},    
            legend style={
                at={(0.5,1)},
                anchor=south,
                draw=none,
            },
            legend columns=-1,
            /tikz/every even column/.append style={column sep=1.2em},
            draw group line={p}{0.0001}{$0.0001$}{-12ex}{4pt},
            draw group line={p}{0.001}{$0.001$}{-12ex}{4pt},
            draw group line={p}{0.01}{$0.01$}{-12ex}{4pt},
            draw group line={p}{0.1}{$0.1$}{-12ex}{4pt},
            draw group line={p}{0.2}{$0.2$}{-12ex}{4pt},
            draw group line={p}{0.3}{$0.3$}{-12ex}{4pt},
            draw group line={p}{0.4}{$0.4$}{-12ex}{4pt},
            draw group line={p}{0.5}{$0.5$}{-12ex}{4pt},
            draw group line={p}{0.6}{$0.6$}{-12ex}{4pt},
            draw group line={p}{0.7}{$0.7$}{-12ex}{4pt},
            draw group line={p}{0.8}{$0.8$}{-12ex}{4pt},
            draw group line={p}{0.9}{$0.9$}{-12ex}{4pt},
            draw group line={p}{0.99}{$0.99$}{-12ex}{4pt},
            draw group line={p}{1}{$1$}{-12ex}{4pt},
            after end axis/.append code={
                \path [anchor=base east, yshift=-13.5ex]
                    (rel axis cs:0,0) node  {$p$};
            }
        ]
        \addplot [
        fill=green!30!white,
        postaction={pattern=horizontal lines, pattern color=black}
    ] table [x=X, y=Index probing] {\datatable};
    
    \addplot [
        fill=gray!25,
        postaction={pattern=grid, pattern color=black}
    ] table [x=X, y=Position sampling] {\datatable};

    \addplot [
        fill=yellow!30,
        postaction={pattern=dots, pattern color=black}
    ] table [x=X, y=Index building] {\datatable};

        \end{axis}
\end{tikzpicture}

\caption{$|S| = 10^6$, $\deg_y(T) = 10^1$}
\end{subfigure}

    \caption{$O = |T| = 10^7$}
    \label{fig:synth-10000k}
\end{figure*}

 }

\end{document}